\documentclass[11pt,reqno]{amsart}

\usepackage{xcolor,mathtools,bbm}

 \usepackage{framed}

\usepackage{geometry}                
\usepackage{amsmath,amssymb,amsthm}

\usepackage{amsmath,amssymb,amsthm}
\usepackage[authoryear]{natbib}
\usepackage{graphicx}
\usepackage{comment}

\theoremstyle{plain}
\newtheorem{theorem}{Theorem}

\newtheorem{lemma}{Lemma}
\newtheorem{corollary}{Corollary}

 \newtheorem{assumption}{Assumption}
\newtheorem{algorithm}{Algorithm}
\theoremstyle{definition}

\newtheorem{remark}{Remark}

\DeclareMathOperator{\diag}{diag}

\newcommand{\mT}{\mathcal{T}}
\newcommand{\mK}{\mathcal{K}}

\newcommand{\mB}{\mathcal{B}}
\newcommand{\mD}{\mathcal{D}}
\renewcommand{\qed}{\hfill{\tiny \ensuremath{\blacksquare} }}%

\newcommand{\lto}{\leftarrow}
\newcommand{\Ep}{{\mathbb{E}}}

\renewcommand{\Pr}{{\mathrm{P}}}

\newcommand{\mC}{\mathcal{C}}

\newcommand{\mA}{\mathcal{A}}
\newcommand{\mY}{\mathcal{Y}}
\newcommand{\mX}{\mathcal{X}}

\newcommand{\plim}{\operatorname*{plim}}
\newcommand{\argmax}{\operatorname*{argmax}}

\usepackage{tikz}
\usetikzlibrary {positioning}

\usepackage{bm}
\usepackage{epstopdf}

\begin{document}

\title[Network and Panel Distribution Regression]{ Network and Panel Quantile Effects \\ Via Distribution Regression}
\thanks{Initial Discussion: 7/8/2015. We thank the editor Xiaohong Chen, two anonymous referees, Manuel Arellano, Joao Santos Silva and seminar participants at Bonn, Boston College, Bristol, Cemfi, Chicago, 2017 EEA-ESEM Conference, Humboldt,  Northwestern, Princeton, Surrey, UC-Davis, UC-Irvine,  USC, UConn New Frontiers in Econometrics Conference, York,  UvA-Econometrics Panel Data Workshop, and Berkeley/CeMMAP Conference on Networks for comments. Gianluca Russo and Siyi Luo provided capable research assistance.  
Financial support from the National Science Foundation, 
 Economic and Social Research Council through the ESRC Centre for Microdata Methods and Practice grant RES-589-28-0001,
and European Research Council grants ERC-2014-CoG-646917-ROMIA and ERC-2018-CoG-819086-PANEDA
 is gratefully   acknowledged.}
\author{Victor Chernozhukov \and Ivan Fernandez-Val  \and Martin Weidner }
\date{\today}
\begin{abstract}  This paper provides a method to construct simultaneous confidence bands for quantile functions and quantile effects in nonlinear network and panel models with unobserved two-way effects, strictly exogenous covariates, and possibly discrete outcome variables. The method is based upon projection of simultaneous confidence bands for distribution functions constructed from fixed effects distribution regression estimators. These fixed effects estimators are debiased to deal with the incidental parameter problem. Under asymptotic sequences where both dimensions of the data set grow at the same rate, the confidence bands for the quantile functions and effects have correct joint coverage in large samples. An empirical application to gravity models of trade illustrates the applicability of the methods to network data.
\end{abstract}

\maketitle 


%
%
%
%
%
%

\section{Introduction}

Standard regression analyzes average    effects of covariates  on outcome variables. In many applications it is equally important to consider distributional effects. For example, a policy maker might be interested in the effect of an education reform not only on the mean but also the entire distribution of test scores or wages. Availability of panel data is very useful to identify \textit{ceteris paribus} average and distributional effects because it allows the researcher to control for multiple sources of unobserved heterogeneity that might cause endogeneity or omitted variable problems. The idea is to use variation of the covariates over time for each individual or over individuals for each time period to account for unobserved individual and time effects. In this paper we develop inference methods for distributional effects in nonlinear models with  two-way unobserved effects. They apply not only to traditional panel data models where the unobserved effects correspond to individual and time fixed effects, but also to models for other types of data where the unobserved effects reflect some grouping structure such as unobserved sender and receiver effects in network data models.
 The unobserved effects will be treated as fixed effects, i.e. parameters to be estimated, leaving their relation to observed covariates unrestricted.

We develop inference methods for quantile functions and effects. The quantile function corresponds to the marginal distribution of the outcome in  a counterfactual scenario where the treatment covariate of interest is set exogenously at a desired level and the rest of the covariates and unobserved effects are held fixed, extending the construction of Chernozhukov Fernandez-Val and Melly~\citeyearpar{ChernozhukovFernandezValMelly2013} for the cross section case. The quantile effect is the difference of quantile functions at two different treatment levels. Our methods apply to continuous and discrete treatments by appropriate choice of the treatment levels, and have causal interpretation under standard unconfoundedness assumptions for panel data.
The inference is based upon the generic method of Chernozhukov, Fernandez-Val, Melly, and Wuthrich \citeyearpar{CFVMW2015} that projects joint confidence bands for distributions into joint confidence bands for quantile functions and effects. This method has the appealing feature that applies without modification to any type of outcome, let it be continuous, discrete or mixed.  

The key input for the inference method is a joint confidence band for the counterfactual distributions at the treatment levels of interest. We construct this band from fixed effects distribution regression (FE-DR) estimators of the conditional distribution of the outcome given the observed covariates and unobserved effects. In doing so, we extend the distribution regression approach to model conditional distributions  with unobserved effects. This version of the DR model is semiparametric because not only  the DR coefficients can vary with the level of the outcome as in the cross section case, but also the distribution of the unobserved effects is left unspecified. We show that the FE-DR estimator can be obtained as a sequence of binary response fixed effects estimators where the binary response is an indicator of the outcome passing some threshold.    To deal with the incidental parameter problem associated with the estimation of the unobserved effects (Neyman and Scott~\citeyearpar{NeymanScott1948}), we extend the analytical bias corrections of \cite{FernandezValWeidner2016} for single binary response estimators to multiple  (possibly a continuum) of binary response estimators. In particular,  the main technical contribution is to establish  functional central limit theorems for the fixed effects estimators of the DR coefficients and associated counterfactual distributions, and show the validity of the bias corrections under asymptotic sequences where the two dimensions of the data set pass to infinity at the same rate.  As in the single binary response model, the bias corrections remove the asymptotic bias of the fixed effects estimators without increasing their asymptotic variances.  

We implement the inference method using multiplier bootstrap \citep{gz-84}. This version of bootstrap constructs draws of an estimator as weighted averages of its influence function, where the weights are independent from the data. Compared to empirical bootstrap, multiplier bootstrap has the computational advantage that it does not involve any parameter reestimation. This advantage is particularly convenient in our setting because the parameter estimation require multiple nonlinear optimizations that can be highly dimensional due to the fixed effects. Multiplier bootstrap is also convenient to account for data dependencies. In network data, for example, it might be important to account for reciprocity or pairwise clustering. Reciprocity arises because observational units corresponding to the same pair of agents but reversing their roles as sender and receiver might be dependent even after conditioning on the unobserved effects. By setting the weights of these observational units equal, we account for this dependence  in the multiplier bootstrap. In addition to the previous practical reasons, there are some theoretical reasons for choosing multiplier bootstrap. Thus, \cite{cck-16} established bootstrap functional central limit theorems for multiplier bootstrap in high dimensional settings that cover  the network and panel models that we consider.

The methods developed in this paper apply to models that include unobserved effects to capture  grouping or clustering structures in the data such as models for panel and network data. These effects allow us to control for unobserved group heterogeneity that might be related to the covariates causing endogeneity or omitted variable bias. They also serve to parsimoniously account for dependencies in the data. We illustrate the wide applicability with an empirical example to gravity models of trade.  In this case the outcome is the volume of trade between two countries and each observational unit corresponds to a country pair indexed by  exporter country (sender) and importer country (receiver). We estimate the distributional effects of gravity variables such as the geographical distance controlling for exporter and importer country effects that pick up  unobserved heterogeneity possibly correlated with the gravity variables. We uncover significant heterogeneity in the effects of distance and other gravity variables across the distribution, which is missed by traditional mean methods. We also find that the Poisson model, which is commonly used in the trade literature to deal with zero trade in many country pairs,  does not provide a good approximation to the distribution of the volume of trade due to heavy tails.


%


\paragraph{\textbf{Literature review}} Unlike mean effects, there are different ways to define distributional and quantile effects. For example, we can distinguish conditional effects versus unconditional or marginalized effects, or  quantile effects versus quantiles of the effects.  Here we give a brief review of the recent literature on distributional and quantile effects in panel data models emphasizing  the following aspects: (1) type of effect considered;  (2) type of unobserved effects in the model; and (3) asymptotic approximation. For the unobserved effects, we distinguish models with one-way effects versus two-way effects. For the asymptotic approximation we distinguish  short panels with large $N$ and fixed $T$ versus long panels with large $N$ and large $T$ , where $N$ and $T$ denote the dimensions of the panel.  We focus mainly on fixed effects approaches where the unobserved effects are treated as parameters to be estimated, but also mention some correlated random effects approaches that impose restrictions on the distribution of the unobserved effects. This paper deals with inference on marginalized quantile effects in large panels with two-way effects, which has not been previously considered in the literature. Indeed, to the best of our knowledge, it is the first paper to provide inference methods for quantile treatment effects from panel and network models with two-way fixed effects. 

Koenker~\citeyearpar{Koenker2004} introduced fixed effects quantile regression estimators of conditional quantile effects in large panel models with one-way individual effects using shrinkage to control the variability in the estimation of the unobserved effects. \cite{Lamarche2010} discussed the optimal choice of a tuning parameter in Koenker's method. In the same framework, Kato, Galvao, and Montes-Rojas~\citeyearpar{KatoGalvaoMontesRojas2012},
Galvao, Lamarche and Lima~\citeyearpar{GalvaoLamarcheLima2013},
Galvao and Kato~\citeyearpar{GalvaoKato2016} and Arellano and Weidner~\citeyearpar{ArellanoWeidner2016} considered fixed effects quantile regression estimators without shrinkage and developed bias corrections.  All these papers require that $T$ pass to infinity faster than $N$, making it difficult to extend the theory to models with two-way individual and time effects.  Graham, Hahn and Powell~\citeyearpar{GrahamHahnPowell2009} found a special case where the fixed effects quantile regression estimator does not suffer of  incidental parameter problem. Machado and Santos Silva~\citeyearpar{mss18} has recently proposed a method to estimate conditional quantile effects in a location-scale model via moments. 

In short panels, Rosen~\citeyearpar{Rosen2012} showed that a linear quantile restriction is not sufficient to point identify conditional effects in 
a panel linear  quantile regression model with unobserved individual effects.   Chernozhukov, Fernandez-Val, Hahn and Newey \citeyearpar{ChernozhukovFernandezValHahnNewey2013} and \cite{chernozhukov2015nonparametric} discussed identification and estimation of marginalized quantile effects in nonseparable panel models with unobserved individual effects and location and scale time effects under a time homogeneity assumption. They showed that the effects are point identified only for some subpopulations and characterized these subpopulations. Graham, Hahn, Poirier and Powell~\citeyearpar{GrahamHahnPoirierPowell2015} considered quantiles of effects  in linear quantile regression models with two-way effects.  Finally, Abrevaya and Dahl~\citeyearpar{AbrevayaDahl2008} and Arellano and Bonhomme~\citeyearpar{ArellanoBonhomme2016} developed estimators for conditional quantile effects in linear quantile regression model with unobserved individual effects using correlated random effects approaches. None of the previous quantile regression based methods apply to discrete outcomes.

Finally, we review previous applications of panel data methods to network data. These include \cite{candelaria16}, \cite{charbonneau17}, \cite{CFW-16}, \cite{Dzemski2017},  \cite{FernandezValWeidner2016}, \cite{gao20}, \cite{Graham2016, Graham2017},  \cite{jochmans18}, \cite{toth17}, and  \cite{Yan2016statistical}, which developed methods for models of network formation with unobserved sender and receiver effects for directed and undirected networks.\footnote{We refer to \cite{paula19} for an excellent up to date review on this topic.} None of these papers consider estimation of quantile effects as the outcome variable is binary, whether or not a link is formed between two agents.

\paragraph{\textbf{Plan of the paper}} Section \ref{sec:model} introduces the distribution regression model with unobserved effects for network and panel data, and describes the quantities of interest including model parameters, distributions, quantiles and quantile effects. Section \ref{sec:estimation} discusses fixed effects estimation, bias corrections to deal with the incidental parameter problem, and uniform inference methods. Section \ref{sec:asymptotics} provides asymptotic theory for the fixed effects estimators, bias corrections, and multiplier bootstrap. Section \ref{sec:trade} and \ref{sec:mc} report results of the empirical application to the gravity models of trade and a Monte Carlo simulation calibrated to the application, respectively. The proofs of the main results are given in the Appendix, and additional technical results are provided in the Supplementary Appendix. 

\paragraph{\textbf{Notation}} 
For any two real
numbers $a$ and $b$, $a\vee b = \max\{a,b\}$ and $a\wedge b =
\min\{a,b\}$.  For a real number $a$, $\lfloor a \rfloor$ denotes the integer part of $a$. For a set $\mA$, $| \mA|$ denotes the cardinality or number of elements of $\mA$.
 
\section{Model and Parameters of Interest}\label{sec:model}

\subsection{Distribution Regression Model with Unobserved Effects} We observe the data set $\{(y_{ij}, x_{ij}) : (i,j) \in \mD \}$, where $y_{ij}$ is a scalar outcome variable with region of interest $\mY$, and $x_{ij}$ is a vector of covariates with support $\mX \subseteq \mathbb{R}^{d_x}$.\footnote{If $y_{ij}$ has unbounded support, then the region $\mY$ is usually a subset of the support to avoid tail estimation.}  The variable $y_{ij}$ can be discrete,  continuous or mixed. The subscripts $i$ and $j$ index individuals and  time periods in traditional panels, but they might index other dimensions in more general data structures. In our empirical application, for example, we use a panel where $y_{ij}$ is the volume of trade between country $i$ and country $j$, and $x_{ij}$ includes gravity variables such as the distance between country $i$ and country $j$.  Both $i$ and $j$  index countries as exporters and importers respectively. The set  $\mD$ contains the indexes of the pairs $(i,j)$ that are observed.  It is a subset of the set of all possible pairs $\mD_0  := \{(i,j) :  i = 1,\dots,I; j = 1, \dots, J \}$, where  $I$ and $J$ are the dimensions of the panel. 
We introduce $\mD$ to  allow for certain forms of missing data that are common in panel and network applications, see Assumption \ref{ass:baseline}(v) in Section \ref{sec:asymptotics}. For example, in the trade application  $I=J$ and $\mD = \mD_0 \setminus  \{(i,i) :  i = 1,\dots,I \}$ because we do not observe trade of a country with itself. We denote the total number of observed units by $n$, i.e. $n = |\mD|$. 

Let $v_i$ and $w_j$ denote vectors of unspecified dimension that contain unobserved random variables or effects that might be related to the covariates $x_{ij}$. In traditional panels, $v_i$ are  individual effects that capture unobserved individual heterogeneity and $w_j$ are time effects that account for aggregate shocks. More generally, these variables serve  to capture some forms of endogeneity and group dependencies in a parsimonious fashion. We specify the conditional distribution of $y_{ij}$ given $(x_{ij},v_i,w_{j})$ using the \textit{distribution regression (DR) model with unobserved effects} 
\begin{equation}\label{eq:pdr}
F_{y_{ij} }(y \mid x_{ij}, v_i, w_{j}) = \Lambda_y(P(x_{ij})'\beta(y) +  \alpha(v_i,y) + \gamma(w_{j},y)), \ \ y \in \mY,  \ \ (i,j) \in \mD,
\end{equation}
where $\Lambda_y$ is a known link function such as the normal or logistic distribution, which may vary with $y$, $x \mapsto P(x)$ is a dictionary of transformations of $x$ such us polynomials, b-splines and tensor products, $\beta(y)$ is an unknown parameter vector, which can vary with $y$, and $(v,y) \mapsto \alpha(v,y)$ and $(w,y) \mapsto \gamma(w,y)$ are unspecified measurable functions.  
This DR model is a semiparametric model for the conditional distribution because $y \mapsto  \theta(y) := (\beta(y),\alpha(v_1,y), \ldots, \alpha(v_I,y), \gamma(w_1,y), \ldots,  \gamma(w_J,y))$ is a function-valued parameter and the dimension of $\theta(y)$ varies with $I$ and $J$, although we do not make this dependence explicit. We shall treat the dimension of $P(x)$ as fixed  and set $\Lambda_y$ equal to the logistic distribution for all $y$ in the asymptotic analysis.

When $y_{ij}$ is continuous, the model \eqref{eq:pdr} has the following representation as an implicit nonseparable model by the probability integral transform
$$
\Lambda_{y_{ij}}(P(x_{ij})'\beta(y_{ij}) + \alpha(v_i,y_{ij}) + \gamma(w_{j}, y_{ij})) = u_{ij}, \ \ u_{ij} \mid x_{ij}, v_i, w_{j} \sim U(0,1),
$$
where the error $u_{ij}$ represents the unobserved ranking of the observation $y_{ij}$ in the conditional distribution. 
The parameters of the model are related to derivatives of the conditional quantiles. Let $Q_{y_{ij}}(u \mid x_{ij}, v_{i}, w_{j})$ be the $u$-quantile of $y_{ij}$ conditional on $(x_{ij}, v_{i}, w_{j})$ defined as the left-inverse of $y \mapsto  F_{y_{ij} }(y \mid x_{ij}, v_i, w_{j})$ at $u$, namely
$$
Q_{y_{ij} }(u \mid x_{ij}, v_{i}, w_{j}) = \inf\{y \in \mY : F_{y_{ij} }(y \mid x_{ij}, v_i, w_{j}) \geq u\}  \wedge \sup \{ y \in {\mY}\},
$$
 and $x_{ij} = (x_{ij}^1,\ldots, x_{ij}^{d_x})$.\footnote{We use the convention $\inf \{\emptyset\} = + \infty$.} 
 Then, it can be shown that if $y \mapsto F_{y_{ij} }(y \mid x_{ij}, v_i, w_{j})$ is strictly increasing in the support of $y_{ij}$, $\partial \Lambda_y(z)/\partial z > 0$ for all $y$ in the support of  $y_{ij}$ and $x_{ij} \mapsto Q_{y_{ij} }(u \mid x_{ij}, v_{i}, w_{j})$ is differentiable,\footnote{Indeed, $\Lambda_y(P(x_{ij})'\beta(y) + \alpha(v_i,y) + \gamma(w_{j}, y)) =u$ at $y = Q_{y_{ij}}(u \mid x_{ij}, v_{i}, w_{j})$. Differencing this expression with respect to $x_{ij}^k$ yields $$\left.  \partial_{x_{ij}^{k}} P(x_{ij})'\beta(y) \right|_{y = Q_{y_{ij}}(u \mid x_{ij}, v_{i}, w_{j})} = - \left. \frac{\partial \Lambda_y\left(P(x_{ij})'\beta(y) + \alpha(v_i,y) + \gamma(w_{j}, y) \right)/\partial y}{\lambda_y\left(P(x_{ij})'\beta(y) + \alpha(v_i,y) + \gamma(w_{j}, y)\right)}  \right|_{y = Q_{y_{ij}}(u \mid x_{ij}, v_{i}, w_{j})}  \partial_{x_{ij}^{k}} Q_{y_{ij} }(u \mid x_{ij}, v_{i}, w_{j}),$$
where $\lambda_y(z) = \partial \Lambda_y(z)/\partial z$. Note that the first term of the right hand side does not depend on $k$ and is positive because $y \mapsto F_{y_{ij} }(y \mid x_{ij}, v_i, w_{j}) = \Lambda_y(P(x_{ij})'\beta(y) + \alpha(v_i,y) + \gamma(w_{j}, y))$ is strictly increasing at $y = Q_{y_{ij}}(u \mid x_{ij}, v_{i}, w_{j})$.}
$$
\left. \partial_{x_{ij}^{k}} P(x_{ij})'\beta(y) \right|_{y=Q_{y_{ij} }(u \mid x_{ij}, v_{i}, w_{j})}  \propto - \partial_{x_{ij}^{k}} Q_{y_{ij} }(u \mid x_{ij}, v_{i}, w_{j}), \ \ k = 1,\ldots,d_x, \ \ \partial_{x_{ij}^{k}} := \partial/ \partial x_{ij}^k.
$$
If $P(x_{ij}) = x_{ij}$, then $\partial_{x_{ij}^{k}} P(x_{ij})'\beta(y) = \beta_{k}(y)$  such that
$$
\left. \frac{\beta_{\ell}(y)}{\beta_k(y)} \right|_{y = Q_{y_{ij}}(u \mid x_{ij}, v_{i}, w_{j})} = \frac{\partial_{x_{ij}^{\ell}} Q_{y_{ij} }(u \mid x_{ij}, v_{i}, w_{j})}{\partial_{x_{ij}^k} Q_{y_{ij} }(u \mid x_{ij}, v_{i}, w_{j})}, \ \ \ell,k = 1,\ldots,d_x, 
$$
provided that $\partial_{x_{ij}^k} Q_{y_{ij} }(u \mid x_{ij}, v_{i}, w_{j}) \neq 0.$ The DR coefficients therefore are proportional to (minus) derivatives of the conditional quantile function, and ratios of DR coefficients correspond to ratios of derivatives.

\begin{remark}[Parametric models] There are many parametric models that are special cases of the DR model. Thus, \cite{ChernozhukovFernandezValMelly2013} and Chernozhukov, Fernandez-Val, Melly, and Wuthrich \citeyearpar{CFVMW2015} showed that the standard linear model, Cox proportional hazard model and Poisson regression model are encompassed by the DR model in the cross section case. These inclusions carry over to the panel versions of these models with two-way unobserved effects. \qed
\end{remark}

\subsection{Estimands} In addition to the model parameter $\beta(y)$, we are interested in measuring the effect on the outcome of changing one of the covariates holding the rest of the covariates and the unobserved effects fixed. Let $x = (t, z')'$, where $t$ is the covariate of interest or treatment and $z$ are the rest of the covariates that usually play the role of controls. One effect of interest is the quantile (left-inverse) function (QF)
$$
Q_k(\tau) = F_k^{\leftarrow}(\tau) := \inf \{y \in \mY : F_k(y) \geq \tau\}  \wedge \sup \{ y \in {\mY}\}, \ \ \tau \in (0,1),
$$
where 
$$
F_k(y) = n^{-1} \sum_{(i,j) \in \mD} \Lambda_y(P(t_{ij}^k,z_{ij}')'\beta(y) + \alpha(v_i,y) + \gamma(w_j,y)),
$$
$t_{ij}^k$ is a level of the treatment that may depend on $t_{ij}$, and $k \in \{0,1\}$.
  We provide examples below. Note that in the construction of the counterfactual distribution $F_k$, we marginalize $(x_{ij},v_i,w_j)$ using the empirical distribution. The resulting effects are \textit{finite population} effects.  We shall focus on these effects because conditioning on the covariates and unobserved effects is natural in the trade application.\footnote{The distinction between finite and infinite population effects does not affect estimation, but affects inference \citep{AAIW-14}. The estimators of infinite population effects need to account for the additional sampling variation coming from the estimation of the distribution of $(x_{ij},v_i,w_j)$.}
We construct the quantile effect function (QEF) by taking differences of the QF at two treatment levels
$$
\Delta(\tau) = Q_1(\tau) - Q_0(\tau), \ \ \tau \in (0,1).
$$ 


We can also obtain the average effect using the relationship between averages and distributions. Thus, the average effect is
$$
\Delta = \mu_1 - \mu_0,
$$ 
where $\mu_k$ is the counterfactual average obtained from $F_k$ as 
\begin{equation}\label{eq:ae}
\mu_k = \int   \, [1(y \geq 0) - F_k(y)] \, dy, \ \ k \in \{0,1\}.
\end{equation}
The integral in \eqref{eq:ae} is over the real line, but the formula nevertheless is applicable to the
case where the support of $dF_k$ is discrete or mixed.

The choice of the levels $t_{ij}^0$ and $t_{ij}^1$ is usually based on the scale of the treatment:
\begin{itemize}
\item If the treatment is binary, $\Delta(\tau)$ is the $\tau$-quantile treatment effect with $t_{ij}^0 = 0$ and $t_{ij}^1 = 1$.
\item If the treatment is continuous,  $\Delta(\tau)$ is the $\tau$-quantile effect of a unitary or one standard deviation increase in the treatment  with $t_{ij}^0 = t_{ij}$ and $t_{ij}^1 = t_{ij} + d$, where $d$ is $1$ or the standard deviation of $t_{ij}$.  
\item If the treatment is the logarithm of a continuous treatment,  $\Delta(\tau)$ is the $\tau$-quantile effect of doubling the treatment (100\% increase) with $t_{ij}^0 = t_{ij}$ and $t_{ij}^1 = t_{ij} +\log 2$. 
\end{itemize}
For example, in the trade application we use the levels $t_{ij}^0 = 0$ and $t_{ij}^1 = 1$ for binary covariates such as the indicators for common legal system and free trade area, and $t_{ij}^0 = t_{ij}$ and $t_{ij}^1 = t_{ij} +\log 2$ for the logarithm of distance.

All the previous estimands have causal interpretation under the standard unconfoundedness or conditional independence assumption for panel data where the conditioning set includes not only the observed controls but also the unobserved effects.


\section{Fixed Effects Estimation and  Uniform Inference} \label{sec:estimation}
To simplify the notation in this section we write  $P(x_{ij}) = x_{ij}$ without loss of generality, and define $\alpha_i(y) :=  \alpha(v_i,y)$ and $\gamma_{j}(y) :=  \gamma(w_{j},y)$. 

\subsection{Fixed Effects  Distribution Regression  Estimator} 
\label{sec:DRestimators}
The parameters of the DR model can be estimated from multiple binary regressions with two-way effects.  To see this, note that the conditional distribution  in \eqref{eq:pdr} can be expressed as 
$$
\Lambda_y(x_{ij}'\beta(y) + \alpha_i(y) + \gamma_{j}(y)) = \Ep[1\{y_{ij} \leq y \} \mid x_{ij}, v_i, w_{j}] .
$$
Accordingly, we can construct  a collection of  binary variables, 
$$1\{y_{ij} \leq y\},  \quad (i,j) \in \mD,  \quad  y \in  \mY,$$
and estimate the parameters  for each $y$ by conditional maximum likelihood with fixed effects. Thus, $\widehat \theta(y) := (\widehat \beta(y),$  $\widehat \alpha_1(y), \ldots,  \widehat \alpha_I(y),$ $\widehat \gamma_1(y), \ldots, \widehat \gamma_{J}(y))$, the \textit{fixed effects distribution regression estimator} of $\theta(y)  := (\beta(y),$  $\alpha_1(y), \ldots,  \alpha_I(y),$ $\gamma_1(y), \ldots,  \gamma_{J}(y))$, is obtained as
\begin{align}
\widehat \theta(y) \in \argmax_{\theta \in \mathbb{R}^{d_x + I + J}}  \sum_{(i,j) \in \mD} 
&
\bigg(1\{y_{ij} \leq y \} \log \Lambda_y(x_{ij}'\beta + \alpha_i + \gamma_{j}) 
\nonumber \\ & \qquad \qquad
+ 1\{y_{ij} > y \} \log [1 - \Lambda_y(x_{ij}'\beta + \alpha_i + \gamma_{j})]
\bigg),
\label{fe-cmle}
\end{align}
for $y \in \mY$.
When the link function is the normal or logistic distribution, the previous program is concave and smooth in parameters and therefore has good computational properties. See \cite{FernandezValWeidner2016}, \cite{CFW-16} and \cite{alpaca17} for a discussion on computation of logit and probit regressions with two-way effects and available software.

The quantile functions and effects are estimated via  plug-in rule, i.e.,
$$
\widehat Q_k(\tau) = \widehat F_k^{\leftarrow}(\tau) \wedge \sup \{ y \in {\mY}\}, \quad \tau \in (0,1), \quad k \in \{0,1\}, 
$$
where
$$
\widehat F_k(y) = n^{-1} \sum_{(i,j) \in \mD}  \Lambda_y((t_{ij}^k,z_{ij}')' \widehat \beta(y) + \widehat \alpha_i(y) + \widehat \gamma_{j}(y)), \ \ y \in {\mY},
$$
and
$$
\widehat \Delta(\tau) = \widehat Q_1(\tau) - \widehat Q_0(\tau) \quad \tau \in (0,1).
$$

{\begin{remark}[Computation]\label{remark:computation} When $\mY$ is not finite, we replace $\mY$ by a finite subset $\bar{\mY}$. Theoretically, this approximation works provided that the Hausdorff distance between $\bar{\mY}$ and $\mY$ goes to zero at a rate faster than $1/{\sqrt n}$. In practice,  if $\mY$ is an interval $[\underline{y}, \bar{y}]$,  $\bar \mY $ can be a fine mesh of $\sqrt{n} \log \log n$ equidistant points covering $\mY$, i.e., $\bar \mY = \{ \underline{y}, \underline{y}+d, \underline{y}+2d, \ldots,  \bar{y}\}$ for $d = (\bar y - \underline{y})/ (\sqrt{n} \log \log n)$. Alternatively, if $\mY$ is the support of $y_{ij}$, $\bar \mY$ can be a grid of $\sqrt{n} \log \log n$ sample quantiles with equidistant indexes.
\end{remark}

\subsection{Incidental Parameter Problem and Bias Corrections} Fixed effects estimators  can be severely biased in nonlinear models because of the incidental parameter problem (\citealp{NeymanScott1948}). These models include the binary regressions that we estimate to obtain the DR coefficients and estimands. We deal with the incidental parameter problem using the analytical bias corrections of
\cite{FernandezValWeidner2016} for parameters and average partial effects (APE) in binary regressions with two-way effects. We note here that  the distributions $F_0(y)$ and $F_1(y)$ can be seen as APE, i.e., they are averages of functions of the data, unobserved effects and parameters.    

The bias corrections are based  on expansions of the bias of the fixed effects estimators as $I,J \to \infty$. For example,  Theorem \ref{th:expansion} shows that
\begin{equation}\label{eq:bias}
\Ep[\widehat F_k(y) - F_k(y)] = \frac{I}{n} B_k^{(F)}(y)+ \frac{J}{n} D_k^{(F)}(y)+ R^{(F)}_k(y),
\end{equation}
where $n R^{(F)}_k(y) = o(I \vee J)$.\footnote{\cite{FernandezValWeidner2016} considered the case where $n = IJ$, i.e., there is no missing data, so that $I/n = 1/J$ and $J/n = 1/I$.}  In Section \ref{sec:asymptotics}  we establish that this expansion holds uniformly in $y \in \mY$ and $k \in \{0,1\}$, i.e.,  $$\sup_{k \in \{0,1\}, y \in \mY} \| n R^{(F)}_k(y) \| = o(I \vee J).$$ 
This result generalizes the analysis  of  \cite{FernandezValWeidner2016} from a single binary regression to multiple (possibly a continuum) of binary regressions. This generalization is required to implement our  inference methods for quantile functions and effects.

The expansion \eqref{eq:bias} is the basis for the bias corrections. 
Let $\widehat B^{(F)}_k(y)$ and $\widehat D^{(F)}_k(y)$ be  estimators of $B^{(F)}_k(y)$ and $D^{(F)}_k(y)$, which are uniformly consistent in $y \in \mY$ and $k \in \{0,1\}$. Bias corrected  fixed effects estimators of $F_k$ and $Q_k$ are formed as
\begin{eqnarray*}
\widetilde Q_k(\tau) &=& \widetilde F_k^{\leftarrow}(\tau) \wedge \sup \{ y \in {\mY}\}, \\ 
\widetilde F_k(y) &=& \widehat F_k(y) - \frac{I}{n} \widehat B_k(y) - \frac{J}{n} \widehat D_k(y), \ \ y \in {\mY}.
\end{eqnarray*}
We also use the corrected estimators $\widetilde F_k$ as the basis for inference and to form a bias corrected estimator of the average effect.
 

\begin{remark}[Shape Restrictions] If the bias corrected estimator $y \mapsto \widetilde F_k(y)$ is non-monotone on $ \mY$, we can rearrange it into a monotone function by simply sorting the values of
function in a nondecreasing order. \cite{ChernozhukovFernandezValGalichon2009} showed that the rearrangement improves the finite sample properties of the  estimator.  Similarly, if the $\widetilde F_k(y)$ takes values outside of $[0,1]$, winsorizing its range to this interval improves the finite sample properties of the estimator (\cite{ccfkl18}). \qed
\end{remark}

\subsection{Uniform Inference}  One inference goal is to construct confidence bands that cover the QF  $\tau \mapsto Q_k(\tau)$  and the QEF  $\tau \mapsto \Delta(\tau)$ simultaneously over a set of quantiles $\mT \subseteq [\varepsilon,1-\varepsilon]$, for some $0 < \varepsilon < 1/2$, and treatment levels $k \in \mK \subseteq \{0,1\}$. 
The set $\mT$ is chosen such that $Q_k(\tau) \in [\inf \{y \in \mY\}, \sup \{y \in \mY \}]$, for all $\tau \in \mT$ and $k \in \mK$. 

We use the generic method of \cite{CFVMW2015} to construct confidence bands for quantile functions and  effects from confidence bands for the corresponding distributions. Let $\mathbb{D}$ denote the space of weakly increasing functions,
mapping ${\mY}$ to $[0,1]$.   
 Assume we have a confidence band $I_k = [L_k,U_k]$ for $F_k$, with lower and upper endpoint functions $y \mapsto L_k(y)$ and $y \mapsto U_k(y)$ such that $L_k, U_k \in \mathbb{D}$ and $L_k(y) \leq U_k(y)$ for all $y \in \mY$.\footnote{If $[L_k',U_k']$ is a confidence band for $F_k$ that does not obey the constraint $L_k',U_k' \in \mathbb{D}$, we can transform $[L_k',U_k']$ into a new band $[L_k,U_k]$ such that $L_k,U_k \in \mathbb{D}$ using the rearrangement method of \cite{ChernozhukovFernandezValGalichon2009}.} 
We say that $I_k$ covers $F_k$ if $F_k \in I_k$ pointwise, namely $L_k(y) \leq F_k(y) \leq U_k(y)$ for all $y \in {\mY}$. If $U_k$ and $L_k$ are some data-dependent bands, we say that $I_k$ is a confidence
band for $F_k$ of level $p$, if $I_k$ covers $F_k$  with probability at least $p$. Similarly, we say that the set of bands $\{I_k : k \in \mK\}$ is a joint confidence band for the set of functions  $\{F_k : k \in \mK\}$ of level $p$, if $I_k$ covers $F_k$  with probability at least $p$ simultaneously over $k \in \mK$. The index set $\mK$ can be a singleton to cover individual confidence bands or  $\mK = \{0,1\}$ to cover joint confidence bands. In Section \ref{sec:asymptotics} we provide a multiplier bootstrap algorithm for
computing joint confidence bands based on the joint asymptotic distribution of the bias corrected estimators $\{\widetilde F_k : k \in \mK\}$.  


The following result provides a method to construct joint confidence bands for  $\{Q_k = F_k^\lto :  k \in \mK\}$, from joint confidence bands for  $\{ F_k : k \in \mK \}$.
\begin{lemma}[\citeauthor{CFVMW2015} (2016, Thm. 2(1))] \label{theorem:bquant} Consider
 a set of distribution functions $\{F_k : k \in \mK\}$   and endpoint functions $ \{L_k : k \in \mK \}$ and $\{U_k : k \in \mK\}$ with components in the class $\mathbb{D}$. If $\{F_k : k \in \mK\}$ is jointly covered by $ \{I_k : k \in \mK \}$
with probability  $p$, then $\{Q_k = F_k^\lto :  k \in \mK\}$ is jointly covered by $ \{I_k^{\lto} : k \in \mK \}$
  with probability $p$, where
$$
I_k^\lto(\tau): =  [U_k^\lto(\tau), L_k^\lto(\tau)], \ \ \tau \in \mT, \ \ k \in \mK.
$$
\end{lemma}This Lemma establishes that we can construct confidence bands for quantile functions by inverting the endpoint functions of confidence bands for distribution functions. The geometric intuition is that the inversion amounts to rotate and flip the bands, and these operations  preserve coverage.  


We next construct simultaneous confidence bands for the quantile effect function $\tau \mapsto \Delta(\tau)$
defined by
$$\Delta(\tau) = Q_1(\tau) - Q_0(\tau) = F_1^\lto(\tau) - F_0^\lto(\tau), \quad \tau \in \mT.$$
The basic idea is to take appropriate differences of the bands for the quantile  functions $Q_1$ and $Q_0$ as the confidence band for the quantile effect.  Specifically,
suppose  we have the set of confidence bands $\{I^\lto_k = [U^\lto_k, L^\lto_k] : k = 0,1\}$ for the set of functions $\{F_k^\lto: k = 0,1\}$ of level $p$.
  \cite{CFVMW2015} showed that a confidence band  for the difference $Q_1 - Q_0$ of size $p$ can be constructed as
 $[U^\lto_1 - L^\lto_0, L^\lto_1 - U^\lto_0]$, i.e., $I^\lto_1 \ominus I^\lto_0$ where $\ominus$ is the pointwise Minkowski difference. 

\begin{lemma}[\citeauthor{CFVMW2015} (2016, Thm. 2(2))] \label{theorem:bqte} Consider a set of distribution functions $\{ F_k : k = 0,1\}$  and endpoint functions $\{ L_k : k = 0,1\}$ and  $\{U_k : k = 0,1\}$, with components in the class $\mathbb{D}$. If  the set of distribution functions $\{ F_k : k = 0,1\}$ is jointly covered by the set of bands $\{I_k : k = 0,1\}$ with probability $p$, then
the quantile effect function $\Delta= F_1^\lto- F_0^\lto$ is covered by $I^{\lto}_{\Delta}$ with probability
at least $p$, where $ I^{\lto}_{\Delta}$ is defined by:
$$
  I^\lto_\Delta(\tau)   :=  
 [U^\lto_1(\tau),  L^\lto_1(\tau)]  \ominus  [ U^\lto_0(\tau),  L^\lto_0(\tau)]  =
 [  U^\lto_1(\tau) -L^\lto_0(\tau),   L^\lto_1(\tau)- U^\lto_0(\tau)], \ \ \tau \in \mT.
 $$
\end{lemma}


\section{Asymptotic Theory}\label{sec:asymptotics}

This section derives  the  asymptotic properties of the fixed effect estimators of $y \mapsto\beta(y)$
and $\{ F_k : k \in {\mathcal K} \}$, as both dimensions $I$ and $J$ grow to infinity. We focus on the case where the link function is the logistic distribution at all levels, $\Lambda_y = \Lambda$, where $\Lambda(\xi) = (1+ \exp(-\xi))^{-1}$.   We choose the  logistic distribution for analytical convenience. In this case the Hessian of the log-likelihood function does not depend on $y_{it}$, leading to several simplifications in the asymptotic expansions. In particular, there are various terms that drop out from the  second order expansions that we use to characterize the structure of the incidental parameter bias of the estimators $\widehat \beta(y)$
and $\widehat F(y)$.  For the case of single binary regressions, \cite{FernandezValWeidner2016} showed that the properties of fixed effects estimators are similar for the logistic distribution and other smooth log-concave distributions such as the normal distribution. Accordingly, we expect that our results can be extended to other  link functions,  but at the cost of more complicated proofs and derivations to account for additional terms.

We make the following assumptions:
\begin{assumption}[Sampling and Model Conditions]~
\label{ass:baseline}
\begin{itemize}
    \item[(i)] Sampling: The outcome variable $y_{ij}$ is independently distributed over $i$ and $j$  conditional on all the observed and unobserved covariates $ \mC_B:= \{(x_{ij}, v_i, w_j): (i,j) \in \mD\}$.
        
   \medskip
    
    \item[(ii)] Model: For all $y \in  \mY$,
\begin{align*}
     F_{y_{ij}}(y  \mid \mC_B)  = F_{y_{ij}}(y  \mid x_{ij}, v_i, w_j) =  \Lambda(x_{ij}' \beta(y) + \alpha(v_{i},y) + \gamma(w_{j},y) ),
\end{align*}
where $y \mapsto \beta(y)$, $y \mapsto \alpha(\cdot,y)$ and $y \mapsto \gamma(\cdot,y)$ are measurable functions.
   \medskip
     

     \item[(iii)]  Compactness:
     the support $\mX$ of $x_{ij}$ is compact, and $\alpha(v_i,y)$ and $\gamma(w_j,y)$
     are bounded uniformly over $i$, $j$, $I$, $J$ and $y \in \mY$.

     \medskip


    \item[(iv)] Compactness and smoothness: Either $\mY$ is a discrete finite set, or $\mY \subset \mathbb{R}$ is a bounded interval. In the latter case, we assume that the conditional density
function $f_{y_{ij}}(y \mid x_{ij}, v_i, w_j)$ exists, is uniformly bounded above and away from zero, and is uniformly continuous
in $y$ on  the interior of $\mY$, uniformly over the support of $(x_{ij}, v_i, w_j)$.
  \medskip

%
  \item[(v)] Missing data: There is only a fixed number of missing observations for every $i$ and $j$, that is,
      $\max_{i} ( J-|\{(i',j') \in \mD : i' = i\}| ) \leq c_2$
     and $\max_{j}( I - |\{(i',j') \in \mD : j' = j\}| ) \leq c_2$
     for some constant $c_2<\infty$ that is independent of the sample size.
     
   \medskip
     
     \item[(vi)]  Non-collinearity: The regressors $x_{ij}$ are non-collinear after projecting out the two-way fixed effects,
     that is, there exists a constant
        $c_3>0$, independent of the sample size, such that
     \begin{align*} 
         \min_{\{ \delta \in \mathbb{R}^{d_x} \, : \, \| \delta \|=1 \}} \; \;  \min_{(a,b) \in \mathbb{R}^{I+J}} 
         \left[ 
         \frac 1 n  \sum_{(i,j) \in \mD}     (x'_{ij} \delta  -  a_i - b_j )^2 \right]    \geq \; c_3 .
    \end{align*} 
    
    \item[(vii)] Asymptotics: We consider asymptotic sequences where $I_n, J_n \to \infty$ with $I_n/J_n \to c$ for some positive and finite $c$, as the total sample size $n \to \infty$.  We drop the indexing by $n$ from $I_n$ and $J_n$, i.e. we shall write $I$ and $J$.
     \end{itemize}
 
 \end{assumption}
 
 \begin{remark}[Assumption \ref{ass:baseline}]
Part (i) holds if  $(y_{ij},x_{ij})$ is i.i.d. over $i$ and  $j$, $v_i$ is i.i.d. over $i$, and $w_j$ is i.i.d. over $j$; but it is more general as it does not restrict the distribution of  $(x_{ij},v_i,w_j)$ nor its dependence across $i$ and $j$.  We show how to relax this assumption allowing for a form of weak conditional dependence in Section \ref{subsec:clustering}. Part (ii) holds if the observed covariates are strictly exogenous conditional on the unobserved effects and the conditional distribution is correctly specified for all $y \in \mY$.  We expect that our theory carries over to predetermined or weakly exogenous covariates that are relevant in panel data models, following the analysis \cite{FernandezValWeidner2016}. We focus on the strict exogeneity assumption because it is applicable to both panel and network data, and leave the extension to weak exogeneity  to future research.   Part (iii) imposes that the covariates $x_{ij}$ and  
unobserved effects $\alpha(v_i,y)$ and $\gamma_j(w_j,y)$ are all uniformly bounded.
For fixed values $y$ it is possible to obtain asymptotic results of our estimators without
the  compact  support assumption, see e.g.~\cite{Yan2016statistical}, but deriving  empirical process results that hold uniformly over $y$ is much more involved without this assumption. The compact support assumption guarantees that the conditional probabilities
of the  events $\{y_{ij} \leq y\}$ are bounded away from zero and one, that is, the network of binarized outcomes
$1\{y_{ij} \leq y\}$ is assumed to be dense. In the network econometrics literature
 \cite{charbonneau17}, \cite{Graham2017} and \cite{jochmans18} provide methods 
that are also applicable to sparse networks.
Part (iv) can be slightly weakened to   Lipschitz continuity with uniformly bounded Lipschitz constant, instead of differentiability. It covers discrete, continuous, and  mixed outcomes with mass points at the boundary of the support such as censored variables. For the mixed outcomes, the data generating process for the mass points can be arbitrarily different from the rest of the support because the density $y \mapsto f_{y_{ij}}(y \mid \cdot)$ only needs to be continuous in the interior of $\mY$. 
Part (v) of the assumption allows for a finite (and asymptotically bounded) number of missing observations for each unit $i$,
and each unit $j$. For example, in the trade network example only the observations with $i=j$ are missing,
implying that there is  one missing observation for every $i$ and for every $j$, i.e. $c_2=1$.
If the panel is balanced, part (vi) can be stated  as
        \begin{align*}
              \frac 1 {I J} \sum_{i=1}^I \sum_{j=1}^J   \widetilde x_{ij}  \widetilde x'_{ij} \;  \geq \; c_3 \;  \mathbb{I}_{d_x},
        \end{align*}
        where $\widetilde x_{ij} = x_{ij} - x_{i \cdot} - x_{\cdot j} + x_{\cdot \cdot}$,
         $x_{i \cdot} =  J^{-1}  \sum_{j=1}^J x_{ij}$, $x_{\cdot j} = I^{-1} \sum_{i=1}^I x_{ij}$,
        and $x_{\cdot \cdot} =  (I J)^{-1} \sum_{i=1}^I \sum_{j=1}^J  x_{ij}$. This is the typical condition in linear panel models requiring that all the covariates display variation in both dimensions. The asymptotic sequences considered in part (vii)  exactly balance the order of the bias and standard deviation of the fixed effect estimator yielding a non-degenerate asymptotic distribution. \qed
\end{remark}

\subsection{Asymptotic Distribution of the Uncorrected Estimator}
\label{sec:AsyDistrUncorrected}
We introduce first some further notation.
Denote the $q^{th}$ derivatives of the cdf $\Lambda$ by $\Lambda^{(q)}$,
and define $\Lambda^{(q)}_{ij}(y) = \Lambda^{(q)}(  x_{ij}'\beta(y) + \alpha_{i}(y) + \gamma_{j}(y) )$
and $\Lambda^{(q)}_{ij,k}(y) = \Lambda^{(q)}(  \mathbbm{x}_{ij,k}'\beta(y) + \alpha_{i}(y) + \gamma_{j}(y) )$ with $\mathbbm{x}_{ij,k}:=(t_{ij}^k,z_{ij}')'$ and $q = 1,2, \ldots$. 
For $\ell \in \{1,\ldots, d_x\}$
define the following projections of the $\ell$'th covariate $x^\ell_{ij}$,
\begin{align}
  \left( \alpha^\ell_x(y) , \gamma^\ell_x(y) \right)
  &\in
    \arg \min_{(a,c) \in \mathbb{R}^{I+J}} 
       \left[ 
       \sum_{(i,j) \in \mD}  \Lambda^{(1)}_{ij}(y) \, \left(x^{\ell}_{ij}  -  a_i - c_j \right)^2 \right] ,
   \label{DefProjectX}    
\end{align}
and let $ \alpha_{x,i}(y)$ and $\gamma_{x,j}(y)$ be the $d_x$-vectors with components 
$\alpha^\ell_{x,i}(y)$ and $\gamma^\ell_{x,j}(y)$, where $\alpha^\ell_{x,i}(y)$ is the $i$th component of $\alpha^\ell_{x}(y)$ and $\gamma^\ell_{x,j}(y)$ 
is the $j$th component of $\gamma^\ell_{x}(y)$.
Also define $\widetilde x_{ij}(y) =  x_{ij} -  \alpha_{x,i}(y) -  \gamma_{x,j}(y)$ 
and  $\widetilde{\mathbbm{x}}_{ij,k}(y) =  \mathbbm{x}_{ij,k} -  \alpha_{x,i}(y) -  \gamma_{x,j}(y)$.
Notice that $\widetilde{\mathbbm{x}}_{ij,k}(y)$ is defined using  projections  of $x_{ij}$ instead of ${\mathbbm{x}}_{ij,k}$.
Also,  while the locations of $\alpha_{x,i}(y)$ and $\gamma_{x,j}(y)$ are not identified, 
$\widetilde x_{ij}(y)$ and $\widetilde{\mathbbm{x}}_{ij,k}(y)$ are uniquely defined.
Analogous to the projection of $x_{ij}^{\ell}$ above, we define 
$ \Psi_{ij,k}(y) = \alpha^\Psi_{i}(y) + \gamma^\Psi_{j}(y) $, where
\begin{align}
    \left( \alpha^\Psi(y) , \gamma^\Psi(y) \right)
  &\in
    \arg \min_{(a,c) \in \mathbb{R}^{I+J}} 
       \left[ 
       \sum_{(i,j) \in \mD}  \Lambda^{(1)}_{ij}(y) 
       \, \left(  \frac{\Lambda^{(1)}_{ij,k}(y)} {\Lambda^{(1)}_{ij}(y)} 
         -  a_i - c_j \right)^2
        \right] .
    \label{DefPsi}        
\end{align}
For example, if $\mathbbm{x}_{ij,k} = x_{ij}$, then  $\Psi_{ij,k}(y)=1$.
Furthermore, we define\footnote{\label{footnote1}
The FOC of  problem \eqref{DefProjectX}
imply that $\sum_{(i,j) \in \mD}  \Lambda^{(1)}_{ij,k}(y)  \, \widetilde{x}_{ij}(y)^{\, \prime} =0$,
and we can therefore equivalently write
$
    \partial_{\beta} F_k(y) 
   = \frac 1 {n} \sum_{(i,j) \in \mD}  \Lambda^{(1)}_{ij,k}(y)  \, 
   \left[ \widetilde{\mathbbm{x}}_{ij,k}(y) -  \widetilde{x}_{ij}(y) \right]^{\, \prime}
   = \frac 1 {n} \sum_{(i,j) \in \mD}  \Lambda^{(1)}_{ij,k}(y)  \, 
   \left[  {\mathbbm{x}}_{ij,k}(y) -   {x}_{ij}(y) \right]^{\, \prime}.
$
}
\begin{align*}
     W(y) &=  \frac 1 {n} \sum_{(i,j) \in \mD} \Lambda^{(1)}_{ij}(y) \, \widetilde x_{ij}(y) \, \widetilde x_{ij}(y)',
   &
  \partial_{\beta} F_k(y) 
   &= \frac 1 {n} \sum_{(i,j) \in \mD}  \Lambda^{(1)}_{ij,k}(y)  \, \widetilde{\mathbbm{x}}_{ij,k}(y)^{\, \prime} ,
\end{align*}
and
\begin{align*}
     B^{(\beta)}(y) &=  -  \frac 1 2   W^{-1}(y)  \left[ \frac 1 {I}  \sum_{i=1}^I  
               \frac{\sum_{j \in \mD_i}   \Lambda^{(2)}_{ij}(y) \, \widetilde x_{ij}(y)  }
                      { \sum_{j \in \mD_i}   \Lambda^{(1)}_{ij}(y) } \right] ,
    \\
     D^{(\beta)}(y)   &=  - \frac 1 2   W^{-1}(y)  \left[   \frac 1 {J}  \sum_{j=1}^J  
               \frac{\sum_{i \in \mD_j}   \Lambda^{(2)}_{ij}(y) \, \widetilde x_{ij}(y)  }
                      { \sum_{i \in \mD_j}   \Lambda^{(1)}_{ij}(y) } \right] ,       
   \\
     B^{(\Lambda)}_{k}(y) &=  \frac 1 {2 \, I}  \sum_{i=1}^I  
               \frac{\sum_{j \in \mD_i} \left[   \Lambda^{(2)}_{ij,k}(y)  - \Lambda^{(2)}_{ij}(y) \Psi_{ij,k}(y) \right] }
                      { \sum_{j \in \mD_i}   \Lambda^{(1)}_{ij}(y) } ,
    \\
     D^{(\Lambda)}_{k}(y)   &=   \frac 1 {2 \, J}  \sum_{j=1}^J  
               \frac{\sum_{i \in \mD_j} \left[   \Lambda^{(2)}_{ij,k}(y)  - \Lambda^{(2)}_{ij}(y) \Psi_{ij,k}(y) \right]  }
                      { \sum_{i \in \mD_j}   \Lambda^{(1)}_{ij}(y) },      
\end{align*}
where $\mD_i := \{(i',j') \in \mD : i' = i\}$ and $\mD_j := \{(i',j') \in \mD : j' = j\}$ are the subsets of observational units that contain the index $i$ and $j$, respectively. In the previous expressions,  $\partial_{\beta} F_k(y)$ is a $1 \times d_x$ vector
for each $k \in {\mathcal K}$ that we stack in the  $|{\mathcal K}| \times d_x$ matrix
$\partial_{\beta} F(y)
  = [\partial_{\beta} F_{k}(y) \, : \, k \in {\mathcal K}]$.
Similarly, $F_{k}(y)$, $B^{(\Lambda)}_{k}(y)$, $ D^{(\Lambda)}_{k}(y)$
and $\Psi_{ij,k}(y)$ are scalars for each $k \in {\mathcal K}$,
that we stack in the
$|{\mathcal K}| \times 1$ vectors $ F(y) =  [F_{k}(y) \, : \, k \in {\mathcal K}]$,
 $ B^{(\Lambda)}(y) =  [B^{(\Lambda)}_{k}(y) \, : \, k \in {\mathcal K}]$,
$ D^{(\Lambda)}(y) =  [D^{(\Lambda)}_{k}(y) \, : \, k \in {\mathcal K}]$,
$ \Psi_{ij}(y) =  [\Psi_{ij,k}(y) \, : \, k \in {\mathcal K}]$.

Let $\ell^{\infty}(\mY)$ be the space of real-valued bounded functions on $\mY$ equipped with the sup-norm $\| \cdot \|_{\mY}$, and $\rightsquigarrow$ denote weak convergence (in distribution). We establish a functional central limit theorem for the fixed effects estimators of $y \mapsto \beta(y)$ and $y \mapsto F(y)$ in $\mY$. All stochastic statements are conditional on $  \{(x_{ij}, v_i, w_j): (i,j) \in \mD\}$.
\begin{theorem}[FCLT for Fixed Effects DR Estimators]
    \label{th:expansion}
    Let Assumption~\ref{ass:baseline} hold.
     For all $y_1,y_2 \in  \mY$ with $y_1 \geq y_2$ we assume the existence of
    \begin{align*}
         \overline V(y_1,y_2)
         &=
         \plim_{n \rightarrow \infty}
          \frac 1 n \sum_{(i,j) \in \mD}
           \Lambda_{ij}(y_1)  \left[  1 -   \Lambda_{ij}(y_2) \right]
            \;  \widetilde x_{ij}(y_1) \; \widetilde x_{ij}(y_2)' ,
        \\   
        \overline \Omega(y_1,y_2)
         &=
         \plim_{n \rightarrow \infty}
          \frac 1 n \sum_{(i,j) \in \mD}
           \Lambda_{ij}(y_1)  \left[  1 -   \Lambda_{ij}(y_2) \right] \; \Xi_{ij}(y_1) \Xi_{ij}(y_2)' ,
    \end{align*}    
    where $\Xi_{ij}(y) = \Psi_{ij}(y)
            +  \partial_{\beta} F(y)   W^{-1}(y)  \, \widetilde x_{ij}(y)$. Let $\overline V(y_2,y_1) := \overline V(y_1,y_2)'$,
    $\overline \Omega(y_2,y_1) := \overline \Omega(y_1,y_2)'$, and  $\overline W(y_1) := \overline V(y_1,y_1)$. 
Then, in the metric space $\ell^{\infty}(\mY)^{d_x}$,
    \begin{align*}
       \sqrt{n}\left[  \widehat \beta(y)  - \beta(y)
                    -  \frac I {n}    B^{(\beta)}(y)
                   -  \frac J {n}     D^{(\beta)}(y) 
                \right]     
          \rightsquigarrow  Z^{(\beta)}(y) ,
    \end{align*}
    and, in the metric space $\ell^{\infty}(\mY)^{|\mK|}$,
    \begin{multline*}
       \sqrt{n}\left\{  \widehat F(y) -  F(y)
                    -  \frac I {n} 
          \underset{B^{(F)}(y)}{\underbrace{\left[ 
          B^{(\Lambda)}(y) 
         + (\partial_{\beta} F(y))    B^{(\beta)}(y)
        \right]}}
                  - \frac J {n} 
         \underset{D^{(F)}(y)}{\underbrace{ \left[
       D^{(\Lambda)}(y)
      + (\partial_{\beta} F(y))      D^{(\beta)}(y) 
      \right]}}            
                \right\}     \\
          \rightsquigarrow  Z^{(F)}(y) ,
    \end{multline*}
    as stochastic processes indexed by $y \in \mY$, where $y \mapsto Z^{(\beta)}(y)$
    and $y \mapsto Z^{(F)}(y)$ 
    are tight zero-mean Gaussian processes with covariance functions
   $(y_1,y_2) \mapsto \overline W^{-1}(y_1) \;  \overline V(y_1,y_2) \;  \overline W^{-1}(y_2)$
     and      
   $(y_1,y_2) \mapsto  \overline \Omega(y_1,y_2)$, respectively.
\end{theorem}

Assumption \ref{ass:baseline}(vi) guarantees the invertibility of $W(y)$ and $\overline W(y) $.
Notice that $\overline W(y) $ is equal to the limit of  $W(y)$ because
 $\Lambda^{(1)}_{ij}(y) = \Lambda_{ij}(y)  \left[  1 -   \Lambda_{ij}(y) \right]$ by the properties of the logistic distribution. This information equality follows by the correct specification condition in Assumption~\ref{ass:baseline}(ii).  
By Assumption \ref{ass:baseline}(v), we could have used
$\sqrt{IJ}$ instead of $\sqrt{n}$, $1/J$ instead of $I/n$, and $1/I$ instead of $J/n$.
However, if the panel is not balanced, then we expect 
the expressions in the theorem to provide a more accurate finite-sample approximation, because the 
standard deviation of the estimates will generally be of order $1/\sqrt{n}$ for unbalanced panels,
and the leading order incidental parameter biases are generally proportional to the number of incidental parameters
($I$ and $J$ here) divided by the total sample size $n$, see e.g. \cite{FernandezValWeidner2018}.

\begin{remark}[Comparison with binary response models] \cite{FernandezValWeidner2016} derived central limit theorems (CLTs) for the fixed effects estimators of coefficients and APEs in panel regressions with two-way effects.  Pointwise, for given $y \in \mY$, Theorem~\ref{th:expansion} yields these CLTs. Moreover, it covers multiple binary regressions  by establishing the limiting distribution of  $\widehat \beta(y) $ and $ \widehat F(y)$ treated as stochastic processes indexed by $y \in \mY$. This generalization is key for our inference results and does not follow from well-known empirical process results. We need to deal with a double asymptotic approximation where both $I$ and $J$ grow to infinity, and to bound all the remainder terms in the second order expansions used by \cite{FernandezValWeidner2016}  uniformly over $y \in \mY$. 
       We refer to the appendix and supplementary material for more details. \qed
\end{remark}

\begin{remark}[Case $\mathbbm{x}_{ij,k} = x_{ij}$]  
When $\mathbbm{x}_{ij,k} = x_{ij}$, that is, when the counterfactual values are equal to the observed values, then
 the asymptotic bias of $\widehat F_k$ vanishes, because  $B_k^{(\Lambda)}(y) = D_k^{(\Lambda)}(y) = 0$, and $\partial_{\beta} F_k(y)=0$ (see footnote \ref{footnote1}). In fact, 
in that case
$\widehat F_k$ is equal to the empirical distribution function, namely
$$
\widehat F_k(y) = \frac{1}{n} \sum_{(i,j) \in \mD}
           \Lambda(x_{ij}'\widehat \beta(y) + \widehat \alpha_i(y) + \widehat \gamma_j(y)) = \frac{1}{n} \sum_{(i,j) \in \mD} 1\{y_{ij} \leq y \},
$$
by the first order conditions of the fixed effects logit DR estimator with respect to the fixed effect parameters. This property provides another appealing feature to choose  the logistic distribution. \qed
\end{remark}

\subsection{Bias Corrections} 
Theorem \ref{th:expansion} shows that the fixed effects DR estimator has asymptotic bias of the same order as the asymptotic standard deviation under the  approximation that we consider. The finite-sample implications are that this estimator can have substantial bias and that confidence regions constructed around it can have severe undercoverage. We deal with these problems by removing the first order bias of the estimator.

We estimate the bias components using the plug-in rule. Define
$\widehat \Lambda^{(q)}_{ij}(y) = \Lambda^{(q)}(  x_{ij}' \widehat \beta(y) + \widehat \alpha_{i}(y) + \widehat \gamma_{j}(y) )$
and $\widehat \Lambda^{(q)}_{ij,k}(y) = \widehat \Lambda^{(q)}(  \mathbbm{x}_{ij,k}' \widehat \beta(y) + \widehat \alpha_{i}(y) + 
\widehat \gamma_{j}(y) )$.
Replacing $\Lambda^{(1)}_{ij}(y)$ and $\Lambda^{(1)}_{ij,k}(y)$
by $\widehat \Lambda^{(1)}_{ij}(y)$ and $\widehat \Lambda^{(1)}_{ij,k}(y)$
in the definitions of  $\alpha^\ell_x(y)$, $\gamma^\ell_x(y)$, $\alpha^\Psi(y)$, and $\gamma^\Psi(y)$ yields the
corresponding estimators. We plug-in these estimators to obtain
$\widehat x_{ij}(y) =  x_{ij} -  \widehat \alpha_{x,i}(y) -  \widehat \gamma_{x,j}(y)$,
$\widehat  {\mathbbm x}_{ij,k}(y) =  \mathbbm{x}_{ij,k} -  \widehat \alpha_{x,i}(y) -  \widehat \gamma_{x,j}(y)$,
and $\widehat \Psi_{ij,k}(y) = \widehat \alpha^\Psi_{i}(y) + \widehat \gamma^\Psi_{j}(y) $.
Then we construct
\begin{align*}
     \widehat W(y) &=  \frac 1 {n} \sum_{(i,j) \in \mD} \widehat \Lambda^{(1)}_{ij}(y) \, \widehat x_{ij}(y) \, \widehat x'_{ij}(y) ,
   &
  \partial_{\beta} \widehat F_k(y) 
   &= \frac 1 {n} \sum_{(i,j) \in \mD}  \widehat \Lambda^{(1)}_{ij,k}(y)  \, \widehat{\mathbbm{x}}_{ij,k}(y)^{\, \prime} ,
\end{align*}
and
\begin{align*}
     \widehat B^{(\beta)}(y) &=  -  \frac 1 2   \widehat W^{-1}(y)  \left[ \frac 1 {I}  \sum_{i=1}^I  
               \frac{\sum_{j \in \mD_i}   \widehat \Lambda^{(2)}_{ij}(y) \, \widehat x_{ij}(y)  }
                      { \sum_{j \in \mD_i}   \widehat \Lambda^{(1)}_{ij}(y) } \right] ,
    \\
     \widehat D^{(\beta)}(y)   &=  - \frac 1 2   \widehat W^{-1}(y)  \left[   \frac 1 {J}  \sum_{j=1}^J  
               \frac{\sum_{i \in \mD_j}   \widehat \Lambda^{(2)}_{ij}(y) \, \widehat x_{ij}(y)  }
                      { \sum_{i \in \mD_j}   \widehat \Lambda^{(1)}_{ij}(y) } \right] ,       
   \\
     \widehat B^{(\Lambda)}_{k}(y) &=  \frac 1 {2 \, I}  \sum_{i=1}^I  
               \frac{\sum_{j \in \mD_i} \left[   \widehat \Lambda^{(2)}_{ij,k}(y)  - \widehat \Lambda^{(2)}_{ij}(y) \widehat \Psi_{ij,k}(y) \right] }
                      { \sum_{j \in \mD_i}   \widehat \Lambda^{(1)}_{ij}(y) } ,
    \\
     \widehat D^{(\Lambda)}_{k}(y)   &=   \frac 1 {2 \, J}  \sum_{j=1}^J  
               \frac{\sum_{i \in \mD_j} \left[   \widehat \Lambda^{(2)}_{ij,k}(y)  - \widehat \Lambda^{(2)}_{ij}(y) \widehat \Psi_{ij,k}(y) \right]  }
                      { \sum_{i \in \mD_j}   \widehat \Lambda^{(1)}_{ij}(y) }  .     
\end{align*}
We also define the $|{\mathcal K}| \times d_x$ matrix
$\partial_{\beta} \widehat F(y)
  = [(\partial_{\beta} \widehat F_{k}(y)) \, : \, k \in {\mathcal K}]$,
and the  $|{\mathcal K}| \times 1$ vectors 
 $ \widehat B^{(F)}(y) =  [\widehat B^{(F)}_{k}(y) \, : \, k \in {\mathcal K}]$,
$ \widehat D^{(F)}(y) =  [\widehat D^{(F)}_{k}(y) \, : \, k \in {\mathcal K}]$,
$ \widehat \Psi_{ij}(y) =  [\widehat \Psi_{ij,k}(y) \, : \, k \in {\mathcal K}]$. 
Finally, we also construct the estimator of the asymptotic variance of $\widehat F(y)$
\begin{align*}
        \widehat  \Omega(y)
         &=
          \frac 1 n \sum_{(i,j) \in \mD}
           \widehat  \Lambda^{(1)}_{ij}(y)   \; 
       \widehat \Xi(y) \;  \widehat \Xi(y)' .
\end{align*}  
where $\widehat \Xi(y) = \widehat  \Psi_{ij}(y)
            +  (\partial_{\beta} \widehat  F(y))   \widehat  W^{-1}(y)  \, \widehat  x_{ij}(y)$.
            
 Lemma~\ref{lemma:bias_estimators} in the Appendix   shows that the estimators of the asymptotic bias are consistent, uniformly in $y \in \mY$.          
Bias corrected estimators of $\beta(y)$ and $F(y)$ can then be formed as
\begin{equation}\label{eq:bc}
     \widetilde \beta(y) = \widehat \beta(y)  
                    -  \frac I {n}    \widehat B^{(\beta)}(y)
                   -  \frac J {n}      \widehat D^{(\beta)}(y)  ,
\end{equation}
and
\begin{align*}
   \widetilde F(y) =   \widehat F(y)  
                    -  \frac I {n} 
          \underset{\widehat B^{(F)}(y)}{\underbrace{\left[ 
          \widehat B^{(\Lambda)}(y) 
         + (\partial_{\beta} \widehat F(y))    \widehat B^{(\beta)}(y)
        \right]}}
                  - \frac J {n} 
         \underset{\widehat D^{(F)}(y)}{\underbrace{ \left[
       \widehat D^{(\Lambda)}(y)
      + (\partial_{\beta} \widehat F(y))      \widehat D^{(\beta)}(y) 
      \right] }}.
\end{align*}
Alternatively, we could define the bias corrected version of $ \widehat F(y)  $ as
\begin{align*}
\widetilde F^*_k(y) &=\left[  \frac 1 n  \sum_{(i,j) \in \mD} 
 \Lambda \left( \mathbbm{x}_{ij,k}' \, \widetilde \beta(y) + \widetilde \alpha_i(y) + \widetilde \gamma_{j}(y) \right) \right]
  -  \frac I {n} 
          \widehat B^{(\Lambda)}_k(y)
 - \frac J {n} 
       \widehat D^{(\Lambda)}_k(y) ,
\end{align*}
where $\widetilde \xi(y) := (\widetilde \alpha_1(y), \ldots, \widetilde \alpha_I(y), \widetilde \gamma_{1}(y), \ldots, \widetilde \gamma_J(y))$ is a solution to
$$
   \max_{\xi \in \mathbb{R}^{I+J}}  \sum_{(i,j) \in \mD}
      (1\{y_{ij} \leq y \} \log \Lambda(x_{ij}'\widetilde \beta(y) + \alpha_i + \gamma_{j}) + 1\{y_{ij} > y \} \log [1 - \Lambda(x_{ij}'\widetilde \beta(y) + \alpha_i + \gamma_{j})]).
$$
It can be shown that  $\sup_{y \in  \mY} \sqrt{n} \left| \widetilde F^*_k(y) - \widetilde F_k(y) \right| = o_P(1)$,
that is, the difference between those alternative bias corrected estimators is asymptotically negligible. There is no obvious
reason to prefer one over the other, and we present result for $\widetilde F_k$ in the following, which equivalently
hold for $\widetilde F^*_k$.\footnote{We use the estimator $\widetilde F^*_k$ in the numerical examples for computational convenience as the bias correction involves estimating less terms. } 

\begin{remark}[Alternative Approaches] The conditional approach of \cite{charbonneau17} and \cite{jochmans18}  for the logit model with two-way effects could be also adopted to estimate the coefficient $\beta(y)$. However, this approach does not produce estimators of $F(y)$ as it is based on differencing-out the unobserved effects. The bias correction method proposed is analytical in that it requires explicit characterization and estimation of the bias. A natural alternative is a correction based on Jackknife or bootstrap following the analysis of  \cite{CFW-16},  \cite{DhaeneJochmans2015}, \cite{FernandezValWeidner2016}, \cite{Hahn:2004p882}, and \cite{KimSun2016} for nonlinear panel models. We do not consider any of  these corrections because they require repeated parameter estimation that can be  computationally expensive in this case.   \qed
\end{remark}


 The following main result establishes the functional central limit theorem for the bias corrected  estimators and uniform consistency of the estimators of the variance function.
\begin{theorem}[FCLT for Bias Corrected Fixed Effects DR Estimators]
    \label{th:bc}
    Let Assumption~\ref{ass:baseline} hold.
 Then, in the metric space $\ell^{\infty}( \mY)^{d_x}$,
    \begin{align*}
       \sqrt{n}\left[  \widetilde \beta(y)  - \beta(y)
                \right]     
          \rightsquigarrow  Z^{(\beta)}(y) ,
    \end{align*}
    and, in the metric space $\ell^{\infty}(\mY)^{|\mK|}$,
   $$
       \sqrt{n}\left[ \widetilde F(y) -  F(y)
      \right]            
          \rightsquigarrow  Z^{(F)}(y) ,
    $$
        as stochastic processes indexed by $y \in \mY$,    where $Z^{(\beta)}(y)$
    and $Z^{(F)}(y)$ 
    are the same Gaussian processes
    that appear in Theorem~\ref{th:expansion}. Moreover,
    $$
     \sup_{y \in  \mY} \left\|  \widehat W(y)^{-1} - \overline W(y)^{-1} \right\| = o_P(1) \ \  \text{ and } \ \ \sup_{y \in  \mY} \left\|   \widehat  \Omega(y) - \overline \Omega(y) \right\|= o_P(1).
    $$
\end{theorem}

\subsection{Uniform Confidence Bands and Bootstrap}
\label{sec:UniformBands}
We show how to  construct pointwise and uniform confidence bands for $y \mapsto \beta(y)$ and $y \mapsto F(y)$ on $ \mY$ using 
 Theorem~\ref{th:bc}. The uniform bands for $F$ can be used as inputs in Lemmas \ref{theorem:bquant} and \ref{theorem:bqte} to construct uniform bands for the QFs $\tau \mapsto Q_k(\tau) = F_k^{\lto}(\tau),$ $k \in \mK$, and the QEF $\tau \mapsto \Delta(\tau)$ on $\mT$.

Let $\mB \subseteq \{1, \ldots, d_x\}$ be the set of indexes for the coefficients of interest.  For given $y \in  \mY$, $\ell \in \mB$, $k \in \mK$, and $p \in (0,1)$,  a pointwise $p$-confidence interval
for $\beta_{\ell}(y)$, the $\ell$'th component of $\beta(y)$, is
\begin{equation}\label{eq:ci}
[ \widetilde \beta_{\ell}(y) \pm \Phi^{-1}(1-p/2) \widehat{\sigma}_{\beta_\ell}(y)],
\end{equation}
and a pointwise $p$-confidence intervals
for $F_k(y)$ is
$$[ \widetilde F_k(y) \pm \Phi^{-1}(1-p/2) \widehat{\sigma}_{F_k}(y)],$$ 
where $\Phi$ denotes the cdf of the standard normal distribution, $\widehat{\sigma}_{\beta_\ell}(y)$
is the standard error of $\widetilde \beta_{\ell}(y)$ given in \eqref{eq:se_beta}, and $\widehat{\sigma}_{F_k}(y)$
is the standard error of $\widetilde F_k(y)$ given in \eqref{eq:se_F}.  These intervals have coverage $p$ in large samples by 
 Theorem~\ref{th:bc}.


We construct  joint uniform bands for  the coefficients and distributions using Kolmogorov-Smirnov type critical values,  instead of quantiles from the normal distribution. A uniform $p$-confidence band
joint for the vector of functions $\{ \beta_{\ell}(y) : \ell \in \mB, y \in  \mY\}$ is
\begin{equation}\label{ConfBandCoeff}
I_{\beta} = \{ [ \widetilde \beta_{\ell}(y) \pm t_{\mB, \mY}^{(\beta)}(p) \widehat{\sigma}_{\beta_{\ell}}(y)] : \ell \in \mB, y \in  \mY\},
\end{equation}
where $t_{\mB, \mY}^{(\beta)}(p)$ is the $p$-quantile of the maximal $t$-statistic 
\begin{align}\label{eq:maxt-coeff}
    t_{\mB, \mY}^{(\beta)} = \sup_{y \in  \mY, \, \ell \in \mB} \frac{\big| Z^{(\beta)}_{\ell}(y) \big|}{\sigma^{(\beta)}_{\ell}(y)},
\end{align} 
where $\sigma^{(\beta)}_{\ell}(y) = [\overline{W}(y)^{-1}]_{\ell,\ell}^{1/2},$ the square root of the $(\ell,\ell)$ element of the matrix  $\overline{W}(y)^{-1}$. 
Similarly, a uniform $p$-confidence band joint
for the set of distribution functions $\{ F_k(y) : k \in \mK, y \in  \mY\}$ is
\begin{align}
    I_{F} = \{[ \widetilde F_k(y) \pm t_{\mK,  \mY}^{(F)}(p) \widehat{\sigma}_{F_k}(y)] : k \in \mK, y \in  \mY\},
    \label{ConfBandCDF}
\end{align}    
where  $t_{\mK,  \mY}^{(F)}(p)$ is the $p$-quantile of the maximal $t$-statistic 
\begin{align}\label{eq:maxt-dist}
    t_{\mK,  \mY}^{(F)} = \sup_{y \in  \mY, \, k \in \mK} \frac{\big| Z^{(F)}_{k}(y) \big|}{\sigma^{(F)}_{k}(y)},
\end{align} 
where $\sigma^{(F)}_{k}(y) = [\overline{\Omega}(y)]_{k,k}^{1/2}$, the square root of the $(k,k)$ element of the matrix $\overline{\Omega}(y,y)$.
The previous confidence bands also have coverage $p$ in large samples by Theorem~\ref{th:bc}.


The maximal t-statistics used to construct the bands  $I_{\beta}$  and $I_{F}$  are not pivotal, but  their distributions can be approximated by simulation after replacing the variance functions of the limit processes by uniformly consistent estimators. In practice, however, we find it more convenient to use resampling methods.   We consider a multiplier bootstrap scheme that resamples the efficient scores or influence functions of the fixed effects estimators $\widehat \beta(y)$ and $\widehat F(y)$. This scheme is computationally convenient because it does not need to solve the high dimensional nonlinear fixed effects conditional maximum likelihood program \eqref{fe-cmle} or making any bias correction in each bootstrap replication. In these constructions we rely on the uncorrected fixed effects estimators instead of the bias corrected estimators, because they have the same influence functions and the uncorrected estimators are consistent under the asymptotic approximation that we consider.

To describe the standard errors and multiplier bootstrap we need to introduce some notation for the influence functions of $\widehat \theta(y)$ and $\widehat F(y)$.  Let $\theta = (\beta, \alpha_1, \ldots, \alpha_I, \gamma_1, \ldots, \gamma_J)$ be a generic value for the parameter $\theta(y)$, the influence function of $\widehat \theta(y)$ is the $(d_x+I+J)$-vector $\psi_{ij}^y(\theta(y))$, where 
$$
\psi_{ij}^y(\theta) = H(\theta)^{\dagger} [\boldsymbol{1}\{y_{ij} \leq y \} - \Lambda(x_{ij}'\beta + \alpha_i + \gamma_j)] w_{ij}, \ \ w_{ij} = (x_{ij}, e_{i,I}, e_{j,J}) , \ \ y \in  \mY,
$$
$e_{i,I}$ is a unit vector of dimension $I$ with a one in the position $i$, $e_{j,J}$ is defined analogously, $H(\theta)^{\dagger}$ is the Moore-Penrose pseudo-inverse of $H(\theta)$, and
$$
H(\theta) = \frac{1}{n} \sum_{(i,j)\in \mD} \Lambda^{(1)}(x_{ij}'\beta + \alpha_i + \gamma_j) w_{ij} w_{ij}', 
\ \ \Lambda^{(1)}(z) = \Lambda(z) \Lambda(-z),
$$ 
is minus the Hessian of the log-likelihood with respect to $\theta$, which does not depend on $y$ in the case of the logistic distribution.\footnote{We use the Moore-Penrose pseudo-inverse because $H(\theta)$ is singular if we do not impose a normalization on the location of $\alpha_i(y)$ and $\gamma_j(y)$.}
  The influence function of $\widehat F_k(y)$ is $\varphi_{ij,k}^y(\theta(y))$, where
$$
\varphi_{ij,k}^y(\theta) =  J_k(\theta)'  \psi_{ij}^y(\theta),
$$
and
$$
J_k(\theta) = \frac{1}{n} \sum_{(i,j)\in \mD}  \Lambda^{(1)}(\mathbbm{x}_{ij,k}'\beta + \alpha_i + \gamma_j) \mathbbm{w}_{ij,k}, \ \ \mathbbm{w}_{ij,k} = (\mathbbm{x}_{ij,k}, e_{i,I}, e_{j,J}).
$$

The standard error of $\widetilde \beta_{\ell}(y)$ is constructed as
\begin{equation}\label{eq:se_beta}
\widehat{\sigma}_{\beta_\ell}(y) =  n^{-1} \left[ \sum_{(i,j) \in \mD} \psi_{ij}^y(\widehat \theta(y)) \psi_{ij}^y(\widehat \theta(y))'\right]^{1/2}_{\ell,\ell},
\end{equation}
the square root of the $(\ell,\ell)$ element of the sandwich matrix $n^{-2} \sum_{(i,j) \in \mD} \psi_{ij}^y(\widehat \theta(y)) \psi_{ij}^y(\widehat \theta(y))'$. Similarly, the standard error of $\widetilde F_k(y)$ is constructed as
\begin{equation}\label{eq:se_F}
\widehat{\sigma}_{F_k}(y) =  n^{-1}  \left[\sum_{(i,j) \in \mD} \varphi_{ij,k}^y(\widehat \theta(y))^2 \right]^{1/2}.
\end{equation}


The following algorithm describes a multiplier bootstrap scheme to obtain the critical values for a set of parameters indexed by $\ell \in  \mB \subseteq \{1, \ldots, d_x\}$ and a set of distributions indexed by $k \in \mK \subseteq \{0,1\}$. This scheme is based on perturbing the first order conditions of the fixed effects estimators with random multipliers independent from the data.
\begin{algorithm}[Multiplier Bootstrap]\label{alg:mb} (1) Let $\bar \mY$ be some grid that satisfies the conditions of Remark \ref{remark:computation}. (2)  Draw the bootstrap multipliers $\{\omega_{ij}^m : (i,j) \in \mD \}$ independently from the data as  
$
\omega_{ij}^m = \tilde \omega_{ij}^m - \sum_{(i,j) \in \mD} \tilde  \omega_{ij}^m/n, \ \ \tilde \omega_{ij}^m \sim \text{ i.i.d. } \mathcal{N}(0,1)$.
Here we have normalized the multipliers to have zero mean as a finite-sample adjustment. 
(3) For each $y \in  \bar \mY$, obtain the bootstrap draws of $\widehat \theta(y)$ as $\widehat \theta^m(y) = \widehat \theta(y) + n^{-1} \sum_{(i,j) \in \mD} \omega_{ij}^m \psi_{ij}^y(\widehat \theta(y)),$ and of $\widehat F_k(y)$ as 
$
\widehat F_k^m(y) = \widehat F_k(y) + n^{-1} \sum_{(i,j) \in \mD} \omega_{ij}^m \varphi_{ij,k}^y(\widehat \theta(y)),$ $k \in \mK.$
(4) Construct the bootstrap draw of the maximal t-statistic for the parameters,
$t^{(\beta),m}_{\mB, \bar{\mY}} = \max_{y \in \bar{\mY}, \ell \in  \mB} |\widehat \beta^m_{\ell}(y) - \widehat \beta_{\ell}(y)| / \widehat \sigma_{\beta_{\ell}}(y)$, where $\widehat \sigma_{\beta_{\ell}}(y)$ is defined in \eqref{eq:se_beta}, 
and $\psi_{ij,\ell}^y(\theta)$ is the component of $\psi_{ij}^y(\theta)$ corresponding to $\beta_{\ell}$.
Similarly, construct the bootstrap draw of the maximal t-statistic for the distributions,
$
t_{\mK, \bar{\mY}}^{(F),m} = \max_{y \in \bar{\mY}, k \in \mK} |\widehat F_k^m(y) - \widehat F_k(y)|/\widehat \sigma_{F_k}(y),
$
where $\widehat \sigma_{F_k}(y)$ is defined in \eqref{eq:se_F}.
(5) Repeat steps (1)--(3) $M$ times and index the bootstrap draws by $m \in \{1, \ldots, M\}$. In the numerical examples we set $M = 500$. 
(6) Obtain the bootstrap estimators of the critical values as   
\begin{eqnarray*}
\widehat{t}^{(\beta)}_{\mB,  \mY}(p) &=&  p-\text{quantile of }  \{t_{\mB, \bar{\mY}}^{(\beta),m} : 1 \leq m \leq M\}, \\
\widehat{t}^{(F)}_{\mK,  \mY}(p) &=&  p-\text{quantile of }  \{t_{\mK, \bar{\mY}}^{(F),m} : 1 \leq m \leq M\}.
\end{eqnarray*}
\end{algorithm}

The next result shows that the multiplier bootstrap provides consistent estimators of the critical values of the inferential statistics.  The proof follows from Theorem 2.2 of \cite{cck-16}.

\begin{theorem}[Consistency of Multiplier Bootstrap Inference]
    \label{cor:bias_correction}
    Let Assumption~\ref{ass:baseline} hold.
 Then, conditional on the data $\{(y_{ij}, x_{ij}) : (i,j) \in \mD \}$,  as $n \to \infty$ and $M \to \infty$
$$
       \widehat{t}^{(\beta)}_{\mB,  \mY}(p) \to_{\Pr}  t^{(\beta)}_{\mB,  \mY}(p)  \ \text{ and } \
       \widehat{t}^{(F)}_{\mK,  \mY}(p) \to_{\Pr} t^{(F)}_{\mK,  \mY}(p),
    $$
  where $t^{(\beta)}_{\mB,  \mY}(p)$
    and $t^{(F)}_{\mK,  \mY}(p)$ 
    are defined in \eqref{eq:maxt-coeff} and \eqref{eq:maxt-dist}, respectively.
\end{theorem}

 Theorem \ref{cor:bias_correction} together with Theorem  \ref{th:expansion} 
guarantee the asymptotic validity of the confidence bands $I_{\beta}$ and $I_F$ defined in \eqref{ConfBandCoeff} and \eqref{ConfBandCDF} with the critical values $t^{(\beta)}_{\mB,  \mY}(p)$ and  $t^{(F)}_{\mK,  \mY}(p)$ replaced by the bootstrap estimators $\widehat{t}^{(\beta)}_{\mB,  \mY}(p)$ and $\widehat{t}^{(F)}_{\mK,  \mY}(p)$.

\subsection{Pairwise Clustering Dependence or Reciprocity}\label{subsec:clustering} The conditional independence of Assumption~\ref{ass:baseline}(i) can be relaxed to allow for some forms of conditional weak dependence. A form of dependence that is  relevant for network data is \textit{pairwise clustering} or \textit{reciprocity} where the observational units with symmetric indexes $(i,j)$ and $(j,i)$ might be dependent due to unobservable factors not accounted by unobserved effects.\footnote{\cite{CM-14} consider other patterns of dependence in linear models for dyadic data.}  In the trade application, for example, these factors may include distributional channels or multinational firms operating in both countries. Formally, pairwise clustering means that  $(y_{ij},y_{ji})$ is independently distributed across $(i,j) \in \mD$ with $i \leq j$,  conditional on all the observed and unobserved covariates $ \mC_B:= \{(x_{ij}, v_i, w_j): (i,j) \in \mD\}$.

The presence of reciprocity does not change the bias of the fixed effects estimators, but affects the standard errors and the implementation of the multiplier bootstrap. The standard error of $\widetilde \beta_{\ell}(y)$ becomes
\begin{equation}\label{eq:cse_beta}
\widehat{\sigma}_{\beta_\ell}(y) =  n^{-1} \left[ \sum_{(i,j) \in \mD}  \left\{ \psi_{ij}^y(\widehat \theta(y)) +   \psi_{ji}^y(\widehat \theta(y)) \right\} \psi_{ij}^y(\widehat \theta(y))'  \right]^{1/2}_{\ell,\ell}.
\end{equation}
Similarly, the standard error of $\widetilde F_k(y)$ needs to be adjusted to
\begin{equation}\label{eq:cse_F}
\widehat{\sigma}_{F_k}(y) =  n^{-1}  \left[\sum_{(i,j) \in \mD} \left\{ \varphi_{ij,k}^y(\widehat \theta(y)) + \varphi_{ji,k}^y(\widehat \theta(y)) \right\} \varphi_{ij,k}^y(\widehat \theta(y))  \right]^{1/2}.
\end{equation}
In the previous expressions  we assume that if $(i,j) \in \mD$ then $(j,i) \in \mD$ to simplify the notation. The modified multiplier bootstrap algorithm becomes:

\begin{algorithm}[Multiplier Bootstrap with Pairwise Clustering]\label{alg:cmb} (1)  Let $\bar \mY$ be some grid that satisfies the conditions of Remark \ref{remark:computation}.  (2)  Draw the bootstrap multipliers $\{\omega_{ij}^m : (i,j) \in \mD \}$ independently from the data as   
$
\omega_{ij}^m = \tilde \omega_{ij}^m - \sum_{(i,j) \in \mD} \tilde  \omega_{ij}^m/n$, $\tilde \omega_{ij}^m \sim \text{ i.i.d. } \mathcal{N}(0,1)$ if  $i \leq j$, and  $\tilde \omega_{ij}^m = \tilde \omega_{ji}^m$ if $i > j$.
(3) For each $y \in \bar \mY$, obtain the bootstrap draws of $\widehat \theta(y)$ as $\widehat \theta^m(y) = \widehat \theta(y) + n^{-1} \sum_{(i,j) \in \mD} \omega_{ij}^m \psi_{ij}^y(\widehat \theta(y)),$ and of $\widehat F_k(y) $ as 
$
\widehat F_k^m(y) = \widehat F_k(y) + n^{-1} \sum_{(i,j) \in \mD} \omega_{ij}^m \varphi_{ij,k}^y(\widehat \theta(y)),$ $k \in \mK.$
(4) Construct the bootstrap draw of the maximal t-statistic for the parameters,
$t^{(\beta),m}_{\mB, \bar{\mY}} = \max_{y \in \bar{\mY}, \ell \in  \mB} |\widehat \beta^m_{\ell}(y) - \widehat \beta_{\ell}(y)| / \widehat \sigma_{\beta_{\ell}}(y)$, where $\widehat \sigma_{\beta_{\ell}}(y)$ is defined in \eqref{eq:cse_beta}, and and $\psi_{ij,\ell}^y(\theta)$ is the component of $\psi_{ij}^y(\theta)$ corresponding to $\beta_{\ell}$.
Similarly, construct the bootstrap draw of the maximal t-statistic for the distributions,
$
t_{\mK, \bar{\mY}}^{(F),m} = \max_{y \in \bar{\mY}, k \in \mK} |\widehat F_k^m(y) - \widehat F_k(y)|/\widehat \sigma_{F_k}(y),
$
where $\widehat \sigma_{F_k}(y)$ is defined in \eqref{eq:cse_F}. 
(5) Repeat steps (1)--(3) $M$ times and index the bootstrap draws by $m \in \{1, \ldots, M\}$. In the numerical examples we set $M = 500$. 
(6) Obtain the bootstrap estimators of the critical values as   
\begin{eqnarray*}
\widehat{t}^{(\beta)}_{\mB,  \mY}(p) &=&  p-\text{quantile of }  \{t_{\mB, \bar{\mY}}^{(\beta),m} : 1 \leq m \leq M\}, \\
\widehat{t}^{(F)}_{\mK,  \mY}(p) &=&  p-\text{quantile of }  \{t_{\mK, \bar{\mY}}^{(F),m} : 1 \leq m \leq M\}.
\end{eqnarray*}
\end{algorithm}

The clustered multiplier bootstrap preserves the dependence in the symmetric pairs $(i,j)$ and $(j,i)$ by assigning the same multiplier to each of these pairs. 

\subsection{Average Effect} 
A bias corrected estimator of the average effect can be formed as
\begin{equation}\label{eq:DRae}
\widetilde \Delta = \widetilde \mu_1 - \widetilde \mu_0,
\end{equation}
where
$$
\widetilde \mu_k = \int [1(y \geq 0) - \mathbf{C} \widetilde F_k(y)]  dy, \ \ k \in \{0,1\}.
$$
Here the integral is over the real line, and $\mathbf{C}$ is an operator that extends $\widetilde F_k(y)$ from $\mY$ to $\mathbb{R}$ as a step function, that is, it maps any $f: \mY \to \mathbb{R}$ to $\mathbf{C} f: \mathbb{R} \to \mathbb{R}$, where
$\mathbf{C} f(y) = 0$ for $y\leq \inf \mY$, $\mathbf{C} f(y) = 1$ for $y\geq \sup \mY$,
and $\mathbf{C} f(y) = f(\sup\{ y' \in \mY : y' \leq y \})$ otherwise.
The following central limit theorem for the bias corrected estimator of the average effect  is a corollary of 
 Theorem~\ref{th:expansion}
together with the functional delta method.

\begin{corollary}[CLT for Bias Corrected Fixed Effects Estimators of Average Effect]
    \label{cor:bc2}
    Let Assumption~\ref{ass:baseline} hold and $\int_{\mY}  dF_k(y) = 1,$  $k \in \{0,1\}$.
 Then, 
   \begin{equation}\label{eq:ae2}
       \sqrt{n}\left( \widetilde \Delta -  \Delta
      \right)           
          \to_d - \int  \left[  \mathbf{C} Z_1^{(F)}(y) -  \mathbf{C} Z_0^{(F)}(y)\right] dy =: Z^{(\Delta)},
    \end{equation}
 where $Z^{(F)}(y) = [ Z_0^{(F)}(y),  Z_1^{(F)}(y)]'$ is the same Gaussian process
    that appears in Theorem~\ref{th:expansion} with $\mK = \{0,1\}$.
\end{corollary}

\begin{remark}[Support of $Y$] The condition that $\int_{\mY}  dF_k(y) = 1$ guarantees that $\mY$ is the support of the potential outcome corresponding to the distribution $F_k$, so that  \eqref{eq:ae} yields the average potential outcome under $F_k$. Together with Assumption~\ref{ass:baseline}, this condition is satisfied when $Y$ is discrete with finite support $\mY$, or continuous  or mixed with bounded support $\mY$ and conditional density bounded away from zero in the interior of $\mY$.  This support condition is not required for the estimation of the  quantile effects.
\end{remark}

We can construct confidence intervals for the average effect using Corollary \ref{cor:bc2}. Let 
$$
\widehat{\sigma}_{\Delta} =  n^{-1}  \left[\sum_{(i,j) \in \mD} \widehat \varphi_{ij}^2 \right]^{1/2}, \ \ \widehat \varphi_{ij} = - \int  \left[\mathbf{C}\varphi_{ij,1}^y(\widehat \theta(y)) - \mathbf{C}\varphi_{ij,0}^y(\widehat \theta(y)) \right] dy.
$$
 Then, $\widehat{\sigma}_{\Delta}$ is  an estimator of $\sigma_{\Delta}$, the standard deviation of the limit process $Z^{(\Delta)}$ in \eqref{eq:ae2}, and
$$
I_{\Delta} = [\widetilde{\Delta} \pm \Phi^{-1}(1-p/2) \widehat{\sigma}_{\Delta}],
$$
is an asymptotic $p$-confidence interval for  $\Delta$.  The normal critical value $\Phi^{-1}(1-p/2)$ can be replaced by a multiplier bootstrap critical value $\widehat{t}^{(\Delta)}(p)$ obtained from Algorithm \ref{alg:mb} as
$$
\widehat{t}^{(\Delta)}(p) = p-\text{quantile of }  \{t^{(\Delta),m} : 1 \leq m \leq M\}
$$
where $
t^{(\Delta),m} = |\widehat \Delta^m- \widehat \Delta|/\widehat \sigma_{\Delta}
$ and
$
\widehat \Delta^m = \widehat \Delta  + n^{-1} \sum_{(i,j) \in \mD}  \omega_{ij}^m \widehat \varphi_{ij}.
$

The standard errors and critical values of the average effects can be adjusted to account for pairwise clustering following the  procedure described in Section \ref{subsec:clustering}. Thus, the pairwise clustering robust standard error is
$$
\widehat{\sigma}_{\Delta} =  n^{-1}  \left[\sum_{(i,j) \in \mD} \left\{ \widehat \varphi_{ij} + \widehat \varphi_{ji}\right\} \widehat \varphi_{ij} \right]^{1/2}.
$$

\section{Quantile Effects in Gravity Equations for International Trade}\label{sec:trade}
We consider an empirical application to gravity equations for bilateral trade between countries. We use data from \cite{Helpman01052008}, extracted from the Feenstra's World Trade Flows, CIA's World Factbook and Andrew Rose's web site. These data contain information on bilateral trade flows and other trade-related variables for 157 countries in 1986.\footnote{The original data set includes 158 countries. We exclude Congo because it did not export to any other country in 1986.} The data set contains network data where both $i$ and $j$ index countries as senders (exporters) and receivers (importers), and therefore $I = J = 157$. The outcome $y_{ij}$ is the volume of trade in thousands of constant 2000 US dollars from country $i$ to country $j$, and the covariates $P(x_{ij}) = x_{ij}$ include determinants of bilateral trade flows such as the logarithm of the distance in kilometers between country  $i$'s capital and country $j$'s capital and  indicators for common colonial ties, currency union, regional free trade area (FTA),  border,  legal system,  language, and religion. Following \cite{AndersonWincoop2003}, we include unobserved importer and exporter country effects.\footnote{See \cite{Harrigan1994} for an earlier empirical international trade application that includes unobserved country effects.} These effects control for other country specific characteristics that may affect trade such as GDP, tariffs,  population, institutions, infrastructures or natural resources. We allow for these characteristics to affect differently the imports and exports of each country, and be arbitrarily related with the observed covariates.

Table \ref{table:ds} reports descriptive statistics of the variables used in the analysis. There are $157 \times 156 = 24,492$ observations corresponding to different pairs of countries. The observations with $i = j$ are missing because we do not observe trade flows from a country to itself.  The trade variable in the first row is an indicator for positive volume of trade. There are no  trade flows for  55\% of the country pairs.  The volume of trade variable exhibits much larger standard deviation than the mean. Since this variable is bounded below at zero, this indicates the presence of a very heavy upper tail in the distribution. This feature also makes quantile methods specially well-suited for this application on robustness grounds.\footnote{In results not reported, we find that estimates of average effects are very sensitive to the trimming of outliers at the top of the distribution.}

\begin{table}[ht]\caption{Descriptive Statistics}\label{table:ds}
\centering
\begin{tabular}{lcc}
  \hline\hline
 & \text{Mean} & \text{Std. Dev.} \\ 
  \hline
  Trade & 0.45 & 0.50 \\ 
  Trade Volume & 84,542 & 1,082,219 \\ 
  Log Distance & 4.18 & 0.78 \\ 
  Legal & 0.37 & 0.48 \\ 
  Language & 0.29 & 0.45 \\ 
  Religion & 0.17 & 0.25 \\ 
  Border & 0.02 & 0.13 \\ 
  Currency & 0.01 & 0.09 \\ 
  FTA & 0.01 & 0.08 \\ 
  Colony & 0.01 & 0.10 \\ \hline
  \multicolumn{1}{l}{Country Pairs} &  \multicolumn{2}{c}{24,492} \\ 
   \hline\hline
   \multicolumn{3}{l}{\footnotesize{Source: Helpman, Melitz and Rubinstein (08) }}
\end{tabular}
\end{table}

The previous  literature  estimated nonlinear parametric models such as Poisson, Negative Binomial, Tobit and Heckman-selection models to deal with the large number of zeros in the volume of trade (e.g., \citealp{EatonKortum2001}, \citealp{SantosSilvaTenreyro2006}, and \citealp{Helpman01052008}).\footnote{See \cite{HeadMayer2014} for a recent survey on gravity equations in international trade.} These models impose strong conditions on the process that generates the zeros and/or on the conditional heteroskedasticity of the volume of trade.  The DR model deals with zeros and any other fixed censoring points in a very flexible and natural fashion as it specifies the conditional distribution separately at the mass point. In particular, the model coefficients at zero can be arbitrarily different from the model coefficients at other values of the volume of trade.   Moreover, the DR model can also accommodate   conditional heteroskedasticity.

Figure \ref{fig:coeffs} shows estimates and 95\% pointwise confidence intervals for the DR coefficients of log distance and common legal system plotted against the quantile indexes of the volume of trade.  We report uncorrected and bias corrected fixed effects estimates obtained from \eqref{fe-cmle} and \eqref{eq:bc}, respectively. The confidence intervals are constructed using \eqref{eq:ci}. The x-axis starts at .54, the maximum quantile index corresponding to zero volume of trade. The region  of interest $\mY$ corresponds to the interval between zero and the $0.95$-quantile of the volume of trade. 
 The difference between the uncorrected and bias corrected estimates is the same order of magnitude as the width of the confidence intervals for the coefficient of log distance. We find the largest estimated biases for both coefficients at highest quantiles of the volume of trade, where the  indicators $1\{y_{ij} \leq y\}$ take on many ones. The signs of the DR coefficients indicate that increasing distance has a negative effect and having a common legal system has a positive effect on the volume of trade throughout the distribution. Recall that the sign of the effect in terms of volume of trade, $y_{ij},$ is the opposite to the sign of the DR coefficient. 

    \begin{figure}[h!]
        \begin{center}
            \includegraphics[width=\textwidth,height= .5\textwidth,angle=0]{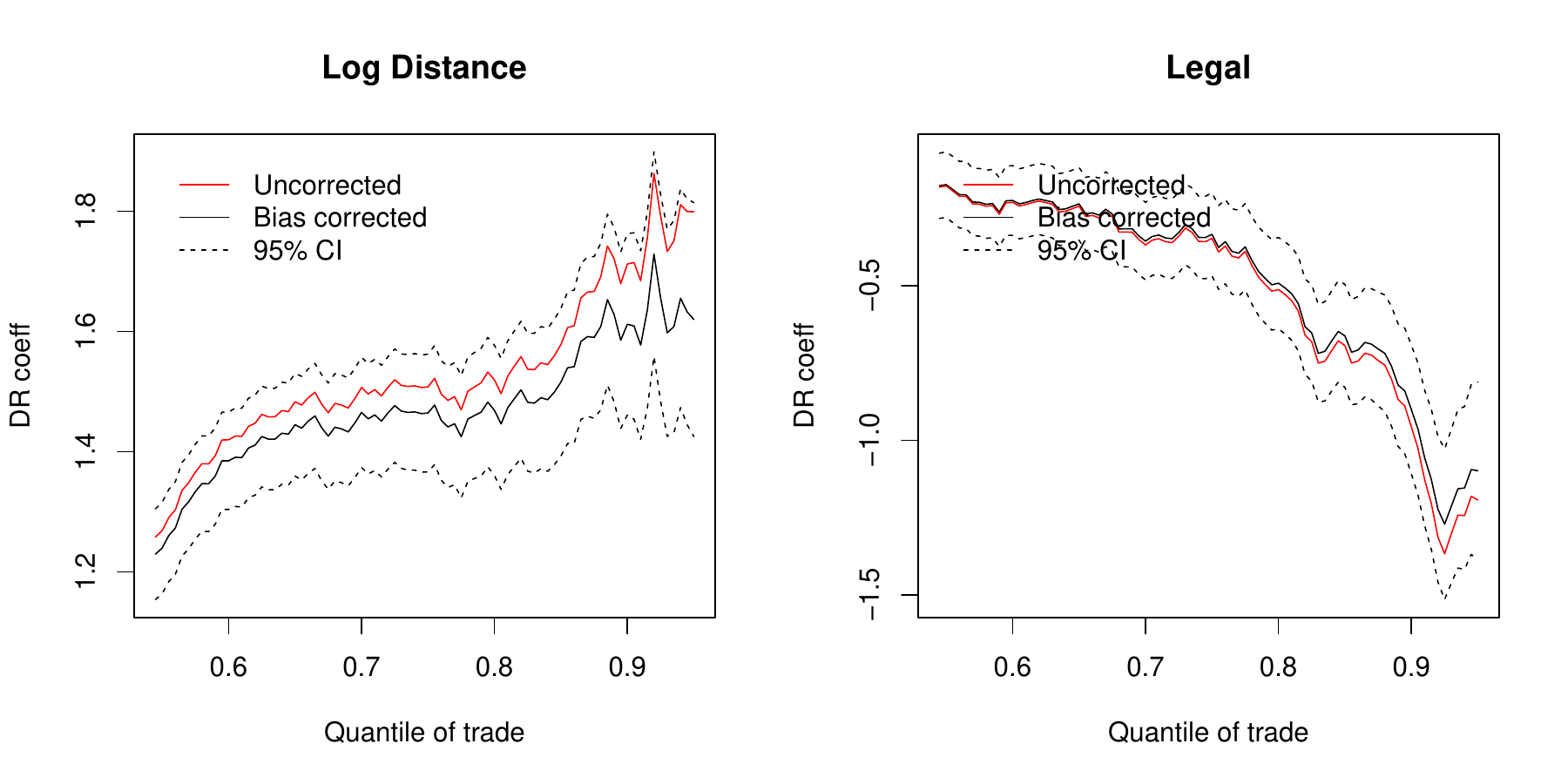}
        \end{center}\caption{Estimates and 95\% pointwise confidence intervals for the DR-coefficients of log distance and common legal system.}\label{fig:coeffs}
    \end{figure}

Figures \ref{fig:dist} and \ref{fig:quant} show estimates and 95\% uniform confidence bands for distribution and quantile functions of the volume of trade at different values of the log of distance and the common legal system. The left panels plot the functions when  distance takes the observed levels ($\text{dist}$) and two times the observed values $(\text{2*dist})$, i.e. when we counterfactually double all the distances between the countries. The right panels plot the functions when all the countries have the same legal system (legal=1) and different systems (legal=0). 
The confidence bands for the distribution are obtained by Algorithm \ref{alg:mb} with 500 bootstrap replications and standard normal multipliers, and a grid of values $\bar \mY$ that includes the sample quantiles of the volume of trade with indexes $\{.54, .55, \ldots, .95\}$. The bands are joint for the two functions displayed in each panel. The confidence bands for the quantile functions are obtained by inverting and rotating the bands for the corresponding distribution functions using Lemma \ref{theorem:bquant}.

    \begin{figure}[h!]
        \begin{center}
            \includegraphics[width=.48\textwidth,height= .48\textwidth,angle=0,page=1]{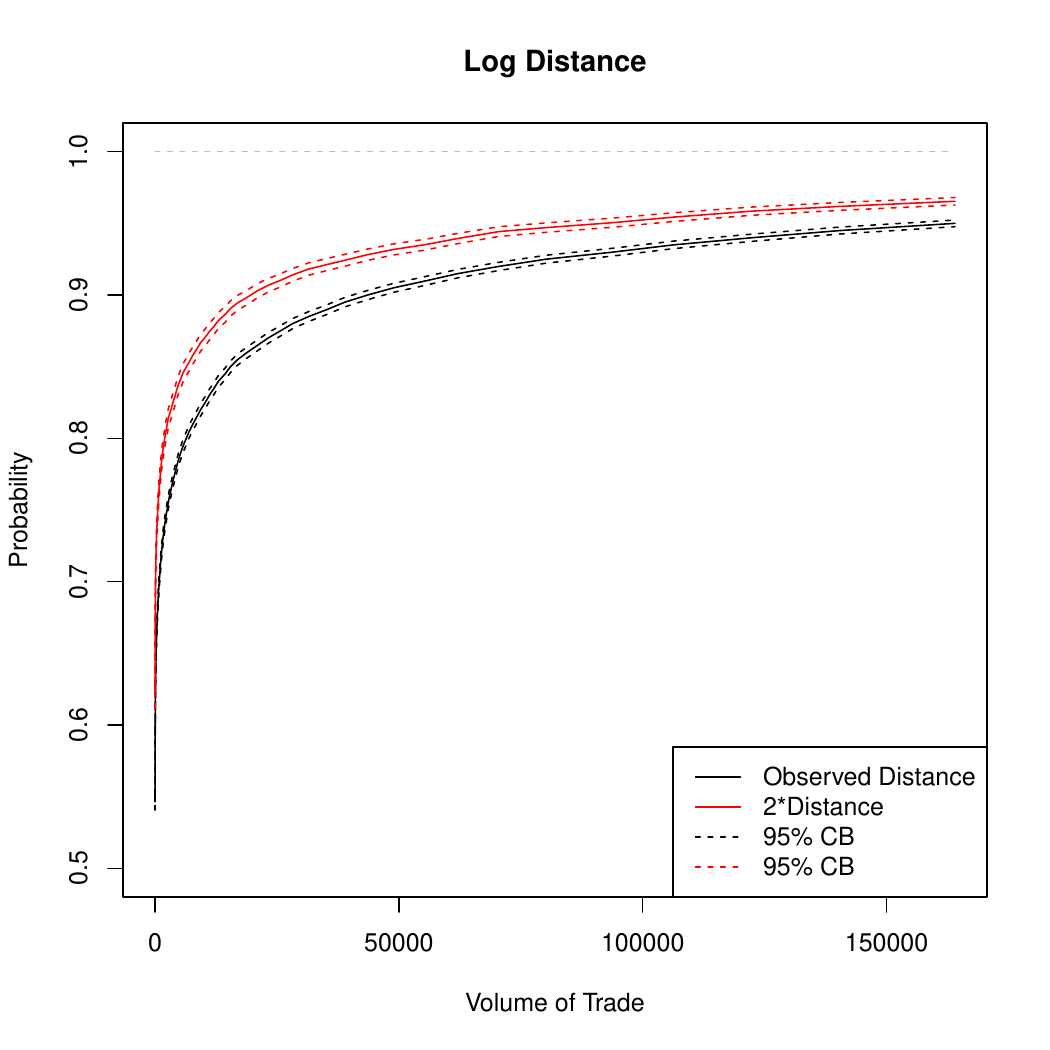}
            \includegraphics[width=.48\textwidth,height= .48\textwidth,angle=0,page=2]{Trade_dfqf.pdf}
        \end{center}\caption{Estimates and 95\% uniform confidence bands for distribution functions of the volume of trade.}\label{fig:dist}
    \end{figure}

    \begin{figure}[h!]
        \begin{center}
            \includegraphics[width=.48\textwidth,height= .48\textwidth,angle=0,page=3]{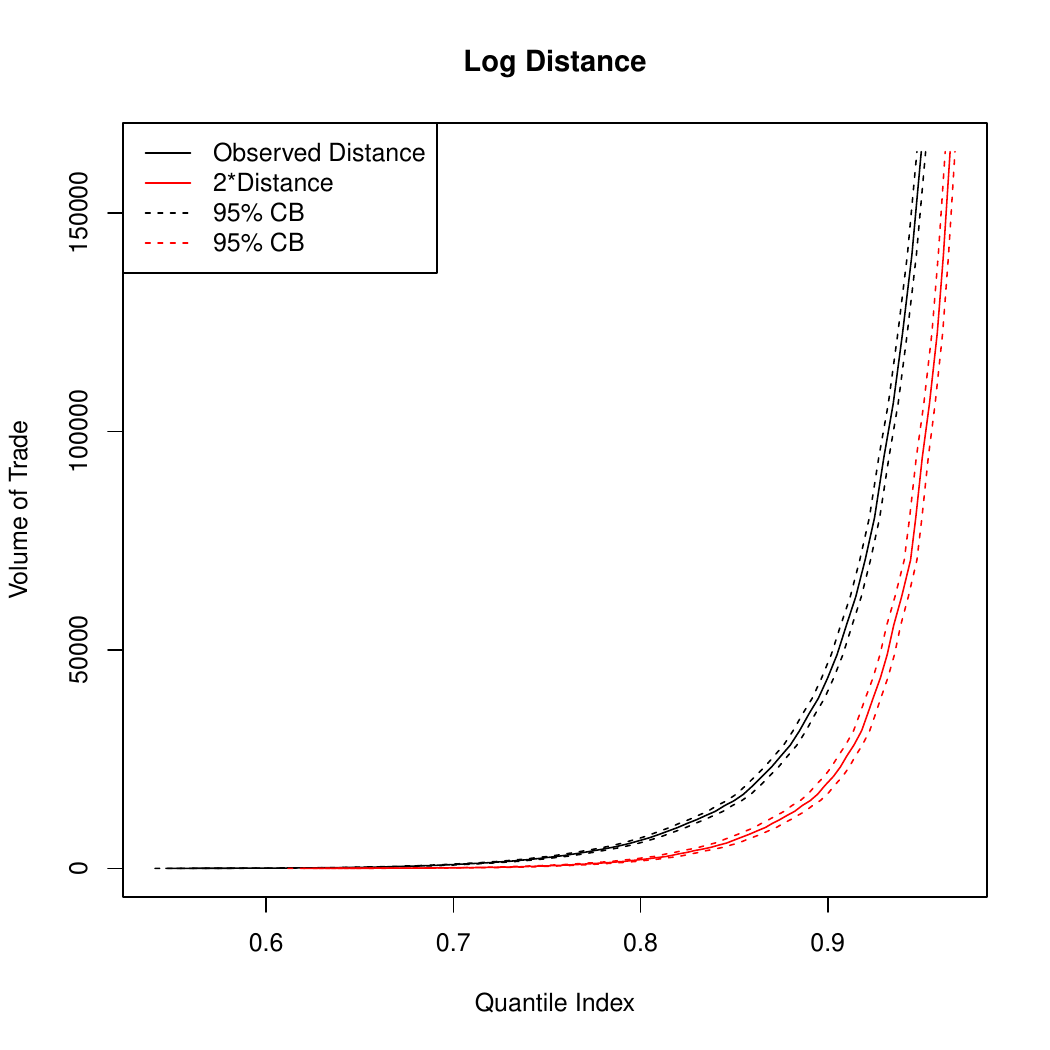}
            \includegraphics[width=.48\textwidth,height= .48\textwidth,angle=0,page=4]{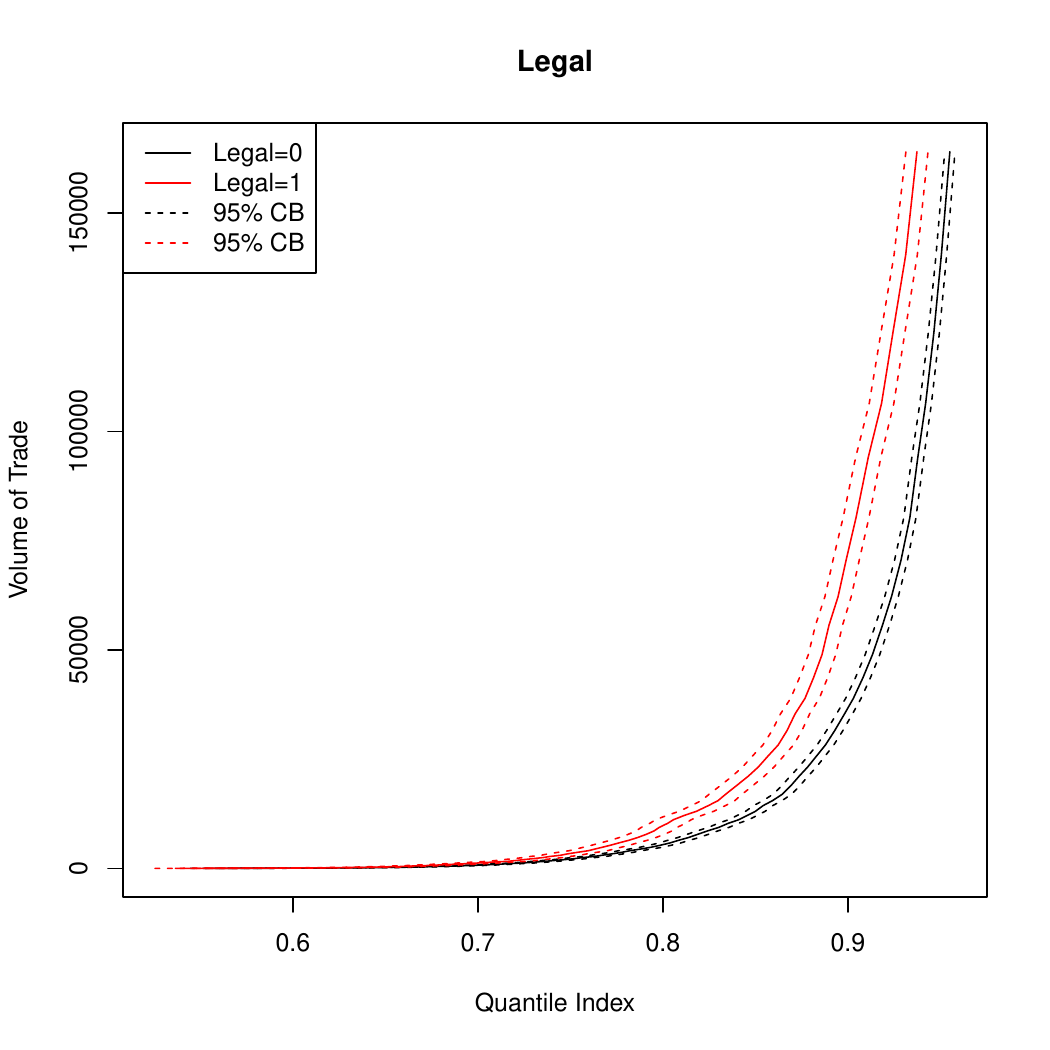}
        \end{center}\caption{Estimates and 95\% uniform confidence bands for quantile functions of the volume of trade.}\label{fig:quant}
    \end{figure}
    
Figure \ref{fig:qte} displays estimates and 95\% uniform confidence bands for the quantile effects of the log of distance and the common legal system on the volume of trade, constructed using Lemma  \ref{theorem:bqte}. For comparison, we also include estimates from a Poisson model. Here,  we replace the  DR estimators of the distributions by
\begin{equation}\label{eq:dpoisson}
\widehat F_k(y) = \frac{1}{n} \sum_{(ij) \in \mD} \exp \lambda_{ij,k} \sum_{\tilde y = 0}^{\lfloor y \rfloor} \frac{\lambda_{ij,k}^{\tilde y}}{\tilde y!}, \ \ k \in \mK,
\end{equation}
where $\lfloor y \rfloor$ is the integer part of $y$, $\lambda_{ij,k} = \exp(\mathbbm{x}_{ij,k}'\widehat \beta + \widehat \alpha_i + \widehat \gamma_j)$, and $\widehat \theta = (\widehat \beta, \widehat \alpha_1, \ldots, \widehat \alpha_I, \widehat \gamma_1, \dots, \widehat \gamma_J)$ is the Poisson fixed effects conditional maximum likelihood estimator
$$
\widehat \theta \in \arg \max_{\theta \in \mathbb{R}^{d_x + I + J}} \sum_{(ij) \in \mD} [y_{ij} (x_{ij}'\beta + \alpha_i + \gamma_j) - \exp (x_{ij}'\beta + \alpha_i + \gamma_j)].
$$
We find that distance and common legal system have heterogeneously increasing effects along the distribution. For example, the negative effects of doubling the distance grows more than proportionally as we move up to the upper tail of the distribution of volume of trade. Putting all the countries under the same legal system has little effects in the extensive margin of trade, but has a strong positive effect at the upper tail of the distribution. The Poisson estimates lie outside the DR confidence bands  reflecting heavy tails in the conditional distribution of the volume of trade that is missed by the Poisson model.\footnote{This misspecification problem with the Poisson model is well-known in the international trade literature. The Poisson estimator is  treated as a quasi-likelihood estimator and standard errors robust to misspecification are reported (\citealp{SantosSilvaTenreyro2006}).} Figure \ref{fig:cqte} shows confidence bands of the quantile effects  that account for pairwise clustering. The bands are constructed from confidence bands from the distributions using Algorithm \ref{alg:cmb} with $500$ bootstrap draws and standard normal multipliers. Accounting for unobservables that affect symmetrically to the country pairs has very little effect on the width of the bands in this case.

    \begin{figure}[h!]
        \begin{center}
            \includegraphics[width=.48\textwidth,height= .48\textwidth,angle=0,page=1]{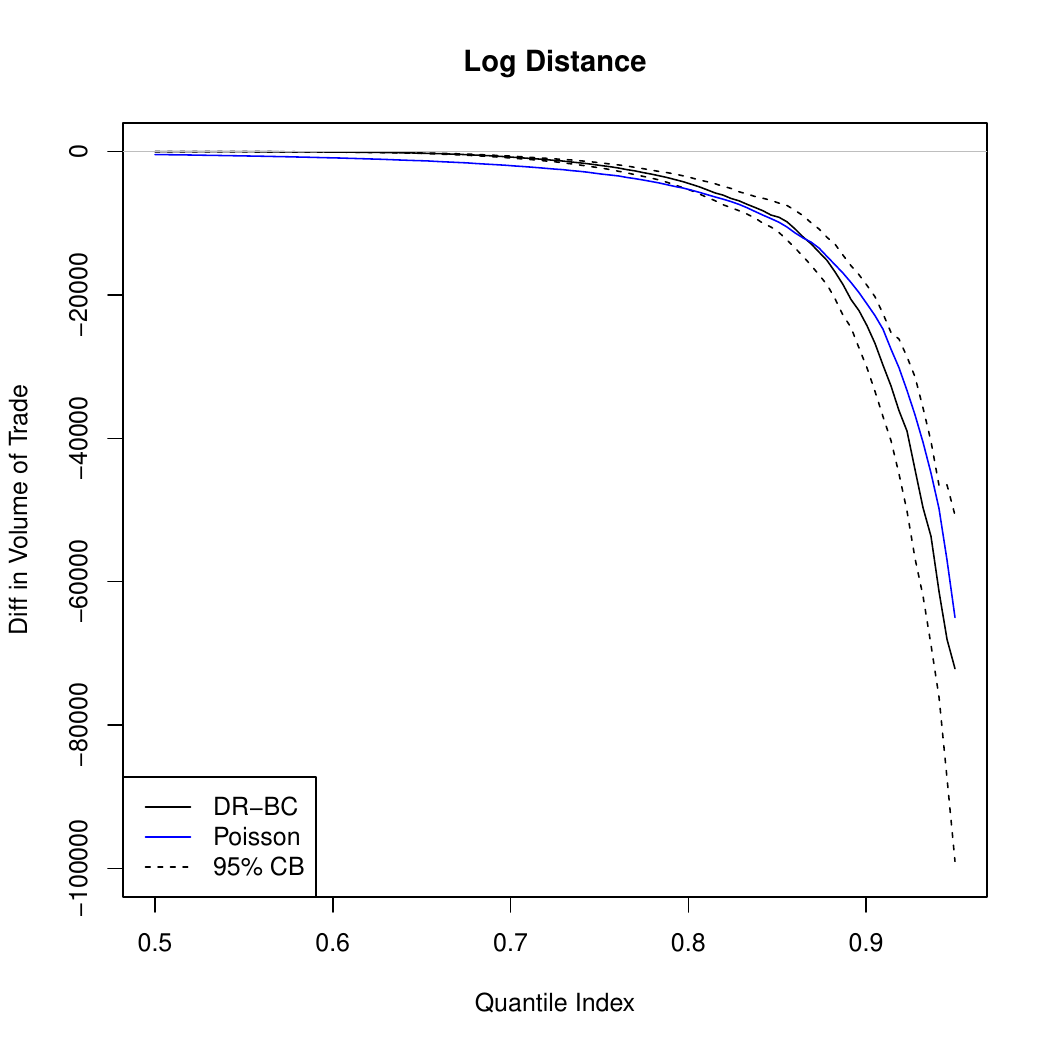}
            \includegraphics[width=.48\textwidth,height= .48\textwidth,angle=0,page=2]{Trade_qef.pdf}
        \end{center}\caption{Estimates and 95\% uniform confidence bands for the quantile effects of log distance and common legal system on the volume of trade.}\label{fig:qte}
    \end{figure}

    \begin{figure}[h!]
        \begin{center}
            \includegraphics[width=.48\textwidth,height= .48\textwidth,angle=0,page=3]{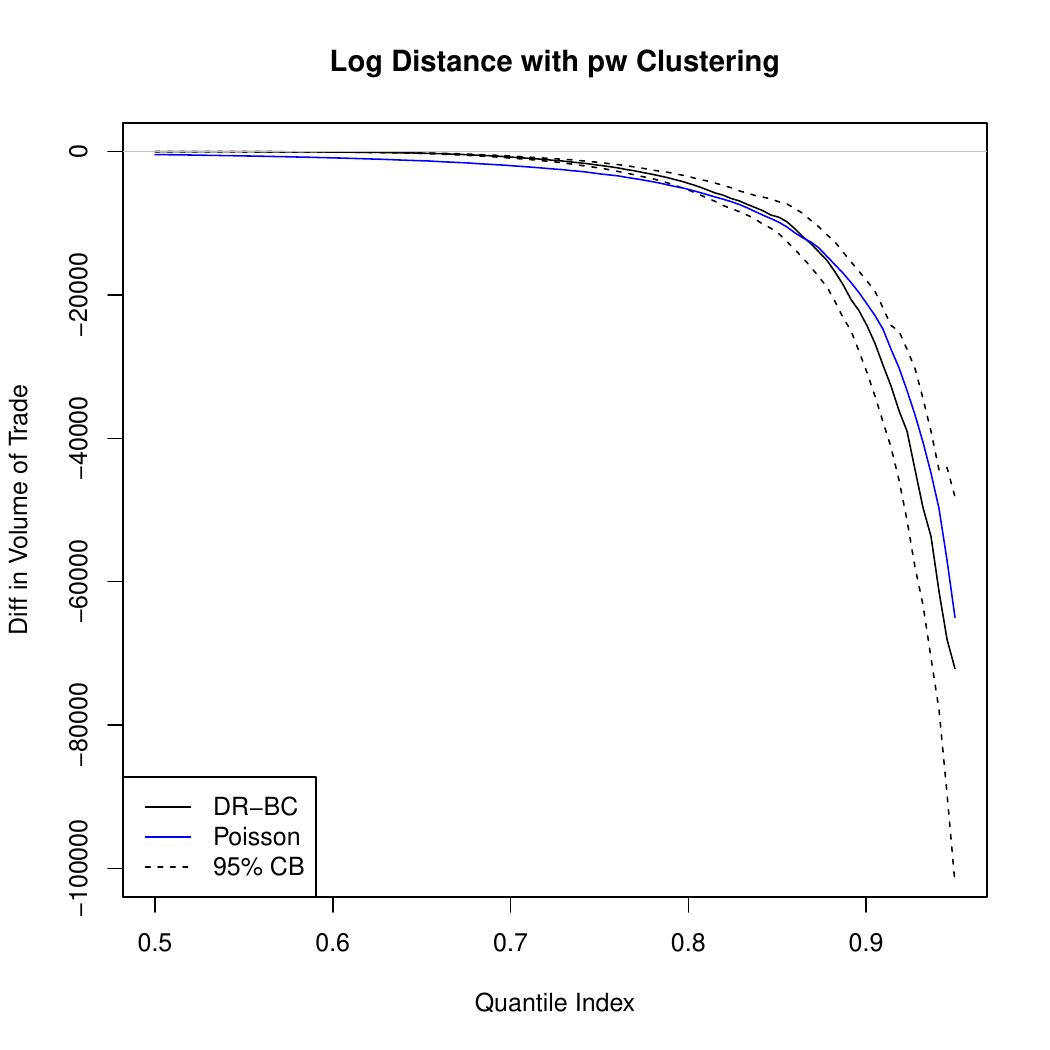}
            \includegraphics[width=.48\textwidth,height= .48\textwidth,angle=0,page=4]{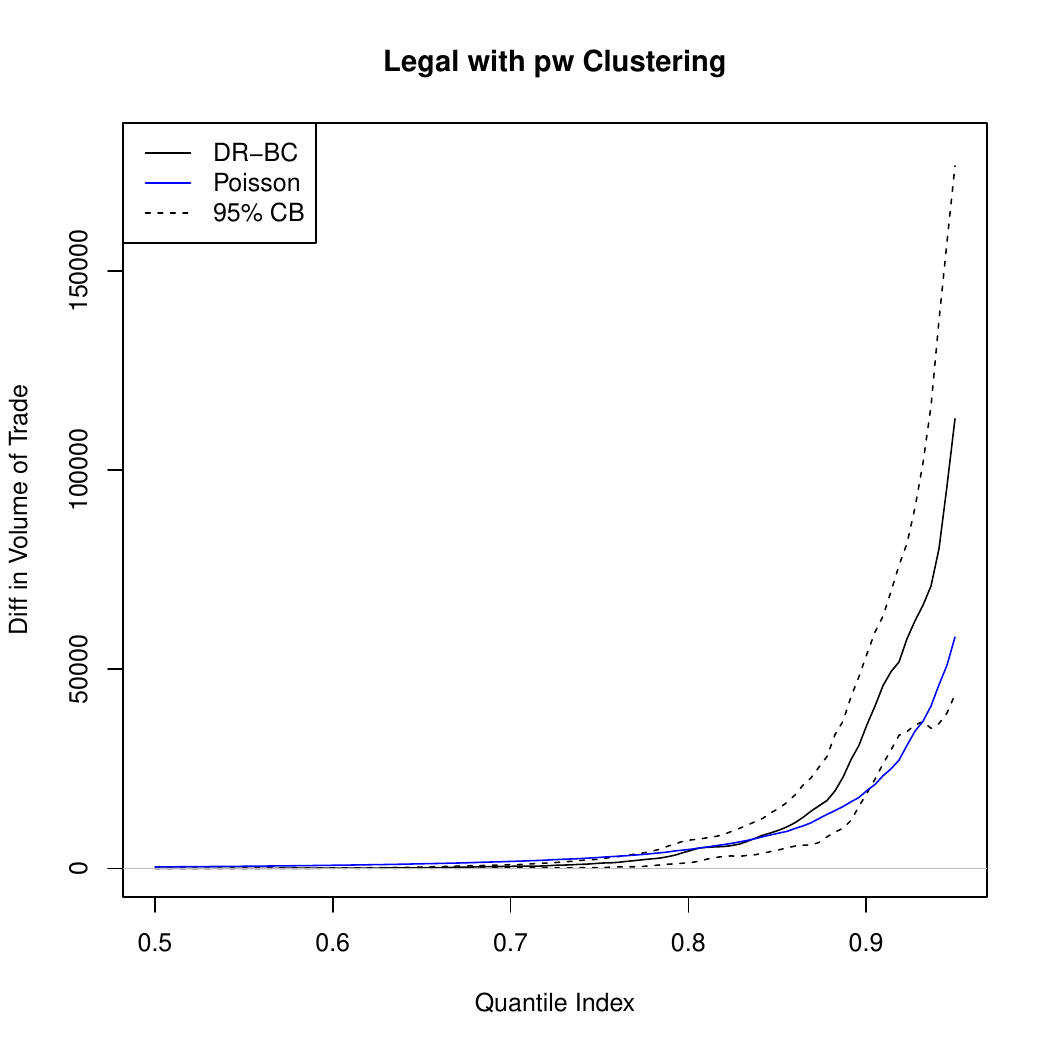}
        \end{center}\caption{Estimates and 95\% uniform confidence bands for the quantile effects of log distance and common legal system on the volume of trade.}\label{fig:cqte}
    \end{figure}

\section{Montecarlo Simulation}\label{sec:mc}

We conduct a Montecarlo simulation calibrated to the empirical application of Section \ref{sec:trade}. The outcome is generated by the censored logistic process
$$
y^s_{ij} = \max\{x_{ij}'\widehat \beta + \widehat \alpha_i + \widehat \gamma_{j} + \widehat{\sigma} \Lambda^{-1}(u^s_{ij})/\sigma_{L}, 0 \},  \ \ (i,j) \in \mD,
$$
where $\mD = \{(i,j): 1 \leq i,j \leq 157, i \neq j\}$, $x_{ij}$ is the value of the covariates for the observational unit $(i,j)$ in the trade data set, $\sigma_L = \pi/\sqrt{3}$, the standard deviation of the logistic distribution, and $(\widehat \beta,  \widehat \alpha_1, \ldots, \widehat \alpha_I, \widehat \gamma_{1}, \ldots, \widehat \gamma_{J}, \widehat \sigma)$ are Tobit fixed effect estimates of the parameters in the trade data set with lower censoring point at zero.\footnote{We upper winsorize the volume of trade $y_{ij}$ at the $95.5\%$ quantile to reduce the effect of outliers in the Tobit estimation of the parameters.} We consider two designs: independent errors with $u^s_{ij} \sim \text{ i.i.d } \mathcal{U}(0,1),$ and pairwise dependent errors with $u^s_{ij} = \Phi(0.75 e^s_{ij} + \sqrt{1-0.75^2} e^s_{ji}),$ where $e^s_{ij} \sim \text{ i.i.d } \mathcal{N}(0,1)$ and $\Phi$ is the standard normal CDF.\footnote{The Spearman rank correlation between $u^s_{ij}$ and $u^s_{ji}$ in the design with pairwise-dependent errors is $0.73$.} In both cases the conditional distribution function of $y_{ij}^s$ is a special case of the DR  model \eqref{eq:pdr} with link function $\Lambda_y = \Lambda$, the logistic distribution, for all $y$,
$$
\beta(y) = \sigma_L (e_1 y - \widehat \beta)/\widehat \sigma, \ \  \alpha_i(y) = -\sigma_L \widehat \alpha_i/\widehat \sigma, \ \text{and}  \ \gamma_j(y) = - \sigma_L \widehat \gamma_j/\widehat \sigma,
$$
where $e_1$ is a unit vector of dimension $d_x$ with a one in the first component. As in the empirical application, the region  of interest $\mY$ is the interval between zero and the $0.95$-quantile of the volume of trade in the data set. All the results are based on 500 simulated panels  $\{(y^s_{ij}, x_{ij}) : (i,j) \in \mD\}$.

Figures \ref{fig:mc} and \ref{fig:mc2} report the biases, standard deviations and root mean square errors (rmses) of the fixed effects estimators of the DR coefficients of log-distance and legal system as a function of the quantiles of $y_{ij}$ in the design with independent errors.\footnote{The design with pairwise dependent errors produces similar results, which are not reported for the sake of brevity.}   All the results are in percentage of the true value of the parameter. As predicted by the large sample theory, the fixed effects estimator displays a bias of the same order of magnitude as the standard deviation. As in fig. \ref{fig:coeffs}, the bias is more severe for the coefficient of log distance. The bias correction removes most of the bias and does not increase the standard deviation, yielding a reduction in rmse of about 5\% for the coefficient of log distance at the highest quantile indexes.

    \begin{figure}[h!]
        \begin{center}
            \includegraphics[width=.32\textwidth,height= .5\textwidth,angle=0,page=1]{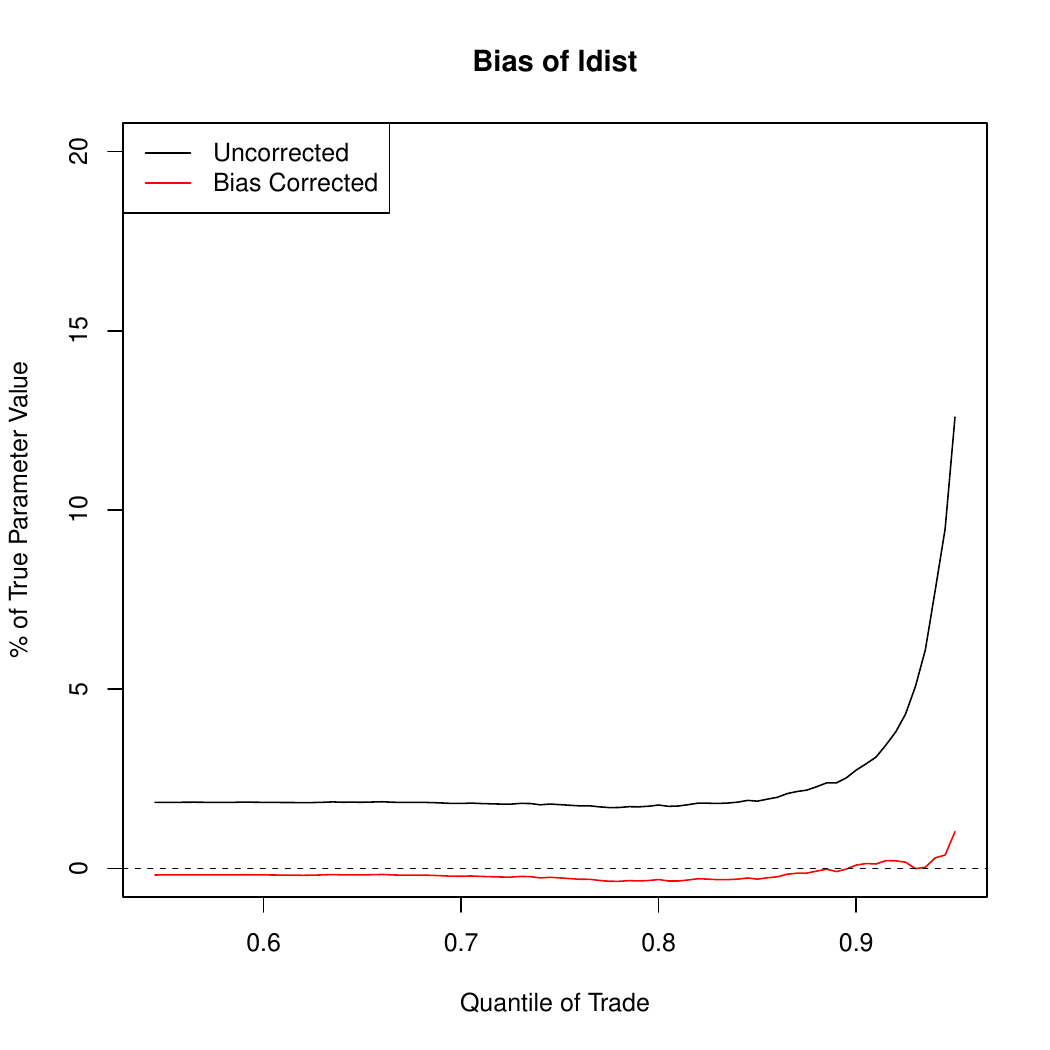}
            \includegraphics[width=.32\textwidth,height= .5\textwidth,angle=0,page=2]{Sim_tobit_coef.pdf}
            \includegraphics[width=.32\textwidth,height= .5\textwidth,angle=0,page=3]{Sim_tobit_coef.pdf}
        \end{center}\caption{Bias, standard deviation and root mean squared error for the estimators of the DR-coefficients of log-distance.}\label{fig:mc}
    \end{figure}

    \begin{figure}[h!]
        \begin{center}
            \includegraphics[width=.32\textwidth,height= .5\textwidth,angle=0,page=6]{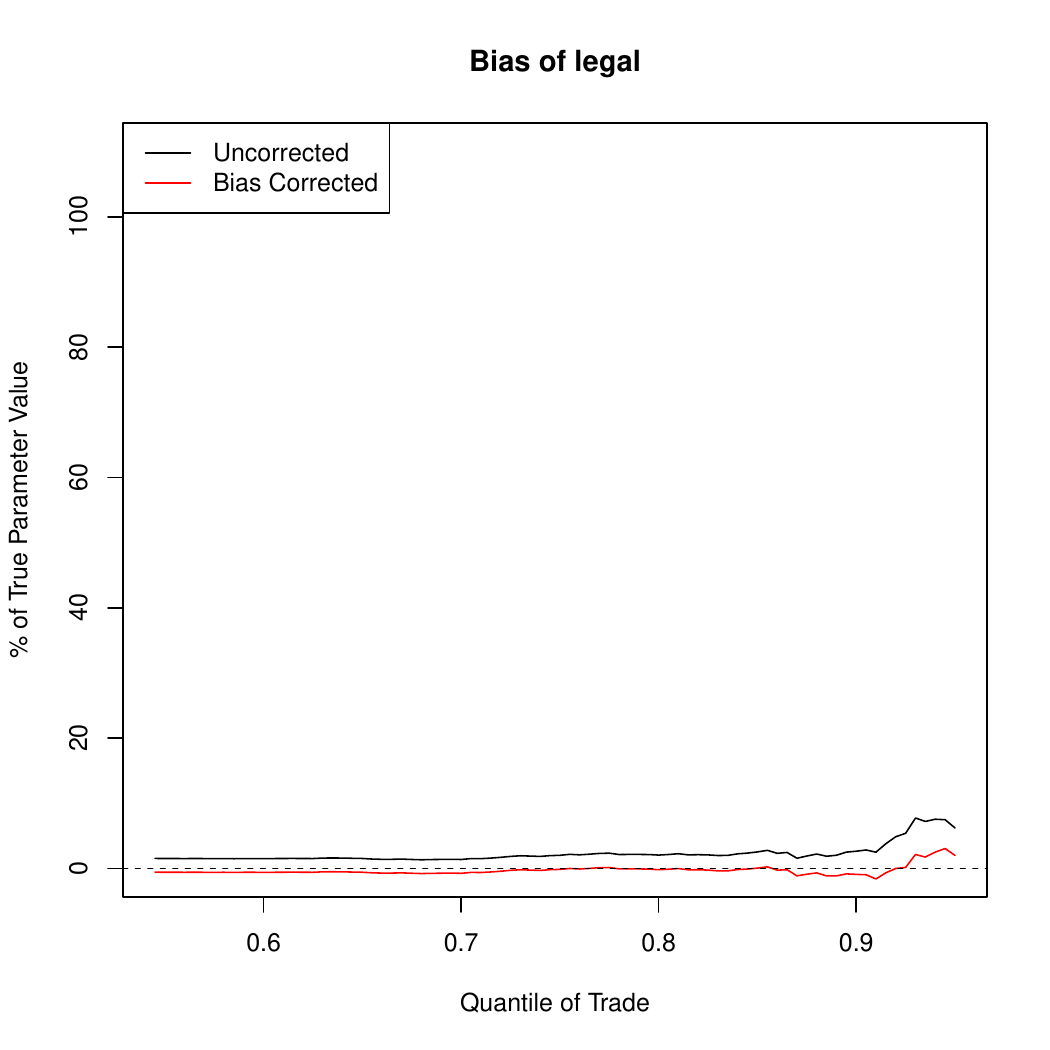}
            \includegraphics[width=.32\textwidth,height= .5\textwidth,angle=0,page=7]{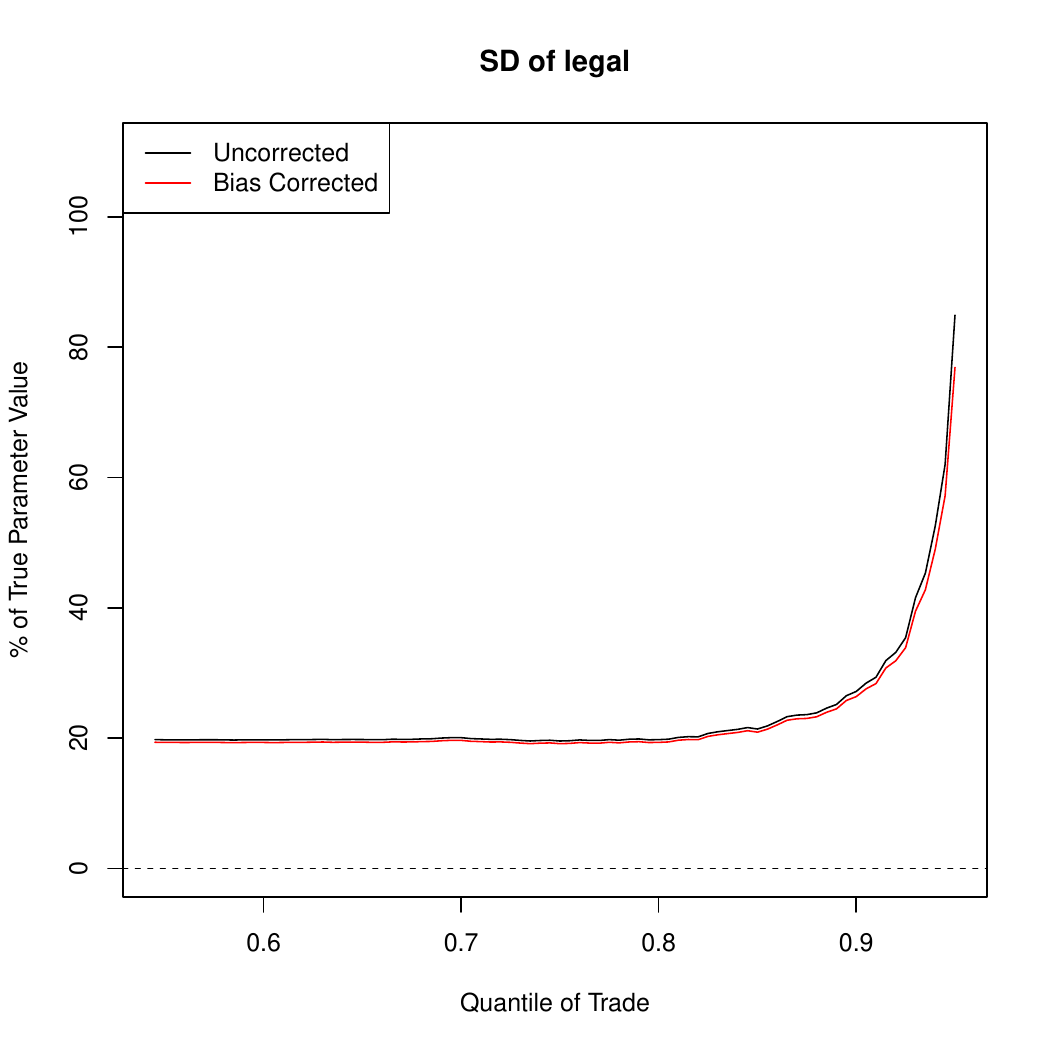}
            \includegraphics[width=.32\textwidth,height= .5\textwidth,angle=0,page=8]{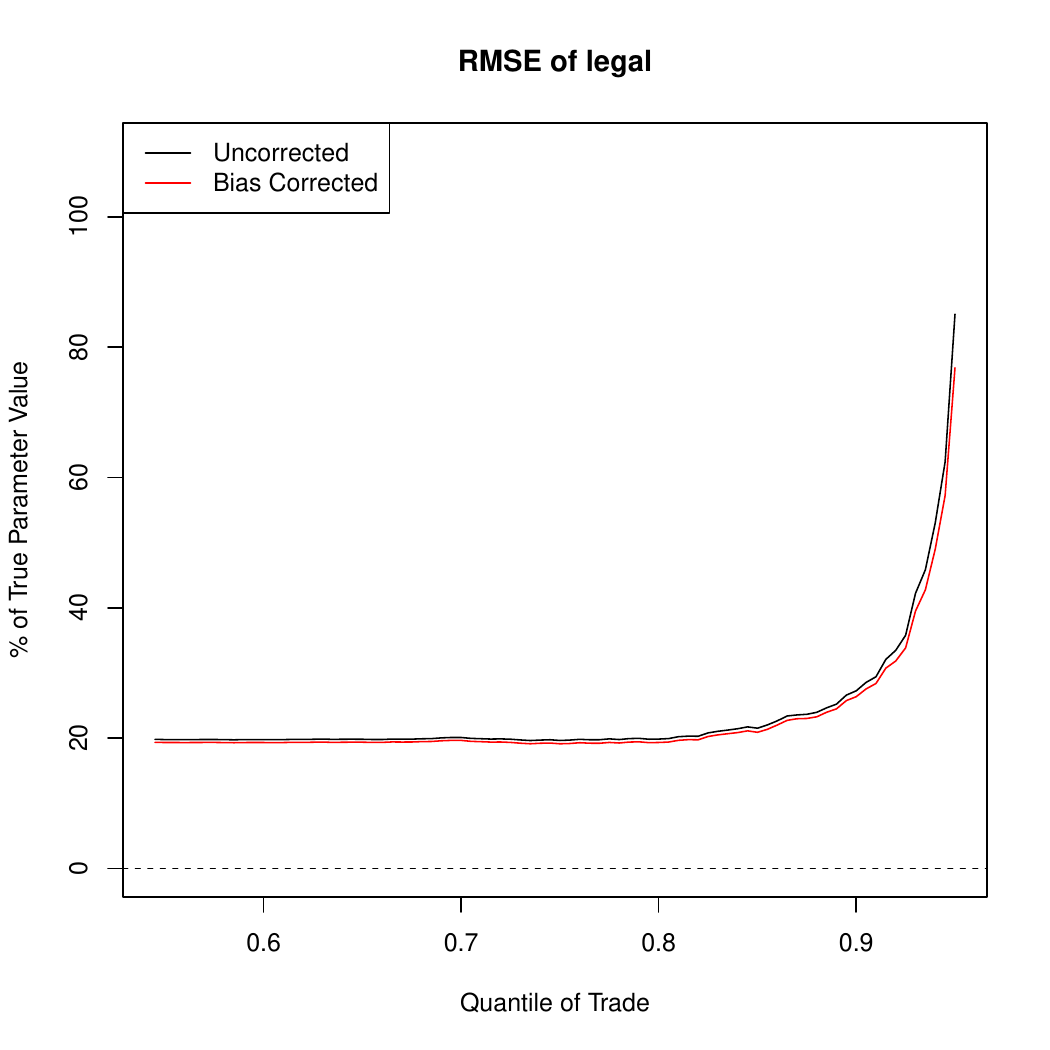}
        \end{center}\caption{Bias, standard deviation and root mean squared error for the estimators of the DR-coefficients of same legal system.}\label{fig:mc2}
    \end{figure}

Figure \ref{fig:mc3} reports the biases, standard deviations and rmses of the estimators of the counterfactual distributions at two levels of log-distance as a function of the quantiles of $y_{ij}$ in the design with independent errors. The levels of distance in these distributions are the same as in the empirical application, i.e. $k=0$ and $k=1$ correspond to the observed values and two times the observed values, respectively.  All the results are in percentage of the true value of the functions. In this case we find that  the uncorrected and bias corrected estimators display small biases relative to their standard deviations, and have similar standard deviations and rmses at both treatment levels. Indeed the standard deviations and rmses are difficult to distinguish in the figure as they are almost superposed.   In results not reported, we find very similar patterns in the design with pairwise dependent errors and for the estimators of the counterfactual distributions at the same two levels of legal as in the empirical application.  

    \begin{figure}[h!]
        \begin{center}
            \includegraphics[width=.32\textwidth,height= .5\textwidth,angle=0,page=1]{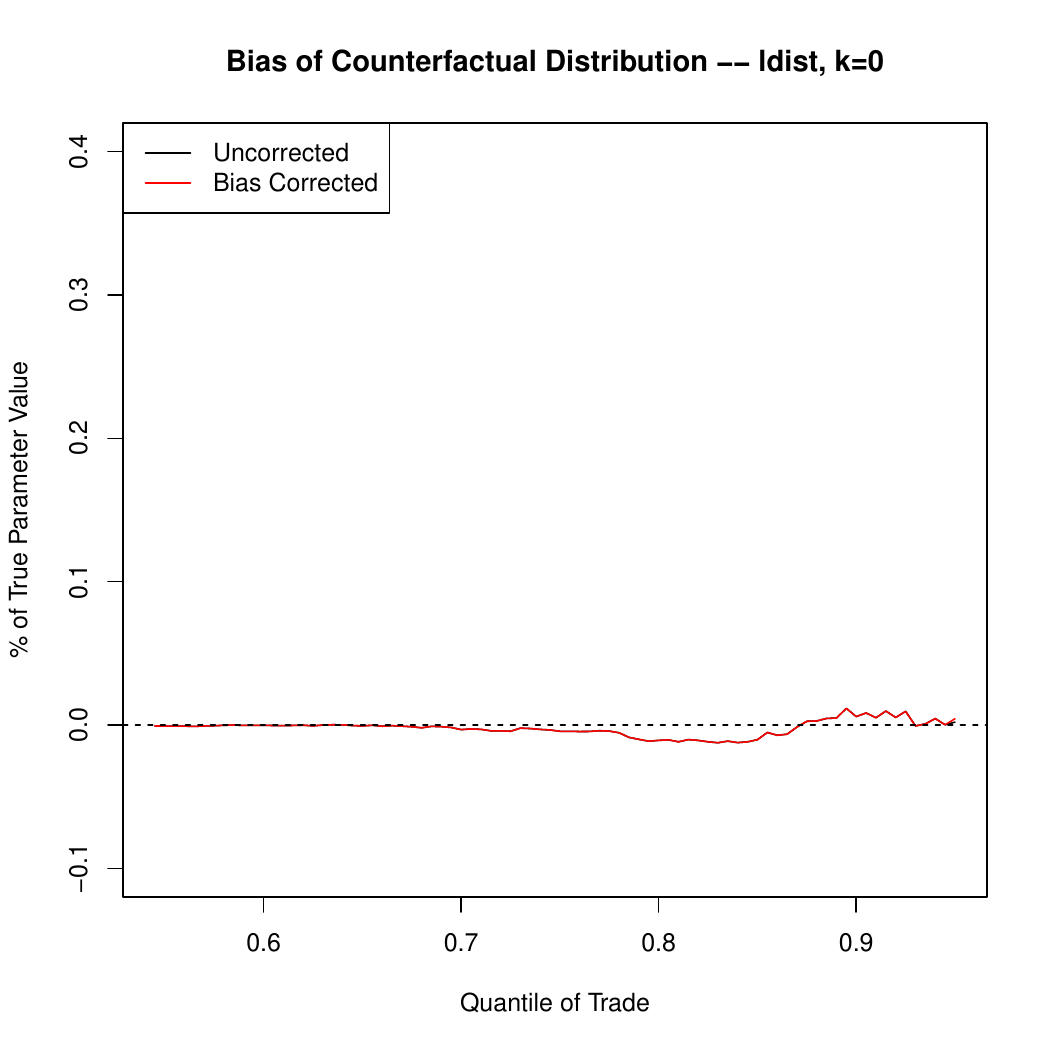}
            \includegraphics[width=.32\textwidth,height= .5\textwidth,angle=0,page=2]{Sim_tobit_df.pdf}
            \includegraphics[width=.32\textwidth,height= .5\textwidth,angle=0,page=3]{Sim_tobit_df.pdf}
            \includegraphics[width=.32\textwidth,height= .5\textwidth,angle=0,page=6]{Sim_tobit_df.pdf}
            \includegraphics[width=.32\textwidth,height= .5\textwidth,angle=0,page=7]{Sim_tobit_df.pdf}
            \includegraphics[width=.32\textwidth,height= .5\textwidth,angle=0,page=8]{Sim_tobit_df.pdf}
        \end{center}\caption{Bias, standard deviation and root mean squared error for the estimators of the counterfactual distributions of log-distance.}\label{fig:mc3}
    \end{figure}

Table \ref{table:mc} shows results on the finite sample properties of 95\% confidence bands for the DR coefficients and counterfactual distributions in the design with independent errors. The confidence bands are constructed by multiplier bootstrap with 500 draws, standard normal weights, and a grid of values $\bar \mY$ that includes the sample quantiles of the volume of trade with indexes $\{.54, .55, \ldots, .95\}$ in the trade data set. For the coefficients, it reports the average length of the confidence bands integrated over threshold values, the average value of the estimated critical values, and the empirical coverages of the confidence bands. For the distributions, it reports the same measures averaged also over the two treatment levels and where the coverage of the bands is joint for the two counterfactual distributions.\footnote{The joint coverage of the bands for the quantile functions and quantile effect is determined by the joint coverage of the bands of the distribution functions in our construction. We refer to \cite{CFVMW2015} for a numerical analysis on the marginal coverage of the bands for the quantile effects.} For comparison, it also reports the coverage of pointwise confidence bands using the normal distribution, i.e. with critical value equal to 1.96. The last row computes the ratio of the standard error averaged across simulations to the simulation standard deviation, integrated over threshold values for the coefficients and over thresholds and treatment levels for the distributions. We consider standard errors and confidence bands with and without accounting for pairwise clustering. All the results are computed for confidence bands centered at the uncorrected fixed effects estimates and at the bias corrected estimates. For the coefficients, we find that the bands centered at the uncorrected estimates undercover the true coefficients, whereas the bands centered at the bias corrected estimates have coverages close to the nominal level.  The joint coverage of the bands for the  distributions is close to the nominal level regardless of whether they are centered at the uncorrected or bias corrected estimates. We attribute this similarity in coverage to the small biases in the uncorrected estimates of the distributions  found in fig. \ref{fig:mc3}.   As expected, pointwise bands severely undercover the  entire functions. The standard errors based on the asymptotic distribution provide a good approximation to the sampling variability of both the uncorrected and bias corrected estimators. Accounting for pairwise clustering in this design where it is not necessary has very little effect on the quality of the inference.

\begin{table}
\caption{95\% Confidence Bands -- Design with Independent Errors}\label{table:mc}
\begin{center}
    \begin{tabular}{lcccccccc}
        \hline\hline
        & \multicolumn{4}{c}{Uncorrected} & \multicolumn{4}{c}{Bias Corrected} \\
        & $\beta_{ldist}$ & $\beta_{legal}$ & $F_{ldist}$ & $F_{legal}$ & $\beta_{ldist}$ & $\beta_{legal}$ & $F_{ldist}$ & $F_{legal}$ \\
        \hline
        \multicolumn{9}{l}{\textbf{Unclustered Inference}} \\
        \hspace{3mm} Average Length & 0.24 & 0.35 & 0.01 & 0.02 & 0.24 & 0.35 & 0.01 & 0.02 \\
        \hspace{3mm} Average Critical Value & 2.90 & 2.89 & 3.10 & 3.13 & 2.90 & 2.89 & 3.10 & 3.13 \\
        \hspace{3mm} Coverage uniform band (\%) & 83 & 91 & 94 & 93 & 95 & 94 & 94 & 94 \\
        \hspace{3mm} Coverage pointwise band (\%) & 35 & 58 & 35 & 29 & 60 & 64 & 35 & 29 \\
        \hspace{3mm} Average SE/SD & 0.97 & 1.01 & 0.99 & 1.01 & 1.00 & 1.04 & 0.99 & 1.01 \\ 
        \hline
        \multicolumn{9}{l}{\textbf{Pairwise Clustered Inference}} \\
        \hspace{3mm} Average Length & 0.23 & 0.35 & 0.01 & 0.02 & 0.23 & 0.35 & 0.01 & 0.02 \\
        \hspace{3mm} Average Critical Value & 2.89 & 2.89 & 3.09 & 3.12 & 2.89 & 2.89 & 3.09 & 3.12 \\
        \hspace{3mm} Coverage uniform band (\%) & 82 & 92 & 93 & 93 & 94 & 93 & 93 & 93 \\
        \hspace{3mm} Coverage pointwise band (\%) & 35 & 57 & 35 & 30 & 59 & 63 & 36 & 29 \\
        \hspace{3mm} Average SE/SD & 0.97 & 1.01 & 0.99 & 1.01 & 1.00 & 1.04 & 0.99 & 1.01 \\ 
        \hline\hline
        \multicolumn{9}{l}{\footnotesize Notes: Nominal level of critical values is 95\%. 500 simulations with 500 multiplier bootstrap draws.}
    \end{tabular}
\end{center}
\end{table}

Table \ref{table:mc2} reports the same results as table \ref{table:mc} for the design with pairwise dependent errors. The bands that do not account for pairwise clustering undercover the functions because the standard errors underestimate the standard deviations of the estimators. Compared to the design with independent errors,  the critical values are similar but the bands that account for clustering are    wider  due to the increase in the standard errors. To sum up, inference methods robust to pairwise clustering perform well in both designs, whereas inference methods that do not account for clustering undercover in the presence of pairwise dependence. The bias corrections are effective in reducing  bias and bringing the coverage probabilities of the bands close to their nominal level for the coefficients, whereas they have little effect for the distributions.

\begin{table}
\caption{95\% Confidence Bands -- Design with Pairwise Dependent Errors}\label{table:mc2}
\begin{center}
    \begin{tabular}{lcccccccc}
        \hline\hline
        & \multicolumn{4}{c}{Uncorrected} & \multicolumn{4}{c}{Bias Corrected} \\
        & $\beta_{ldist}$ & $\beta_{legal}$ & $F_{ldist}$ & $F_{legal}$ & $\beta_{ldist}$ & $\beta_{legal}$ & $F_{ldist}$ & $F_{legal}$ \\
        \hline
        \multicolumn{9}{l}{\textbf{Unclustered Inference}} \\
        \hspace{3mm} Average Length & 0.24 & 0.35 & 0.01 & 0.02 & 0.24 & 0.35 & 0.01 & 0.02 \\
        \hspace{3mm} Average Critical Value & 2.90 & 2.89 & 3.10 & 3.13 & 2.90 & 2.89 & 3.10 & 3.13 \\
        \hspace{3mm} Coverage uniform band (\%) & 64 & 73 & 73 & 68 & 80 & 78 & 74 & 68 \\
        \hspace{3mm} Coverage pointwise band (\%) & 21 & 27 & 11 & 8 & 32 & 36 & 12 & 8 \\
        \hspace{3mm} Average SE/SD & 0.77 & 0.76 & 0.77 & 0.77 & 0.79 & 0.78 & 0.77 & 0.77 \\  
        \hline
        \multicolumn{9}{l}{\textbf{Pairwise Clustered Inference}} \\
        \hspace{3mm} Average Length & 0.30 & 0.44 & 0.02 & 0.02 & 0.30 & 0.44 & 0.02 & 0.02 \\
        \hspace{3mm} Average Critical Value & 2.82 & 2.82 & 3.02 & 3.05 & 2.82 & 2.82 & 3.02 & 3.05 \\
        \hspace{3mm} Coverage uniform band (\%) & 86 & 92 & 93 & 92 & 96 & 93 & 93 & 92 \\
        \hspace{3mm} Coverage pointwise band (\%) & 47 & 59 & 44 & 37 & 67 & 66 & 43 & 37 \\
        \hspace{3mm} Average SE/SD & 1.00 & 0.99 & 1.00 & 0.99 & 1.03 & 1.01 & 1.00 & 0.99 \\ 
        \hline\hline
        \multicolumn{9}{l}{\footnotesize{Notes: Nominal level of critical values is 95\%. 500 simulations with 500 multiplier bootstrap draws.}}
    \end{tabular}
\end{center}
\end{table}

\section{Conclusion}
We have constructed confidence bands for quantile functions and quantile effects in nonlinear network and panel models with two-way unobserved effects. Our construction relies on the generic method of \cite{CFVMW2015} to convert confidence bands for distributions into confidence bands for quantiles. The same method can be applied to more complicated models such as nonlinear models with interactive unobserved effects or factor structure, provided that  confidence bands for distributions  in these models are supplied. Such bands are not currently available, but could be obtained by extending the central limit theorem of Chen, Fern\'andez-Val and Weidner (in press)
 to a functional central limit theorem. We leave such extension to future work. 

\vspace{-0.3cm}
\phantom{\cite{cfw14}} 

\bibliographystyle{chicagoMOD}
\bibliography{references}  

\begin{thebibliography}{}

\bibitem[\protect\citeauthoryear{Abadie, Athey, Imbens, and Wooldridge}{Abadie,
  Athey, Imbens and Wooldridge}{2014}]{AAIW-14}
Abadie, A., S.~Athey, G.~W. Imbens, and J.~M. Wooldridge (2014).
\newblock Finite population causal standard errors.

\bibitem[\protect\citeauthoryear{Abrevaya and Dahl}{Abrevaya and
  Dahl}{2008}]{AbrevayaDahl2008}
Abrevaya, J. and C.~M. Dahl (2008).
\newblock The effects of birth inputs on birthweight: evidence from quantile
  estimation on panel data.
\newblock {\em Journal of Business \& Economic Statistics\/}~{\em 26\/}(4),
  379--397.

\bibitem[\protect\citeauthoryear{Anderson and van Wincoop}{Anderson and van
  Wincoop}{2003}]{AndersonWincoop2003}
Anderson, J.~E. and E.~van Wincoop (2003, March).
\newblock Gravity with gravitas: A solution to the border puzzle.
\newblock {\em American Economic Review\/}~{\em 93\/}(1), 170--192.

\bibitem[\protect\citeauthoryear{Arellano and Bonhomme}{Arellano and
  Bonhomme}{2016}]{ArellanoBonhomme2016}
Arellano, M. and S.~Bonhomme (2016).
\newblock Nonlinear panel data estimationvia quantile regressions.
\newblock {\em Unpublished manuscript\/}.

\bibitem[\protect\citeauthoryear{Arellano and Weidner}{Arellano and
  Weidner}{2016}]{ArellanoWeidner2016}
Arellano, M. and M.~Weidner (2016).
\newblock Instrumental variable quantile regressions in large panels with fixed
  effects.
\newblock {\em Unpublished manuscript\/}.

\bibitem[\protect\citeauthoryear{Cameron and Miller}{Cameron and
  Miller}{2014}]{CM-14}
Cameron, A.~C. and D.~L. Miller (2014).
\newblock Robust inference for dyadic data.

\bibitem[\protect\citeauthoryear{Candelaria}{Candelaria}{2016}]{candelaria16}
Candelaria, L.~E. (2016).
\newblock A semiparametric network formation model with multiple linear fixed
  effects.
\newblock {\em Unpublished manuscript\/}.

\bibitem[\protect\citeauthoryear{Charbonneau}{Charbonneau}{2017}]{charbonneau17}
Charbonneau, K.~B. (2017).
\newblock Multiple fixed effects in binary response panel data models.
\newblock {\em The Econometrics Journal\/}~{\em 20\/}(3), S1--S13.

\bibitem[\protect\citeauthoryear{Chen, Fern\'andez-Val, and Weidner}{Chen,
  Fern\'andez-Val and Weidner}{ress}]{cfw14}
Chen, M., I.~Fern\'andez-Val, and M.~Weidner (in press).
\newblock Nonlinear factor models for network and panel data.
\newblock {\em Journal of Econometrics, available online:
  https://doi.org/10.1016/j.jeconom.2020.04.004\/}.

\bibitem[\protect\citeauthoryear{{Chen}, {Chernozhukov}, {Fern{\'a}ndez-Val},
  {Kostyshak}, and {Luo}}{{Chen}, {Chernozhukov}, {Fern{\'a}ndez-Val},
  {Kostyshak} and {Luo}}{2018}]{ccfkl18}
{Chen}, X., V.~{Chernozhukov}, I.~{Fern{\'a}ndez-Val}, S.~{Kostyshak}, and
  Y.~{Luo} (2018, September).
\newblock {Shape-Enforcing Operators for Point and Interval Estimators}.
\newblock {\em ArXiv e-prints\/}.

\bibitem[\protect\citeauthoryear{Chernozhukov, Chetverikov, and
  Kato}{Chernozhukov, Chetverikov and Kato}{2016}]{cck-16}
Chernozhukov, V., D.~Chetverikov, and K.~Kato (2016).
\newblock Empirical and multiplier bootstraps for suprema of empirical
  processes of increasing complexity, and related {G}aussian couplings.
\newblock {\em Stochastic Process. Appl.\/}~{\em 126\/}(12), 3632--3651.

\bibitem[\protect\citeauthoryear{Chernozhukov, Fernandez-Val, and
  Galichon}{Chernozhukov, Fernandez-Val and
  Galichon}{2009}]{ChernozhukovFernandezValGalichon2009}
Chernozhukov, V., I.~Fernandez-Val, and A.~Galichon (2009).
\newblock Improving point and interval estimators of monotone functions by
  rearrangement.
\newblock {\em Biometrika\/}~{\em 96\/}(3), 559--575.

\bibitem[\protect\citeauthoryear{Chernozhukov, Fern{\'a}ndez-Val, Hahn, and
  Newey}{Chernozhukov, Fern{\'a}ndez-Val, Hahn and
  Newey}{2013}]{ChernozhukovFernandezValHahnNewey2013}
Chernozhukov, V., I.~Fern{\'a}ndez-Val, J.~Hahn, and W.~Newey (2013).
\newblock Average and quantile effects in nonseparable panel models.
\newblock {\em Econometrica\/}~{\em 81\/}(2), 535--580.

\bibitem[\protect\citeauthoryear{Chernozhukov, Fernandez-Val, Hoderlein,
  Holzmann, and Newey}{Chernozhukov, Fernandez-Val, Hoderlein, Holzmann and
  Newey}{2015}]{chernozhukov2015nonparametric}
Chernozhukov, V., I.~Fernandez-Val, S.~Hoderlein, H.~Holzmann, and W.~Newey
  (2015).
\newblock Nonparametric identification in panels using quantiles.
\newblock {\em Journal of Econometrics\/}~{\em 188\/}(2), 378--392.

\bibitem[\protect\citeauthoryear{Chernozhukov, Fern{\'a}ndez-Val, and
  Melly}{Chernozhukov, Fern{\'a}ndez-Val and
  Melly}{2013}]{ChernozhukovFernandezValMelly2013}
Chernozhukov, V., I.~Fern{\'a}ndez-Val, and B.~Melly (2013).
\newblock Inference on counterfactual distributions.
\newblock {\em Econometrica\/}~{\em 81\/}(6), 2205--2268.

\bibitem[\protect\citeauthoryear{{Chernozhukov}, {Fern{\'a}ndez-Val}, {Melly},
  and {W{\"u}thrich}}{{Chernozhukov}, {Fern{\'a}ndez-Val}, {Melly} and
  {W{\"u}thrich}}{2016}]{CFVMW2015}
{Chernozhukov}, V., I.~{Fern{\'a}ndez-Val}, B.~{Melly}, and K.~{W{\"u}thrich}
  (2016, August).
\newblock {Generic Inference on Quantile and Quantile Effect Functions for
  Discrete Outcomes}.
\newblock {\em ArXiv e-prints\/}.

\bibitem[\protect\citeauthoryear{{Cruz-Gonzalez}, {Fernandez-Val}, and
  {Weidner}}{{Cruz-Gonzalez}, {Fernandez-Val} and {Weidner}}{2016}]{CFW-16}
{Cruz-Gonzalez}, M., I.~{Fernandez-Val}, and M.~{Weidner} (2016, October).
\newblock {probitfe and logitfe: Bias corrections for probit and logit models
  with two-way fixed effects}.
\newblock {\em ArXiv e-prints\/}.

\bibitem[\protect\citeauthoryear{de~Paula}{de~Paula}{2019}]{paula19}
de~Paula, A. (2019).
\newblock Econometric models of network formation.

\bibitem[\protect\citeauthoryear{Dhaene and Jochmans}{Dhaene and
  Jochmans}{2015}]{DhaeneJochmans2015}
Dhaene, G. and K.~Jochmans (2015).
\newblock Split-panel jackknife estimation of fixed-effect models.
\newblock {\em The Review of Economic Studies\/}~{\em 82\/}(3), 991--1030.

\bibitem[\protect\citeauthoryear{Dzemski}{Dzemski}{2017}]{Dzemski2017}
Dzemski, A. (2017).
\newblock An empirical model of dyadic link formation in a network with
  unobserved heterogeneity.
\newblock {\em Unpublished manuscript\/}.

\bibitem[\protect\citeauthoryear{Eaton and Kortum}{Eaton and
  Kortum}{2001}]{EatonKortum2001}
Eaton, J. and S.~Kortum (2001).
\newblock Trade in capital goods.
\newblock {\em European Economic Review\/}~{\em 45\/}(7), 1195--1235.

\bibitem[\protect\citeauthoryear{Fernandez-Val and Weidner}{Fernandez-Val and
  Weidner}{2016}]{FernandezValWeidner2016}
Fernandez-Val, I. and M.~Weidner (2016).
\newblock Individual and time effects in nonlinear panel models with large {N},
  {T}.
\newblock {\em Journal of Econometrics\/}~{\em 192\/}(1), 291--312.

\bibitem[\protect\citeauthoryear{Fernandez-Val and Weidner}{Fernandez-Val and
  Weidner}{2018}]{FernandezValWeidner2018}
Fernandez-Val, I. and M.~Weidner (2018).
\newblock Fixed effects estimation of large-{T} panel data models.
\newblock {\em Annual Review of Economics\/}~{\em 10\/}(1), 109--138.

\bibitem[\protect\citeauthoryear{Galvao, Lamarche, and Lima}{Galvao, Lamarche
  and Lima}{2013}]{GalvaoLamarcheLima2013}
Galvao, A.~F., C.~Lamarche, and L.~R. Lima (2013).
\newblock Estimation of censored quantile regression for panel data with fixed
  effects.
\newblock {\em Journal of the American Statistical Association\/}~{\em
  108\/}(503), 1075--1089.

\bibitem[\protect\citeauthoryear{Gao}{Gao}{2020}]{gao20}
Gao, W.~Y. (2020).
\newblock Nonparametric identification in index models of link formation.
\newblock {\em Journal of Econometrics\/}~{\em 215\/}(2), 399 -- 413.

\bibitem[\protect\citeauthoryear{Gin\'e and Zinn}{Gin\'e and
  Zinn}{1984}]{gz-84}
Gin\'e, E. and J.~Zinn (1984).
\newblock Some limit theorems for empirical processes.
\newblock {\em Ann. Probab.\/}~{\em 12\/}(4), 929--998.
\newblock With discussion.

\bibitem[\protect\citeauthoryear{Graham}{Graham}{2016}]{Graham2016}
Graham, B.~S. (2016).
\newblock Homophily and transitivity in dynamic network formation.

\bibitem[\protect\citeauthoryear{Graham}{Graham}{2017}]{Graham2017}
Graham, B.~S. (2017).
\newblock An econometric model of link formation with degree heterogeneity.
\newblock {\em Econometrica\/}~{\em 85\/}(4), 1033--1063.

\bibitem[\protect\citeauthoryear{Graham, Hahn, Poirier, and Powell}{Graham,
  Hahn, Poirier and Powell}{2015}]{GrahamHahnPoirierPowell2015}
Graham, B.~S., J.~Hahn, A.~Poirier, and J.~L. Powell (2015).
\newblock Quantile regression with panel data.

\bibitem[\protect\citeauthoryear{Graham, Hahn, and Powell}{Graham, Hahn and
  Powell}{2009}]{GrahamHahnPowell2009}
Graham, B.~S., J.~Hahn, and J.~L. Powell (2009).
\newblock The incidental parameter problem in a non-differentiable panel data
  model.
\newblock {\em Economics Letters\/}~{\em 105\/}(2), 181--182.

\bibitem[\protect\citeauthoryear{Hahn and Newey}{Hahn and
  Newey}{2004}]{Hahn:2004p882}
Hahn, J. and W.~Newey (2004).
\newblock Jackknife and analytical bias reduction for nonlinear panel models.
\newblock {\em Econometrica\/}~{\em 72\/}(4), 1295--1319.

\bibitem[\protect\citeauthoryear{Harrigan}{Harrigan}{1994}]{Harrigan1994}
Harrigan, J. (1994).
\newblock Scale economies and the volume of trade.
\newblock {\em The Review of Economics and Statistics\/}, 321--328.

\bibitem[\protect\citeauthoryear{Head and Mayer}{Head and
  Mayer}{2014}]{HeadMayer2014}
Head, K. and T.~Mayer (2014).
\newblock Gravity equations: Workhorse, toolkit, and cookbook.
\newblock Volume~4, Chapter~3, pp.\  131--195. Handbook of International
  Economics.

\bibitem[\protect\citeauthoryear{Helpman, Melitz, and Rubinstein}{Helpman,
  Melitz and Rubinstein}{2008}]{Helpman01052008}
Helpman, E., M.~Melitz, and Y.~Rubinstein (2008).
\newblock Estimating trade flows: Trading partners and trading volumes.
\newblock {\em The Quarterly Journal of Economics\/}~{\em 123\/}(2), 441--487.

\bibitem[\protect\citeauthoryear{Jochmans}{Jochmans}{2018}]{jochmans18}
Jochmans, K. (2018).
\newblock Semiparametric analysis of network formation.
\newblock {\em Journal of Business \& Economic Statistics\/}~{\em 36\/}(4),
  705--713.

\bibitem[\protect\citeauthoryear{Kato and Galvao}{Kato and
  Galvao}{2016}]{GalvaoKato2016}
Kato, K. and A.~Galvao (2016).
\newblock Smoothed quantile regression for panel data.
\newblock Technical Report~1.

\bibitem[\protect\citeauthoryear{Kato, Galvao, and Montes-Rojas}{Kato, Galvao
  and Montes-Rojas}{2012}]{KatoGalvaoMontesRojas2012}
Kato, K., A.~F. Galvao, and G.~V. Montes-Rojas (2012).
\newblock Asymptotics for panel quantile regression models with individual
  effects.
\newblock {\em Journal of Econometrics\/}~{\em 170\/}(1), 76--91.

\bibitem[\protect\citeauthoryear{Kim and Sun}{Kim and Sun}{2016}]{KimSun2016}
Kim, M.~S. and Y.~Sun (2016).
\newblock Bootstrap and k-step bootstrap bias corrections for the fixed effects
  estimator in nonlinear panel data models.
\newblock {\em Econometric Theory\/}~{\em 32\/}(6), 1523--1568.

\bibitem[\protect\citeauthoryear{Koenker}{Koenker}{2004}]{Koenker2004}
Koenker, R. (2004).
\newblock Quantile regression for longitudinal data.
\newblock {\em Journal of Multivariate Analysis\/}~{\em 91\/}(1), 74--89.

\bibitem[\protect\citeauthoryear{Lamarche}{Lamarche}{2010}]{Lamarche2010}
Lamarche, C. (2010).
\newblock Robust penalized quantile regression estimation for panel data.
\newblock {\em Journal of Econometrics\/}~{\em 157\/}(2), 396--408.

\bibitem[\protect\citeauthoryear{Machado and Santos~Silva}{Machado and
  Santos~Silva}{2018}]{mss18}
Machado, J.~A. and J.~Santos~Silva (2018).
\newblock Quantiles via moments.
\newblock {\em Unpublished manuscript\/}.

\bibitem[\protect\citeauthoryear{Neyman and Scott}{Neyman and
  Scott}{1948}]{NeymanScott1948}
Neyman, J. and E.~Scott (1948).
\newblock Consistent estimates based on partially consistent observations.
\newblock {\em Econometrica\/}~{\em 16\/}(1), 1--32.

\bibitem[\protect\citeauthoryear{Rosen}{Rosen}{2012}]{Rosen2012}
Rosen, A.~M. (2012).
\newblock Set identification via quantile restrictions in short panels.
\newblock {\em Journal of Econometrics\/}~{\em 166\/}(1), 127--137.

\bibitem[\protect\citeauthoryear{Santos~Silva and Tenreyro}{Santos~Silva and
  Tenreyro}{2006}]{SantosSilvaTenreyro2006}
Santos~Silva, J. and S.~Tenreyro (2006).
\newblock The log of gravity.
\newblock {\em The Review of Economics and statistics\/}~{\em 88\/}(4),
  641--658.

\bibitem[\protect\citeauthoryear{{Stammann}}{{Stammann}}{2017}]{alpaca17}
{Stammann}, A. (2017, Jul).
\newblock {Fast and Feasible Estimation of Generalized Linear Models with
  High-Dimensional k-way Fixed Effects}.
\newblock {\em arXiv e-prints\/}, arXiv:1707.01815.

\bibitem[\protect\citeauthoryear{Toth}{Toth}{2017}]{toth17}
Toth, P. (2017).
\newblock Semiparametric estimation in network formation models with homophily
  and degree heterogeneity.
\newblock {\em Unpublished manuscript\/}.

\bibitem[\protect\citeauthoryear{van~der Vaart and Wellner}{van~der Vaart and
  Wellner}{1996}]{vanderVaartandWellner1996}
van~der Vaart, A.~W. and J.~A. Wellner (1996).
\newblock Weak convergence.
\newblock In {\em Weak Convergence and Empirical Processes}, pp.\  16--28.
  Springer.

\bibitem[\protect\citeauthoryear{Yan, Jiang, Fienberg, and Leng}{Yan, Jiang,
  Fienberg and Leng}{2016}]{Yan2016statistical}
Yan, T., B.~Jiang, S.~E. Fienberg, and C.~Leng (2016).
\newblock Statistical inference in a directed network model with covariates.
\newblock {\em arXiv preprint arXiv:1609.04558\/}.

\end{thebibliography}

\newpage

\begin{appendix}

\renewcommand{\theequation}{A.\arabic{equation}}
\label{appendixA}
\setcounter{equation}{0}

\section{Proofs of Main Text Results}
We present the proofs of  Theorems~\ref{th:expansion} and \ref{th:bc}, and relegate
various technical details   
 to the on-line supplementary appendix.
Once Theorems~\ref{th:expansion} and \ref{th:bc} are shown,
the proof of Theorem~\ref{cor:bias_correction} for the multiplier bootstrap follows from Theorem 2.2 in \cite{cck-16}.
The uniform confidence bands $I_F$ for the cdfs in \eqref{ConfBandCDF} obtained by the multiplier bootstrap
can then be inverted and differenced to obtain uniform confidence bands for the quantile function and quantile effects,
see \cite{CFVMW2015} and also Lemma~\ref{theorem:bquant} and \ref{theorem:bqte} above.
This appendix thus contains the proofs of all the main results that are new to the current paper.
 The proofs for all of the lemmas below are given in the supplementary appendix.
 All stochastic statements in the following are conditional on $  \{(x_{ij}, v_i, w_j): (i,j) \in \mD\}$.

As explained in Section~\ref{sec:asymptotics}, we consider the logistic cdf $\Lambda_y(\pi) = \Lambda(\pi) = (1+\exp(-\pi))^{-1}$
 for all our theorems.
In the following we  indicate the dependence on $y \in \mY$  as a subscript, for example,
we write $\theta_y$ instead of $\theta(y)$ from now on.
We use the column vector $w_{ij} = (x_{ij}', e_{i,I}', e_{j,J}')'$, as in Section~\ref{sec:UniformBands},
and can then write the single index 
$\pi_{y,ij} := x_{ij}' \beta_y + \alpha_{y,i} + \gamma_{y,j}$ simply as
$\pi_{y,ij} = w_{ij}' \theta_y$. The corresponding estimator is $\widehat \pi_{y,ij} = w_{ij}'  \widehat \theta_y$.
We also define minus the log-likelihood function as
$ \ell_{y,ij}(\pi) := - 1\{y_{ij} \leq y \} \log \Lambda(\pi) - 1\{y_{ij} > y \} \log [1 - \Lambda(\pi)]$.
Let $\pi_y$ be a $n$-vector containing $\pi_{y,ij}$, $(i,j) \in \mD$.
For a given $y \in \mY$ we can then rewrite the estimation problem in~\eqref{fe-cmle} as
\begin{align}
\widehat \pi_y &= \arg \min_{\pi_y \in \mathbb{R}^n}  \sum_{(i,j) \in \mD}   \ell_{y,ij}(\pi_{y,ij})
,
&& \text{s.t.}
&
\exists\,  \theta \in \mathbb{R}^{d_x + I + J}: \, \pi_{y,ij} = w_{ij}' \theta_y .
    \label{DefHatPi}
\end{align}
In the following we denote the true parameter values by $\theta^0$,
and correspondingly we write $\pi^0_{y,ij} = w_{ij}' \theta^0_y$, in order to distinguish the true value
from generic values like the argument $\pi_{y,ij}$ in the last display.
For the $k$'th derivative of $\ell_{y,ij}(\pi_{y,ij})$ with respect to $\pi_{y,ij}$
we write  $\partial_{\pi^k} \ell_{y,ij}(\pi_{y,ij})$, and we drop the argument when the derivative is evaluated
at $\pi^0_{y,ij}$, that is, $\partial_{\pi^k} \ell_{y,ij} = \partial_{\pi^k} \ell_{y,ij}(\pi^0_{y,ij})$.
The normalized score for observation $i,j$ then reads
\begin{align*}
     s_{y,ij} :=  \left[   \partial_{\pi^2} \ell_{y,ij} \right]^{-1/2}  \partial_\pi \ell_{y,ij} 
      =   \left(  \Lambda^{(1)}_{y,ij} \right)^{-1/2}  \partial_\pi \ell_{y,ij}  ,
\end{align*}
where $\Lambda^{(1)}_{y,ij} = \Lambda^{(1)}(\pi^0_{y,ij}) = \partial_{\pi}   \Lambda(\pi^0_{y,ij}) $,
as defined in Section~\ref{sec:AsyDistrUncorrected}.
Note that $\Ep s_{y,ij} = 0$ and $\Ep s_{y,ij}^2 =1$.

Let  $s_y$ be the $n$-vector obtained by stacking 
the elements $s_{y,ij}$ across all observations $(i,j) \in \mD$.
Similarly, let $\Lambda^{(1)}_{y}$ be the $n \times n$ diagonal matrix with diagonal elements
given by $\Lambda^{(1)}_{y,ij}$, $(i,j) \in \mD$. Finally,  
 let $w$ be the $n \times (d_x+I+J)$ matrix with rows given by $w_{ij}'$, $(i,j) \in \mD$.
We define the $n \times n$ symmetric idempotent matrix
\begin{align*}
    Q_y :=  \left( \Lambda^{(1)}_{y} \right)^{1/2} w  \left( w' \Lambda^{(1)}_{y} w \right)^\dagger w'  \left( \Lambda^{(1)}_{y} \right)^{1/2} ,
\end{align*}
where $\dagger$ is the Moore-Penrose pseudoinverse.
For the elements of this matrix we write $Q_{y,ij,i'j'}$.
We have
$\left( Q_y s_y  \right)_{ij} =     \sum_{(i',j') \in \mD} Q_{y, ij, i'j'}  s_{y,i'j'}$.
The constraint 
$\exists\,  \theta: \, \pi_{y,ij} = w_{ij}' \theta_y$ in 
\eqref{DefHatPi} can then equivalently be written as\footnote{%
In matrix notation the constraint can be written as
$ \pi_y= w \, \theta_y$, and we thus have
$ Q_y  \left( \Lambda^{(1)}_{y} \right)^{1/2} \pi_y  = Q_y  \left( \Lambda^{(1)}_{y} \right)^{1/2} w \, \theta_y =  \left( \Lambda^{(1)}_{y} \right)^{1/2} w \, \theta_y  =  \left( \Lambda^{(1)}_{y} \right)^{1/2} \pi_y $, where we also used 
that $ Q_y  \left( \Lambda^{(1)}_{y} \right)^{1/2} w = \left( \Lambda^{(1)}_{y} \right)^{1/2}  w$, which follows from the
definition of $Q_y$.
} 
\begin{align}
    Q_y  \left( \Lambda^{(1)}_{y} \right)^{1/2} \pi_y    =  \left( \Lambda^{(1)}_{y} \right)^{1/2} \pi_y .
    \label{QpiProjection}
\end{align} 
The matrix $Q_y$ projects onto the column span of $ \left( \Lambda^{(1)}_{y} \right)^{1/2} w$.
This projector acts in the space  of weighted index vectors 
$\left[ \left( \Lambda^{(1)}_{y,ij} \right)^{1/2}  \pi_{y,ij} \, : \, (i,j) \in \mD \right]$,
and the weighting of each $ \pi_{y,ij}$ by  $\left( \Lambda^{(1)}_{y,ij} \right)^{1/2}$ is natural, because 
 $\Lambda^{(1)}_{y,ij} $ is simply the expected Hessian for observation $(i,j)$. 

\subsection{Technical Lemmas}
We require some results for the proofs of the main theorems below.
The following lemma provides an asymptotic expansion of  $\widehat \pi_{y,ij} - \pi^0_{y,ij} $.

\begin{lemma}[Score expansion of fixed effect estimates]
    \label{lemma:ScoreExpansionPi}
    Under Assumption~\ref{ass:baseline}, for  $y  \in  \mY$ and $(i,j) \in \mD$,  we have
    \begin{align*}
    \left( \Lambda^{(1)}_{y,ij} \right)^{1/2} \left( \widehat \pi_{y,ij} - \pi^0_{y,ij} \right)
     &= - \left( Q_y s_y  \right)_{ij}
        - \frac 1 2 
         \sum_{(i',j') \in \mD} 
                \,   Q_{y, ij, i'j'} \,
     \frac{ \Lambda^{(2)}_{y,i'j'} } { \left(  \Lambda^{(1)}_{y,i'j'} \right)^{3/2} }
        \left[ \left( Q_y s_y  \right)_{i'j'} \right]^2
        + r_{y,ij},
\end{align*}
and the remainder $r_{y,ij}$ satisfies 
 $\sup_{y \in  \mY} \max_{(i,j) \in \mD} \left| r_{y,ij} \right| = o_P(n^{-1/2})$.

\end{lemma}

The expansion in the preceding lemma is a second-order stochastic expansion,
because it does not only describe the terms linear in the score $s_y$,
but also the terms  quadratic in $s_y$. We need to keep track of those quadratic terms, because
they yield the leading order incidental parameter biases that appear in Theorem~\ref{th:expansion}.
The remainder $r_{y,ij}$ contains higher-order terms in $s_y$ (cubic, quartic, etc), which turn out not to matter for
the result in Theorem~\ref{th:expansion}. Note also that $\Lambda^{(2)}_{y,ij} = \partial_{\pi^3} \ell_{y,ij} $. Thus,
the term quadric in the score is proportional to the third derivative of the objective function. 

We now want to decompose the projector $Q_y$
into the parts stemming from $x_{ij}$, $e_{i,I}$ and $e_{j,J}$, respectively. 
We have already introduced  the $d_x$-vector $\widetilde x_{y,ij} = \widetilde x_{ij}(y)$
in Section~\ref{sec:AsyDistrUncorrected}.
Let $\widetilde x_y$ be the $n \times d_x$ matrix with rows given by $\widetilde x'_{y,ij}$, $(i,j) \in \mD$.
The $d_x \times d_x$ matrix  $W_y = W(y) = n^{-1} \widetilde x_y' \Lambda^{(1)}_{y} \widetilde x_y$
was also introduced in Section~\ref{sec:AsyDistrUncorrected}.
Invertibility of $W_y$ is guaranteed by Assumption~\ref{ass:baseline}$(vi)$,
and uniform boundedness of $\Lambda^{(1)}_{y,ij} $ and $\left(\Lambda^{(1)}_{y,ij} \right)^{-1}$,
as formalized by the following lemma.
\begin{lemma}[Invertibility of $W_y$]
    \label{lemma:InvW}
      Let Assumption~\ref{ass:baseline} hold. Then $\sup_{y \in  \mY}  \| W_y^{-1} \| = O_P(1)$.
\end{lemma}

Next, define $w^{(2)}_{ij} =  e_{i,I} $ and $w^{(3)}_{ij} = e_{j,J}$, and let $w^{(2)}$ and $w^{(3)}$ be the corresponding
$n \times I$ and $n \times J$ matrices with rows given by $w^{(2)'}_{ij}$ and $w^{(3)'}_{ij}$, respectively.
Let
\begin{align*}
    Q^{(1)}_y &:=  n^{-1} \, \left( \Lambda^{(1)}_{y} \right)^{1/2} \widetilde x_y  \, W_y^{-1}  \, \widetilde x_y'  \left( \Lambda^{(1)}_{y} \right)^{1/2} ,
    \\
        Q^{({\rm FE})}_y  &:= \left( \Lambda^{(1)}_{y} \right)^{1/2} 
    \left[w^{(2)},w^{(3)} \right]  \left( \left[w^{(2)},w^{(3)} \right]' \Lambda^{(1)}_{y} \left[w^{(2)},w^{(3)} \right] \right)^\dagger \left[w^{(2)},w^{(3)} \right]'  \left( \Lambda^{(1)}_{y} \right)^{1/2} .
\end{align*}
$\widetilde x_{y,ij}$ is defined as the part of $ x_{y,ij}$ that is orthogonal to the fixed effects under a metric given by
$\Lambda^{(1)}_{y,ij} $. We have $ Q^{({\rm FE})}_y  \left( \Lambda^{(1)}_{y} \right)^{1/2} \widetilde x_y = 0$,
which implies that 
\begin{align}
   Q_y = Q^{(1)}_y + Q^{({\rm FE})}_y
   \label{QisQ1QFE}
\end{align}
and
also
$Q^{(1)}_y  Q^{({\rm FE})}_y = Q^{({\rm FE})}_y Q^{(1)}_y = 0$.
Also, because $Q^{(1)}_y   \left( \Lambda^{(1)}_{y} \right)^{1/2}   \widetilde x_{y}
      =  \left( \Lambda^{(1)}_{y} \right)^{1/2}   \widetilde x_{y}$ and also $Q^{({\rm FE})}_y   \left( \Lambda^{(1)}_{y} \right)^{1/2}   \widetilde x_{y}
      = 0$,
we obtain      
\begin{align}
   Q_y   \left( \Lambda^{(1)}_{y} \right)^{1/2}   \widetilde x_{y} = \left(Q^{(1)}_y  + Q^{({\rm FE})}_y\right)   \left( \Lambda^{(1)}_{y} \right)^{1/2}   \widetilde x_{y}
      =  \left( \Lambda^{(1)}_{y} \right)^{1/2}   \widetilde x_{y} .
   \label{QtildeX}   
\end{align}
We have thus decomposed $Q_y$ into the component stemming
from the regressors and a component stemming from the fixed effects. 
For the elements of $ Q^{(1)}_y$,
\begin{align}
     Q^{(1)}_{y,ij,i'j'} &=   n^{-1} \left( \Lambda^{(1)}_{y,ij} \, \Lambda^{(1)}_{y,i'j'} \right)^{1/2}
               \widetilde x'_{y,ij}  \, W_y^{-1}  \, \widetilde x_{y,i'j'} .
     \label{ElementsQ1}          
\end{align}
Next, define the projection matrices
\begin{align*}
      Q^{(2)}_y  &:= \left( \Lambda^{(1)}_{y} \right)^{1/2}  w^{(2)}  \left( w^{(2) \, \prime} \Lambda^{(1)}_{y} w^{(2)} \right)^{-1} w^{(2) \, \prime}  \left( \Lambda^{(1)}_{y} \right)^{1/2} ,
      \\
    Q^{(3)}_y  &:= \left( \Lambda^{(1)}_{y} \right)^{1/2}  w^{(3)}  \left( w^{(3) \, \prime} \Lambda^{(1)}_{y} w^{(3)} \right)^{-1} w^{(3) \, \prime}  \left( \Lambda^{(1)}_{y} \right)^{1/2}  .
\end{align*}
Notice that $w^{(2) \, \prime} \Lambda^{(1)}_{y} w^{(2)}$
and $w^{(3) \, \prime} \Lambda^{(1)}_{y} w^{(3)}$ are simply diagonal $I \times I$ and $J \times J$ matrices 
with diagonal entries $\sum_{j \in \mD_i} \Lambda^{(1)}_{y,ij} $ and $\sum_{i \in \mD_j} \Lambda^{(1)}_{y,ij} $, respectively,
and therefore 
\begin{align}
       Q^{(2)}_{y,ij,i'j'} &= 1(i=i') \frac{ \left( \Lambda^{(1)}_{y,ij} \, \Lambda^{(1)}_{y,i j'} \right)^{1/2}} {\sum_{j'' \in \mD_i} \Lambda^{(1)}_{y,ij''}} ,
&
       Q^{(3)}_{y,ij,i'j'} &= 1(j=j') \frac{ \left( \Lambda^{(1)}_{y,ij} \, \Lambda^{(1)}_{y,i'j} \right)^{1/2}} {\sum_{i'' \in \mD_j} \Lambda^{(1)}_{y,i''j}} .
         \label{ElementsQ23}        
\end{align}
It is not exactly true that $Q^{({\rm FE})}_y$  equals $Q^{(2)}_y + Q^{(3)}_y$, but Lemma~\ref{lemma:PropQ} shows that this is approximately true
in a well-defined sense.


\begin{lemma}[Properties of $Q_y$]
     \label{lemma:PropQ}
       Under Assumption~\ref{ass:baseline},
       \begin{itemize}
       \item[(i)]
       $ Q_y = Q^{(1)}_y + Q^{({\rm FE})}_y $
       and $Q^{({\rm FE})}_y =   Q_y^{(2)} +  Q_y^{(3)} + Q_y^{({\rm rem})},$
            where
           $$   \sup_{y \in  \mY}  \max_{(i,j) \in \mD}  \max_{(i',j') \in \mD}   \left| Q_{y,ij,i'j'}^{({\rm rem})} \right| = O_P(n^{-1}).$$ 
 
       \item[(ii)] $  \sup_{y \in  \mY} \max_{(i,j) \in \mD}  \sum_{(i',j') \in \mD} \left| Q_{y, ij, i'j'} \right|  = O_P(1),$ and \\[3pt]
       $  \sup_{y \in  \mY} \max_{(i,j) \in \mD}  \sum_{(i',j') \in \mD} \left| Q^{({\rm FE})}_{y, ij, i'j'} \right|  = O_P(1).$
       
       \smallskip
       
       \item[(iii)] $  \sup_{y \in  \mY} \max_{(i,j) \in \mD}  \max_{(i',j') \in \mD} \left| Q_{y, ij, i'j'} \right|  = O_P(n^{-1/2}).$
      \end{itemize}  
                 
\end{lemma}

\begin{remark}[Bias of $\widehat \pi_{y,ij}$]
According to part (i) of this lemma the remainder term $Q_y^{({\rm rem})} = Q^{({\rm FE})}_y - Q^{(2)}_y - Q^{(3)}_y$ has elements
uniformly bounded of order $n^{-1}$, and it can easily be seen from \eqref{ElementsQ1} that the same is true for $Q_y^{(1)}$,
because the elements of $\widetilde x_{y}$ are also uniformly bounded under our assumptions.
 By contrast,
$Q_y^{(2)}$ and $Q_y^{(3)}$ have elements of order $J^{-1}$ and $I^{-1}$, respectively, that is, of order $n^{-1/2}$.
Using this and the fact that $s_{y,ij}$ has variance one and is independent across observations $(i,j)$
we find
\begin{align}
     \Ep \left[ \left( Q_y s_y  \right)_{ij} \right]^2
     &=   \sum_{(i',j') \in \mD}   [Q_{y,ij,i'j'}]^2
       = Q_{y,ij,ij}
    =   Q^{(2)}_{y,ij,ij} + Q^{(3)}_{y,ij,ij} + O_P(n^{-1})
  \nonumber  \\
    &=  \frac{\Lambda^{(1)}_{y,ij} }  {\sum_{j' \in \mD_i} \Lambda^{(1)}_{y,ij'} }
    + \frac{ \Lambda^{(1)}_{y,ij}} {  \sum_{i' \in \mD_j} \Lambda^{(1)}_{y,i'j}  } + O_P(n^{-1}) ,
  \label{EQsy}  
\end{align}
where we use that $Q_y$ is idempotent in the second step, and \eqref{ElementsQ23} in the third step.
Combining this with Lemma~\ref{lemma:ScoreExpansionPi} one finds that the leading order bias term
in $\widehat \pi_{y,ij} - \pi^0_{y,ij} $ is given by
\begin{align*}
      - \frac 1 2 
         \sum_{(i',j') \in \mD} 
                \,   Q_{y, ij, i'j'} \,
     \frac{ \Lambda^{(2)}_{y,i'j'} } {  \Lambda^{(1)}_{y,i'j'}  }
        \left[   \frac{1}  {\sum_{j' \in \mD_i} \Lambda^{(1)}_{y,ij'} }
    + \frac{1} {  \sum_{i' \in \mD_j} \Lambda^{(1)}_{y,i'j}  }      
        \right] ,
\end{align*}
which then translates into corresponding bias terms for all other estimators as well.
\end{remark}

For the following lemma, let $Z^{(\beta)}_y = Z^{(\beta)}(y)$, $Z^{(F)}_y = Z^{(F)}(y) $,
     $B^{(\beta)}_y=B^{(\beta)}(y)$, $D^{(\beta)}_y =D^{(\beta)}(y) $,
     $B^{(\Lambda)}_{y,k} =B^{(\Lambda)}_k(y) $ and
      $D^{(\Lambda)}_{y,k}=D^{(\Lambda)}_k(y)$ be as defined in and before
      Theorem~\ref{th:expansion} in the main text.

\begin{lemma}[Properties of score averages]
     \label{lemma:ScoreAve}
       Under Assumption~\ref{ass:baseline},
         \begin{itemize}
          \item[(i)] $\sup_{y \in  \mY} \max_{(i,j) \in \mD}  \left| \left( Q_y s_y  \right)_{ij}  \right| = o_P(n^{-1/6})$.
          
          \smallskip
          
          \item[(ii)] $-   W_y^{-1} \, n^{-1/2} \,  \sum_{(i,j) \in \mD}  \, \widetilde x_{y,ij} \, \partial_\pi \ell_{y,ij}  \rightsquigarrow  Z^{(\beta)}_y$,
          in $\ell^{\infty}(\mY)^{d_x}$.
          
                    \smallskip
          
          \item[(iii)] $- \frac 1 {\sqrt{n}} \sum_{(i,j) \in \mD}  \left[ \Psi_{y,ij} + (\partial_\beta F_{y}) \, W_y^{-1} \,  \widetilde x_{y,ij} \right]  \partial_\pi \ell_{y,ij}   \rightsquigarrow  Z^{(F)}_y$, in   $\ell^{\infty}(\mY)^{|\mK|}$.

                    \smallskip
                    
          \item[(iv)]  $ - \frac 1 2 
        W_y^{-1}  \; \frac 1 {\sqrt{n}}      \sum_{(i,j) \in \mD}   \; 
           \widetilde x_{y,ij}          
           \left( \Lambda^{(1)}_{y,ij} \right)^{-1}
      \Lambda^{(2)}_{y,ij} 
        \left[ \left( Q_y s_y  \right)_{ij} \right]^2
       -  \left( \frac{ I } {\sqrt{n}}    B^{(\beta)}_y
                   +  \frac{ J} {\sqrt{n}}      D^{(\beta)}_y \right)  \rightarrow_P 0
        $,
        uniformly in $y \in \mY$.
        
                  \smallskip
        
      \item[(v)]  $   \frac 1 {2 \sqrt{n}}  \sum_{(i,j) \in \mD}  
         \left( \Lambda^{(1)}_{y,ij} \right)^{-1}
         \left( \Lambda^{(2)}_{y,ij,k} - \Lambda^{(2)}_{y,ij} \Psi_{y,ij,k} \right) 
         \left[ \left( Q_y s_y  \right)_{ij}  \right]^2   
         - \left(  \frac{ I } {\sqrt{n}}     B^{(\Lambda)}_{y,k}
                   +  \frac{ J} {\sqrt{n}}     D^{(\Lambda)}_{y,k} \right)
                   \rightarrow_P 0
        $,
        uniformly in $y \in \mY$.
     \end{itemize}
\end{lemma}

Regarding part (i) of this lemma, 
notice that pointwise we have $ \left( Q_y s_y  \right)_{ij} = O_P(n^{-1/4})$,
because \eqref{EQsy} implies that $ \Ep \left[ \left( Q_y s_y  \right)_{ij} \right]^2 = O_P(n^{-1/2})$.
However, after taking the supremum over $y$, $i$, $j$ the term is growing faster than $n^{-1/4}$.
The rate $o_P(n^{-1/6})$ in part (i) of the lemma is crude, but sufficient for our purposes.

\begin{lemma}[Uniform Consistency of Estimators of Bias and Variance Components]
    \label{lemma:bias_estimators}
    Let Assumption~\ref{ass:baseline} hold.
    Then,
    \begin{align*}
        \sup_{y \in  \mY} \left\|  \widehat W(y) - \overline W(y) \right\| &= o_P(1),
     &
          \sup_{y \in \mY} \left\|  \partial_{\beta} \widehat F(y)  - \partial_{\beta} F(y)  \right\| &= o_P(1),
\\
        \sup_{y \in \mY} \left\|  \widehat B^{(\beta)}(y) - B^{(\beta)}(y) \right\| &= o_P(1),
        &
         \sup_{y \in  \mY} \left\|  \widehat D^{(\beta)}(y) - D^{(\beta)}(y) \right\| &= o_P(1),
\\
        \sup_{y \in \mY}  \left\|  \widehat B^{(\Lambda)}(y) - B^{(\Lambda)}(y)  \right\| &= o_P(1),
        &
         \sup_{y \in \mY}  \left\|  \widehat D^{(\Lambda)}(y) - D^{(\Lambda)}(y)  \right\| &= o_P(1),
\\
        \sup_{y \in  \mY} \left\|   \widehat  \Omega(y) - \overline \Omega(y) \right\| &= o_P(1),
    \end{align*}
  where $\|\cdot\|$ denotes the Frobenius matrix norm, i.e.  $\|A\| = \text{trace}(AA')^{1/2}$ for a matrix $A$.
\end{lemma}

As already mentioned above, the proof of the technical lemmas that we have stated here is provided in the Supplementary Appendix.

\subsection{Proof of Main Text Theorems}\label{app:pthms}

\begin{proof}[\bf Proof of Theorem~\ref{th:expansion}]
   \underline{\#  Part 1: FCLT for $\widehat \beta_y =\widehat \beta(y) $.} \\
   The definition of $\widetilde x_y$ implies that
     $\sum_{i \in \mD_j}  \Lambda^{(1)}_{y,ij} \widetilde x_{y,ij}   = 0$
     and $\sum_{j \in \mD_i}    \Lambda^{(1)}_{y,ij} \widetilde x_{y,ij} = 0$,
     and $n^{-1} \sum_{(i,j) \in \mD}  \allowbreak   \Lambda^{(1)}_{y,ij}  \widetilde x_{y,ij}  x_{ij}' 
     = n^{-1} \sum_{(i,j) \in \mD}  \Lambda^{(1)}_{y,ij}   \widetilde x_{y,ij}      \widetilde  x_{ij}'  = W_y$.
     Using this and 
     \begin{align*}
            \widehat \pi_{y,ij} - \pi^0_{y,ij} := x_{ij}'  \left(  \widehat \beta_y - \beta^0_y \right) 
            + \left( \widehat \alpha_{y,i} -  \alpha^0_{y,i} \right) 
             +  \left( \widehat \gamma_{y,j} - \gamma^0_{y,j} \right) 
     \end{align*}
     we obtain
     \begin{align*}
         n^{-1} \sum_{(i,j) \in \mD} \; \widetilde x_{y,ij} \,   \Lambda^{(1)}_{y,ij}     \, \left(  \widehat \pi_{y,ij} - \pi^0_{y,ij} \right)
         &=   n^{-1} \sum_{(i,j) \in \mD} \; \widetilde x_{y,ij} \,   \Lambda^{(1)}_{y,ij} \,  x_{ij}'  \left(  \widehat \beta_y - \beta^0_y \right) 
           =   W_y \left(  \widehat \beta_y - \beta^0_y \right)  ,
     \end{align*}
     and therefore     
     \begin{align*}
               \widehat \beta_y - \beta^0_y 
               &=  W_y^{-1} \;  n^{-1} \sum_{(i,j) \in \mD} \; \widetilde x_{y,ij} \,   \Lambda^{(1)}_{y,ij}     \, \left(  \widehat \pi_{y,ij} - \pi^0_{y,ij} \right) .
     \end{align*}
     By combining this with Lemma~\ref{lemma:ScoreExpansionPi}  we obtain
    \begin{align}
        \sqrt{n}\left( \widehat \beta_y - \beta^0_y \right)
        &=   T^{(1,\beta)}_{y} + T^{(2,\beta)}_{y} + r^{(\beta)}_{y} ,
        \label{ExpansionBeta}
     \end{align}
     where
     \begin{align*}
          T^{(1,\beta)}_{y} &:= 
          - n^{-1/2} \; W_y^{-1} \; \sum_{(i,j) \in \mD}   \left( \Lambda^{(1)}_{y,ij} \right)^{1/2} \;  \widetilde x_{y,ij}  \,   \left( Q_y s_y  \right)_{ij} ,
          \\
          T^{(2,\beta)}_{y} &:= 
           - \frac 1 2 
          n^{-1/2} \; W_y^{-1} \; \sum_{(i,j) \in \mD}    \left( \Lambda^{(1)}_{y,ij} \right)^{1/2} \;  \widetilde x_{y,ij}  \, 
          \sum_{(i',j') \in \mD}    \,   Q_{y, ij, i'j'} \,
     \frac{ \Lambda^{(2)}_{y,i'j'} } { \left(  \Lambda^{(1)}_{y,i'j'} \right)^{3/2} }     
        \left[ \left( Q_y s_y  \right)_{i'j'} \right]^2  ,
     \end{align*}
     and  $r^{(\beta)}_{y} :=  W_y^{-1} \, n^{-1/2}  \sum_{(i,j) \in \mD}  \widetilde  x_{y,ij} \,   \left(\Lambda^{(1)}_{y,ij} \right)^{1/2} \,    r_{y,ij}$
     satisfies
     \begin{align*}
               \sup_{y \in  \mY}   \left| r^{(\beta)}_{y} \right| 
             &\leq   
             \underbrace{   \left( \sup_{y \in  \mY} \, W_y^{-1}
               \, n^{-1/2} \, \sum_{(i,j) \in \mD}  \left| \widetilde  x_{y,ij} \right| \left|  \left(\Lambda^{(1)}_{y,ij} \right)^{1/2}  \right| \right)
               }_{= O_P(n^{1/2})}
             \underbrace{  \left( \sup_{y \in  \mY}  \max_{(i,j) \in \mD}   \left| r_{y,ij} \right| \right) }_{= o_P(n^{-1/2})}
               = o_P(1) ,
     \end{align*}
     where we also use that $ \Lambda^{(1)}_{y,ij}$ and $\widetilde  x_{y,ij}$ are uniformly bounded under our assumptions.
     For the term linear in the score we find
     \begin{align*}
         T^{(1,\beta)}_{y}
         &=  - n^{-1/2}
         W_y^{-1} \, \widetilde x'_{y} \,
           \left( \Lambda^{(1)}_{y} \right)^{1/2}   Q_y s_y
         = - n^{-1/2}
         W_y^{-1} \, \widetilde x'_{y} \,  \left( \Lambda^{(1)}_{y} \right)^{1/2} \, s_y 
       \\
          &=  - 
         W_y^{-1} \, n^{-1/2} \,  \sum_{(i,j) \in \mD}  \, \widetilde x_{y,ij} \, \partial_\pi \ell_{y,ij} 
         \rightsquigarrow  Z^{(\beta)}_y  ,
     \end{align*}
     where in the second step
      we used \eqref{QtildeX},
      and the final step follows from  part $(ii)$ of Lemma~\ref{lemma:ScoreAve}.

     Employing again \eqref{QtildeX}
      we  find 
      \begin{align*}
            T^{(2,\beta)}_{y} &:= 
           - \frac 1 2 \; W_y^{-1} \,
          n^{-1/2}     \sum_{(i',j') \in \mD}    \,  
          \underbrace{ \sum_{(i,j) \in \mD}    \left( \Lambda^{(1)}_{y,ij} \right)^{1/2} \; \widetilde x_{y,ij}  \, 
         Q_{y, ij, i'j'} }_{= \left( \Lambda^{(1)}_{y,i'j'} \right)^{1/2} \; \widetilde x_{y,i'j'}}
     \frac{ \Lambda^{(2)}_{y,i'j'} } { \left(  \Lambda^{(1)}_{y,i'j'} \right)^{3/2} }     
        \left[ \left( Q_y s_y  \right)_{i'j'} \right]^2
     \\   
           &= 
             - \frac 1 2 
        W_y^{-1}  \;   n^{-1/2}     \sum_{(i,j) \in \mD}   \; 
           \widetilde x_{y,ij}          
     \frac{ \Lambda^{(2)}_{y,ij} } {  \Lambda^{(1)}_{y,ij}  }     
        \left[ \left( Q_y s_y  \right)_{ij} \right]^2 ,
      \end{align*}
      and according to part $(iv)$ of Lemma~\ref{lemma:ScoreAve} we thus have
      \begin{align*}
           T^{(2,\beta)}_{y}
           - \left(  \frac{I} {n^{1/2}}  \,  B^{(\beta)}_y
                   +  \frac{J} {n^{1/2}}    \,  D^{(\beta)}_y \right)      
        & \rightarrow_P   0,
     \end{align*}
     uniformly in $y \in \mY$.
     Combining the above gives the result for  $\sqrt{n}\left( \widehat \beta_y - \beta^0_y \right)$ in the theorem.
  
 \smallskip
\noindent  
\underline{ \#  Part 2: FCLT for $\widehat F_{y,k} =\widehat F_{k}(y) $.}  \\
 Let
 $ \pi_{y,ij,k}^0 :=   \pi_{y,ij}^0
       + (\mathbbm{x}_{ij,k} - x_{ij})'  \beta_y^0 $
and 
       $\widehat \pi_{y,ij,k} := \widehat \pi_{y,ij}
       + (\mathbbm{x}_{ij,k} - x_{ij})' \widehat \beta_y $.
   Because $\mathbbm{x}_{ij,k} - x_{ij} = \widetilde{\mathbbm{x}}_{y,ij,k} - \widetilde x_{y,ij}$
   we have
    \begin{align}
         \widehat \pi_{y,ij,k} -  \pi_{y,ij,k}^0  
         &= \widehat \pi_{y,ij} -  \pi_{y,ij}^0
       + (\widetilde{\mathbbm{x}}_{y,ij,k} - \widetilde x_{y,ij})' ( \widehat \beta_y -\beta_y^0 ) .
       \label{ExpansionPiKbasic}
    \end{align}
    Using \eqref{QpiProjection} and $Q^{(1)}_y \left( \Lambda^{(1)}_{y} \right)^{1/2} \pi_y = \left( \Lambda^{(1)}_{y} \right)^{1/2} \widetilde x_y \beta_y$ for any $\pi_y = w \theta_y$,
    \begin{align*}
         \widehat \pi_y  - \pi^0_y 
    &=  \left( \Lambda^{(1)}_{y} \right)^{-1/2} 
    \underbrace{ \left(Q^{(1)}_y+Q^{(\rm FE)}_y \right) }_{=Q_y} \left( \Lambda^{(1)}_{y} \right)^{1/2} ( \widehat \pi_y  - \pi^0_y) 
    \\
     &= \left( \Lambda^{(1)}_{y} \right)^{-1/2} 
    Q^{(\rm FE)}_y  \left( \Lambda^{(1)}_{y} \right)^{1/2} ( \widehat \pi_y  - \pi^0_y) 
    + \widetilde x_{y}  ( \widehat \beta_y -\beta_y^0 ) .
    \end{align*}
    Combining the above gives
    \begin{align*}
         \widehat \pi_{y,ij,k} -  \pi_{y,ij,k}^0  
         &= \left[ \left( \Lambda^{(1)}_{y} \right)^{-1/2} 
    Q^{(\rm FE)}_y  \left( \Lambda^{(1)}_{y} \right)^{1/2} ( \widehat \pi_y  - \pi^0_y) 
    \right]_{ij}
       +  \widetilde{\mathbbm{x}}_{y,ij,k}' ( \widehat \beta_y -\beta_y^0 ) .
    \end{align*}
    Using     Lemma~\ref{lemma:ScoreExpansionPi} and the properties of $Q_y$, $Q^{(1)}_y$ and $Q^{(\rm FE)}_y$,  we thus find    
    \begin{align}
         \left( \Lambda^{(1)}_{y,ij} \right)^{1/2} 
         \left(   \widehat \pi_{y,ij,k} -  \pi_{y,ij,k}^0   \right)
         &= - \left(  Q^{(\rm FE)}_y  s_y \right)_{ij}
       - \frac 1 2 
         \sum_{(i',j') \in \mD} 
                \,   Q^{(\rm FE)}_{y, ij, i'j'} \,
     \frac{ \Lambda^{(2)}_{y,i'j'} } { \left(  \Lambda^{(1)}_{y,i'j'} \right)^{3/2} }
        \left[ \left( Q_y s_y  \right)_{i'j'} \right]^2
      \nonumber  \\ & \qquad
        + \left( Q^{(\rm FE)} r_{y} \right)_{ij}
       +   \left( \Lambda^{(1)}_{y,ij} \right)^{1/2}   \widetilde{\mathbbm{x}}_{y,ij,k}' ( \widehat \beta_y -\beta_y^0 ) .
       \label{ExpansionPiK}
    \end{align}       
    Next, by  expanding          $\Lambda(\widehat \pi_{y,ij,k})$ in $\widehat \pi_{y,ij,k}$ around $ \pi_{y,ij,k}^0$ we find
  \begin{align*}
      \widehat F_{y,k} -  F_{y,k}  &= n^{-1} \sum_{(i,j) \in \mD}  \left[ \Lambda(\widehat \pi_{y,ij,k}) - \Lambda(\pi^0_{y,ij,k}) \right]
        \\
        & = 
        n^{-1} \sum_{(i,j) \in \mD}  \bigg[
        \Lambda^{(1)}_{y,ij,k} \left( \widehat \pi_{y,ij,k} - \pi^0_{y,ij,k} \right)
           + \frac 1 2 \Lambda^{(2)}_{y,ij,k} \left( \widehat \pi_{y,ij,k} - \pi^0_{y,ij,k} \right)^2
        \\ & \qquad \qquad \qquad \qquad    \qquad \qquad \qquad  
           + \frac 1 6  \Lambda^{(3)}(\widetilde \pi_{y,ij,k}) \left( \widehat \pi_{y,ij,k} - \pi^0_{y,ij,k} \right)^3 \bigg],
  \end{align*}
   where
   $\widetilde \pi_{y,ij,k}$ is some value between $\widehat \pi_{y,ij,k}$ and $ \pi^0_{y,ij,k}$,
   and we use the notation   
    $\Lambda^{(\ell)}_{y,ij,k}   =  \Lambda^{(\ell)}(\pi^0_{y,ij,k})  $, 
   which corresponds to $\Lambda^{(\ell)}_{ij,k}(y)$ 
     in the main text.
    By appropriately 
    inserting \eqref{ExpansionPiKbasic} and \eqref{ExpansionPiK} into this expansion, 
    also using \eqref{ExpansionBeta},
    and
    sorting by terms linear in $s_y$, quadratic in $s_y$, and remainder,
    we find
    \begin{align}
        \sqrt{n}\left(  \widehat F_{y,k} -  F_{y,k}  \right)
        &=  T^{(1,F)}_{y,k} + T^{(2,F)}_{y,k} + r^{(F)}_{y,k} ,
         \label{ExpandFhatproof}    
    \end{align}
    where
   the terms linear in $s_y$ read
    \begin{align*}
         T^{(1,F)}_{y,k}  &=  - \frac 1 {\sqrt{n}} \sum_{(i,j) \in \mD}
          \Lambda^{(1)}_{y,ij,k}  \left[       
         \frac{   \left( Q^{(\rm FE)}_y s_y  \right)_{ij} } {  \left( \Lambda^{(1)}_{y,ij} \right)^{1/2}}
           +      \widetilde{\mathbbm{x}}_{y,ij,k}' W_y^{-1} 
             \frac 1 n  \sum_{(i',j') \in \mD}   \widetilde x_{y,i'j'} \,  \partial_\pi \ell_{y,i'j'} 
             \right],
    \end{align*}
    with $\partial_\pi \ell_{y,i'j'} = \left( \Lambda^{(1)}_{y,i'j'} \right)^{1/2} s_{y,i'j'}$.
     
     The projection $\Psi_{y,ij,k} = \Psi_{ij,k}(y)$, defined just before \eqref{DefPsi} in the main text,    
      can be written in terms of the matrix $Q^{({\rm FE})}_y$
    as 
    \begin{align}
         \Psi_{y,ij,k} &=
          \left( \Lambda^{(1)}_{y,ij} \right)^{-1/2}
           \sum_{(i',j') \in \mD}  
         Q^{({\rm FE})}_{y,ij,i'j'}
             \frac{     \Lambda^{(1)}_{y,i'j',k} } {  \left( \Lambda^{(1)}_{y,i'j'} \right)^{1/2}}       ,
             \label{RelationPsi}
    \end{align}
    which implies that
    $\sum_{(i,j) \in \mD}   \Psi_{y,ij,k}  \partial_\pi \ell_{y,ij} = \sum_{(i,j) \in \mD}    \Lambda^{(1)}_{y,ij,k}  \left( \Lambda^{(1)}_{y,ij} \right)^{-1/2} \left( Q^{({\rm FE})}_y s_y  \right)_{ij}
    $. Using 
    $\partial_\beta F_{y,k} = \partial_\beta F_k(y) = n^{-1} 
    \sum_{(i,j) \in \mD}  \Lambda^{(1)}_{ij,k}(y)  \, \widetilde{\mathbbm{x}}_{y,ij,k}^{\, \prime}$ we obtain
    \begin{align}
        T^{(1,F)}_{y,k} &= - \frac 1 {\sqrt{n}} \sum_{(i,j) \in \mD}  \left( \Psi_{y,ij,k} + \partial_\beta F_{y,k} W_y^{-1}  \widetilde x_{y,ij} \right)  \partial_\pi \ell_{y,ij} .
                     \label{Bound1here}   
    \end{align}
    According to part $(iii)$ of Lemma~\ref{lemma:ScoreAve} the
    vector $T^{(1,F)}_{y} = \left[ T^{(1,F)}_{y,k} \, : \, k \in \mK \right]$ therefore satisfies
    $T^{(1,F)}_{y} \rightsquigarrow  Z^{(F)}_y$ asymptotically.
    
    The terms quadratic in $s_y$ read
    \begin{align*}
         T^{(2,F)}_{y,k} 
            &=
            - \frac 1 2 \, \frac 1 {\sqrt{n}} 
            \sum_{(i,j) \in \mD}
          \Lambda^{(1)}_{y,ij,k}    
       \left( \Lambda^{(1)}_{y,ij} \right)^{-1/2}
         \sum_{(i',j') \in \mD} 
                \,   Q^{(\rm FE)}_{y, ij, i'j'} \,
     \frac{ \Lambda^{(2)}_{y,i'j'} } { \left(  \Lambda^{(1)}_{y,i'j'} \right)^{3/2} }
        \left[ \left( Q_y s_y  \right)_{i'j'} \right]^2
        \\ & \quad
        +  \frac 1 2  \, \frac 1 {\sqrt{n}}  \sum_{(i,j) \in \mD}  
          \frac{ \Lambda^{(2)}_{y,ij,k}} { \Lambda^{(1)}_{y,ij}  } \left[ \left( Q_y s_y  \right)_{ij}  \right]^2   
         +  (\partial_\beta F_{y,k}) 
              T^{(2,\beta)}_{y} ,
    \end{align*}     
    where for the term 
      quadratic in $\widehat \pi_{y,ij,k} - \pi^0_{y,ij,k} $
      in the expansion of $ \widehat F_{y,k} -  F_{y,k} $
      we do not insert \eqref{ExpansionPiK}
      but rather insert \eqref{ExpansionPiKbasic},
      and we ignore the terms involving  $ \widehat \beta_y -\beta_y^0$ here  --- they give
       contributions quadratic in the score $s_y$, but only of smaller order, and 
       we therefore rather include those in the remainder term $r^{(F)}_{y,k}$ below.
       Using again \eqref{RelationPsi} we find
    \begin{align*}
         T^{(2,F)}_{y,k} 
            &=   \frac 1 2  \, \frac 1 {\sqrt{n}}  \sum_{(i,j) \in \mD}  
          \frac{ \Lambda^{(2)}_{y,ij,k} - \Lambda^{(2)}_{y,ij} \Psi_{y,ij,k}} { \Lambda^{(1)}_{y,ij}  } \left[ \left( Q_y s_y  \right)_{ij}  \right]^2   
         +  (\partial_\beta F_{y,k}) 
              T^{(2,\beta)}_{y} .
    \end{align*}
    Using  part $(v)$ of Lemma~\ref{lemma:ScoreAve}, and our previous result for $T^{(2,\beta)}_{y} $,   we thus obtain
    \begin{align}               
         T^{(2,F)}_{y,k} 
         -   \frac{I} {n^{1/2}}  \left[  B^{(\Lambda)}_{y,k} + (\partial_\beta F_{y,k})  B^{(\beta)}_y  \right]
                   -  \frac{J} {n^{1/2}}    \left[  D^{(\Lambda)}_{y,k} + (\partial_\beta F_{y,k})  D^{(\beta)}_y  \right] 
                    & \rightarrow_P 0,
               \label{Bound2here}   
    \end{align}    
     uniformly in $y \in \mY$ and $k \in \mK$.
    
     The remainder term of the expansion reads
     \begin{align*}
           r^{(F)}_{y,k}
           &=   n^{-1/2} \sum_{(i,j) \in \mD}  \bigg\{
        \Lambda^{(1)}_{y,ij,k}     \left( \Lambda^{(1)}_{y,ij} \right)^{-1/2}  \left( Q^{(\rm FE)} r_{y} \right)_{ij}
       + n^{-1/2}  \Lambda^{(1)}_{y,ij,k}    \widetilde{\mathbbm{x}}_{y,ij,k}'  r^{(\beta)}_{y}    
       \\ &     \qquad \qquad \qquad \quad
           + \frac 1 8 \Lambda^{(2)}_{y,ij,k} 
              \left( \Lambda^{(1)}_{y,ij} \right)^{-1}
           \left[ 
            {\textstyle \sum_{(i',j') \in \mD} }
                \,   Q_{y, ij, i'j'} \,
    \left(  \Lambda^{(1)}_{y,i'j'} \right)^{-3/2}  \Lambda^{(2)}_{y,i'j'} 
         \left( Q_y s_y  \right)_{i'j'}^2           
           \right]^2
       \\ &     \qquad \qquad \qquad \quad
           + \frac 1 2 \Lambda^{(2)}_{y,ij,k}  \left( \Lambda^{(1)}_{y,ij} \right)^{-1}  (r_{y,ij})^2
           + \frac 1 2 \Lambda^{(2)}_{y,ij,k} \left[  ( \mathbbm{x}_{ij,k} -  x_{ij})' ( \widehat \beta_y -\beta_y^0 ) \right]^2
       \\ &     \qquad \qquad \qquad \quad
           + \frac 1 6  \Lambda^{(3)}(\widetilde \pi_{y,ij,k})
            \left[ \widehat \pi_{y,ij} -  \pi_{y,ij}^0
       + ( \mathbbm{x}_{ij,k} -  x_{ij})' ( \widehat \beta_y -\beta_y^0 )  \right]^3 \bigg\}.
     \end{align*}
     Our assumptions guarantee that $\Lambda^{(\ell)}_{y,ij} $
     and $\Lambda^{(\ell)}_{y,ij,k} $, $\ell \in \{1,2,3\}$,
     and $\left( \Lambda^{(1)}_{y,ij} \right)^{-1}$ are all uniformly bounded.
     Lemma~\ref{lemma:ScoreExpansionPi} guarantees that $r_{y,ij} = o_P(n^{-1/2})$, uniformly over $y,i,j$,
     and using
    Lemma~\ref{lemma:PropQ}$(ii)$
    this also implies that
     $ \left( Q^{(\rm FE)} r_{y} \right)_{ij} = o_P(n^{-1/2})$, uniformly over $y,i,j$.
     Above we have shown $ r^{(\beta)}_{y}   = o_P(1)$,  uniformly over $y$.
     Lemma~\ref{lemma:PropQ}$(ii)$
     and Lemma~\ref{lemma:ScoreAve}$(i)$ imply that
     \begin{align*}
         \sup_{y \in  \mY} 
         \max_{(i,j) \in \mD} \left[ 
            {\textstyle \sum_{(i',j') \in \mD} }
                \,   Q_{y, ij, i'j'} \,
    \left(  \Lambda^{(1)}_{y,i'j'} \right)^{-3/2}  \Lambda^{(2)}_{y,i'j'} 
         \left( Q_y s_y  \right)_{i'j'}^2           
           \right]^2
           = o_P(n^{-1+1/3}) = o_P(n^{-1/2}) .
     \end{align*}
     Our  asymptotic result for $\widehat \beta_y$ from part 1 of this proof  guarantees that 
     $ \sup_{y \in  \mY}  \| \widehat \beta_y -\beta_y^0 \|^2 = o_P(n^{-1/2})$.
    Lemma~\ref{lemma:ScoreExpansionPi} together with 
    Lemma~\ref{lemma:PropQ}$(ii)$
     and Lemma~\ref{lemma:ScoreAve}$(i)$ guarantee that 
     $ \widehat \pi_{y,ij} -  \pi^0_{y,ij} = o_P(n^{-1/6})$, uniformly over $y,i,j$. 
     We thus find, uniformly over $y \in  \mY$ and $k \in \mK$,
     \begin{align*}
           &   \left|  r^{(F)}_{y,k}  \right|
           \\
           &\leq  \frac 1 {\sqrt{n}}
        \underbrace{
           \left[ \sum_{(i,j) \in \mD}  
      \left| \frac{  \Lambda^{(1)}_{y,ij,k}   } { \left( \Lambda^{(1)}_{y,ij} \right)^{1/2} }
      \right| 
        \right]
        }_{=O_P(n)}
        \underbrace{
       \left[ \max_{(i,j) \in \mD}  
       \left|  \left( Q^{(\rm FE)} r_{y} \right)_{ij} \right|
       \right]  
       }_{=o_P(n^{-1/2})}
       + 
       \underbrace{ \left( \frac 1 n \sum_{(i,j) \in \mD}  
       \left\|  \Lambda^{(1)}_{y,ij,k}    \widetilde{\mathbbm{x}}_{y,ij,k} \right\| \right)
       }_{=O_P(1)}
      \underbrace{  \left\|  r^{(\beta)}_{y}  \right\|  }_{=o_P(1)}
       \\ &    
           + \frac 1 {8 \sqrt{n}} 
           \underbrace{
          \left(   \sum_{(i,j) \in \mD}  
         \left|   \frac{  \Lambda^{(2)}_{y,ij,k}  }
              { \Lambda^{(1)}_{y,ij}}
           \right|   
                         \right)
              }_{=O_P(n)}
       \underbrace{
        \left\{  \max_{(i,j) \in \mD}     \left[ 
            {\textstyle \sum_{(i',j') \in \mD} }
                \,   Q_{y, ij, i'j'} \,
    \left(  \Lambda^{(1)}_{y,i'j'} \right)^{-3/2}  \Lambda^{(2)}_{y,i'j'} 
         \left( Q_y s_y  \right)_{i'j'}^2           
           \right]^2
           \right\}
           }_{=o_P(n^{-1/2})}
       \\ &    
           + \frac 1 {2 \sqrt{n}} 
            \underbrace{
          \left(   \sum_{(i,j) \in \mD}  
          \left|  \frac{  \Lambda^{(2)}_{y,ij,k}  }
              { \Lambda^{(1)}_{y,ij}}
              \right|
              \right)
              }_{=O_P(n)}
              \underbrace{ \left[  \max_{(i,j) \in \mD}    (r_{y,ij})^2
                 \right]
           }_{=o_P(n^{-1})}
           + \frac 1 {2 \sqrt{n}} 
         \underbrace{
           \sum_{(i,j) \in \mD}  
          \left|  \Lambda^{(2)}_{y,ij,k} \right|  \left\|  \mathbbm{x}_{ij,k} -  x_{ij}) \right\|^2
         }_{=O_P(n)}
          \underbrace{
             \left\|  \widehat \beta_y -\beta_y^0  \right\|^2
             }_{=o_P(n^{-1/2)}}
       \\ &    
           + \frac 4 {3 \sqrt{n}} 
        \underbrace{  \left(  \sum_{(i,j) \in \mD}  \left|  \Lambda^{(3)}(\widetilde \pi_{y,ij,k}) \right| \right)
        }_{=O_P(n)}
          \bigg\{  \underbrace{  \max_{(i,j) \in \mD}    \left| \widehat \pi_{y,ij} -  \pi_{y,ij}^0 \right|^3 }_{=o_P(n^{-1/2})}
       +   \underbrace{   \max_{(i,j) \in \mD}     \left\| \mathbbm{x}_{ij,k} -  x_{ij} \right\|^3
       }_{=O_P(1)} 
       \underbrace{ \left\| \widehat \beta_y -\beta_y^0 )  \right\|^3 
       }_{=o_P(n^{-1/2})}
       \bigg\} ,
     \end{align*}
      and therefore     
     \begin{align}
               \sup_{y \in  \mY, k \in \mK}   \left|  r^{(F)}_{y,k} \right| = o_P(1) .
            \label{Bound3here}   
     \end{align}
     Combing \eqref{ExpandFhatproof}, \eqref{Bound1here}, \eqref{Bound2here} and \eqref{Bound3here}
      gives the statement for  $\widehat F(y) - F(y)$ in the theorem.
\end{proof}

\begin{proof}[\bf Proof of Theorem~\ref{th:bc}]
 The theorem follows from  Theorem~\ref{th:expansion}  by
applying  Lemma~\ref{lemma:bias_estimators}, which provides the uniform consistency of the  estimators of the components of the asymptotic bias and variance functions.
\end{proof}

\end{appendix}


\pagebreak

\setcounter{page}{1}
\pagenumbering{roman}
\renewcommand{\thesection}{S.\arabic{section}}
\setcounter{section}{0}
\renewcommand{\theequation}{S.\arabic{equation}}
\renewcommand{\thetheorem}{S.\arabic{theorem}}
\setcounter{theorem}{0}
\renewcommand{\thecorollary}{S.\arabic{corollary}}
\setcounter{corollary}{0}
\renewcommand{\thelemma}{S.\arabic{lemma}}
\setcounter{lemma}{0}
\setcounter{equation}{0}

\begin{center}
{\bf \Large Supplementary Appendix}
\end{center}

\begin{abstract}  This supplementary material contains the proofs of Lemmas \ref{lemma:ScoreExpansionPi}--\ref{lemma:bias_estimators}, together with some technical intermediate results.
\end{abstract}

The following proof of Lemma~\ref{lemma:ScoreExpansionPi} also relies on 
the results of Lemma~\ref{lemma:PropQ}  and Lemma~\ref{lemma:ScoreAve}$(i)$, whose proof is presented afterwards, 
 without using Lemma~\ref{lemma:ScoreExpansionPi} of course.

\begin{proof}[\bf Proof of Lemma~\ref{lemma:ScoreExpansionPi}]
Define $Q^\perp_y := \mathbb{I}_n - Q_y$, which is the $n \times n$ symmetric idempotent matrix that projects
onto the space orthogonal to the column span of  $\left( \Lambda^{(1)}_{y} \right)^{1/2} w $, with $w=(w_{ij}: (i,j) \in \mD)$.
 In component notation we have
$Q^\perp_{y,ij,i'j'} = \delta_{ii'} \delta_{j j'} - Q_{y,ij,i'j'}$, where
$\delta_{..}$ refers to the Kronecker delta.
We also define
\begin{align*}
       \pi^*_{y,ij} &:= \left( \Lambda^{(1)}_{y,ij} \right)^{1/2} \pi_{y,ij} ,
      &
          \pi^{* \, 0}_{y,ij} &:= \left( \Lambda^{(1)}_{y,ij} \right)^{1/2} \pi^0_{y,ij} ,
      &
        \ell^*_{y,ij}( \pi^*_{y,ij}) &:= \ell_{y,ij}\left[  \left( \Lambda^{(1)}_{y,ij} \right)^{-1/2}   \pi^*_{y,ij} \right] ,
\end{align*} 
which is simply a rescaling of $\pi_{y,ij}$ by $ \left( \Lambda^{(1)}_{y,ij} \right)^{1/2}$. The rescaling is infeasible, because
$ \left( \Lambda^{(1)}_{y,ij} \right)^{1/2}$ depends on the true parameter values, but for the analysis here it is more convenient
to work with $ \ell^*_{y,ij}( \pi^*_{y,ij})$ than with  $\ell_{y,ij}( \pi_{y,ij})$.
After the rescaling we have  $s_{y,ij} =  \partial_{ \pi^*}   \ell^*_{y,ij} := \partial_{ \pi^*}   \ell^*_{y,ij}\left(   \pi^{*\, 0}_{y,ij} \right)$ and $1 =   \partial_{ \pi^{*2}}   \ell^*_{y,ij} :=  \partial_{ \pi^{*2}}   \ell^*_{y,ij}\left(   \pi_{y,ij}^{*\, 0} \right)$, that is,
the variance of the score and the Hessian of $ \ell^*_{y,ij}( \pi^*_{y,ij})$ evaluated at the true parameter values   are normalized to one.
Equation \eqref{QpiProjection} can be rewritten as
$ Q^\perp_y  \pi^*_y = 0 $.
where $\pi^*_y$ is the $n$-vector with elements $ \pi^*_{y,ij}$.
 Solving \eqref{DefHatPi} is then equivalent to minimizing
the function
\begin{align*}
      \sum_{(i,j) \in \mD}    \ell^*_{y,ij}( \pi^*_{y,ij})
      +  \sum_{(i,j) \in \mD}  \;  \pi^*_{y,ij} \, \; \sum_{(i',j') \in \mD} 
     \left[   Q^\perp_{y,ij,i'j'} \,     \mu_{y,i'j'}    \right]
\end{align*}
over $ \pi^*_{y}$ and $\mu_{y}$, where the $\mu_{y}$ are the Lagrange multipliers corresponding to
the constraint $Q^\perp_y   \pi^*_{y} = 0$, which is equivalent to existence of $\theta$ such that
$\pi_{y} = w \; \theta_y$. The FOCs with respect to $ \pi^*_{y}$ read
\begin{align*}
     \partial_{ \pi^*}  \ell^*_{y}( \widehat \pi^*_{y})
     +  Q^\perp_{y} \,     \widehat \mu_{y}  
     &= 0 ,
\end{align*}
where $ \partial_{ \pi^*}  \ell^*_{y}( \widehat \pi^*_{y})$ and $ \widehat \mu_{y} $ are $n$-vectors obtained by stacking the elements of  $\partial_{ \pi^*}  \ell^*_{y,ij}( \widehat \pi^*_{y})$ and $\widehat \mu_{y,ij}$ for all $(i,j) \in \mD$. 
Existence of $\widehat \mu_{y}  $ that satisfy those FOCs   is equivalent to
\begin{align}
      Q_{y} \partial_{ \pi^*}  \ell^*_{y}( \widehat \pi^*_{y})
     +  Q_{y} Q^\perp_{y} \,     \widehat \mu_{y}   = Q_{y} \partial_{ \pi^*}  \ell^*_{y}( \widehat \pi^*_{y})  = \sum_{(i',j') \in \mD}    Q_{y,ij,i'j'}  \, \partial_{\pi^*}  \ell^*_{y,i'j'}( \widehat \pi^*_{y,i'j'}) &= 0.
      \label{FOCpiStar}
\end{align} 
In addition to this first order condition we have the  
constraint $Q^\perp_y    \widehat \pi^*_{y} = 0$,
which implies that $\widehat \pi^*_{y}  = Q_y \xi_y$ for some $\xi_y \in \mathbb{R}^n$,
that is, we only need to consider parameters $\pi^*_{y}$  that can be represented as $ Q_y \xi_y$.
In the following we perform three expansion steps for the log-likelihood function
(or for the corresponding score function), each time restricting $ \widehat \pi^*_{y,ij} - \pi^{* \, 0}_{y,ij} $ further.

\# Step 1: We assume uniform boundedness
of all parameters and variables that enter into the single index.
Therefore,
for all $y \in \mY$ and $(i,j) \in \mD$ we have
$\pi^{0}_{y,ij} \in [\pi_{\min}, \pi_{\max}]$, where $[\pi_{\min}, \pi_{\max}]$ is some bounded interval.
By strict convexity of minus the logistic log-likelihood function  
there exist constants $c_{\min}$ and $c_{\max}$ such that
 $0< c_{\min} \leq \Lambda^{(1)}_{y,ij} \leq c_{\max} < \infty$
 for all $y \in  \mY$ and $(i,j) \in \mD$. 
 Hence,
$\pi^{* \, 0}_{y,ij} \in [c_{\min}^{1/2} \pi_{\min}, c_{\max}^{1/2} \pi_{\max} ]$,
 for all $y \in \mY$ and $(i,j) \in \mD$.
Define $\Pi_{\rm bnd} :=  [c_{\min}^{1/2} \pi_{\min} - \epsilon, c_{\max}^{1/2} \pi_{\max}   + \epsilon]$, 
where $\epsilon>0$ is an arbitrary finite constant. 
In the following we only need to consider values of $\pi^{*}_{y,ij}$ inside  $\Pi_{\rm bnd}$.

Because $\Pi_{\rm bnd}$ is bounded and $ \ell^*_{y,ij}( \pi^*_{y,ij}  ) $ is smooth
we know that all the derivatives of $ \ell^*_{y,ij}( \pi^*_{y,ij}  ) $ are uniformly bounded inside $\Pi_{\rm bnd}$.
In particular, there exists a finite constant $b$ such that, for $k \in \{1,2,3\}$,
\begin{align*}
     \sup_{\pi \in \Pi_{\rm bnd}}  \; \sup_{y \in \mY} \; \max_{(i,j) \in \mD} \left| \partial_{\pi^{*k}}   \ell^*_{y,ij}( \pi   ) \right| \leq b .
\end{align*}
By a third order expansion of 
$\pi^*_{y,ij} \mapsto    \ell^*_{y,ij}( \pi^*_{y,ij})$ around $\pi^{*\, 0}_{y,ij}$
we find
\begin{align}
     & \sum_{(i,j) \in \mD}    \ell^*_{y,ij}( \pi^*_{y,ij}) -  \sum_{(i,j) \in \mD}    \ell^*_{y,ij}( \pi^{*\, 0}_{y,ij})
  \nonumber  \\  
      &=    \sum_{(i,j) \in \mD} s_{y,ij}   \left(  \pi^*_{y,ij} -  \pi^{*\,0}_{y,ij} \right)
           + \frac 1 2  
           \sum_{(i,j) \in \mD}   \left(  \pi^*_{y,ij} -  \pi^{*\,0}_{y,ij} \right)^2
           + \frac 1 6  
           \sum_{(i,j) \in \mD}   \left(\partial_{\pi^{*3}}  \ell^*_{y,ij}( \underline \pi^*_{y,ij}  )  \right)  \left(  \pi^*_{y,ij} -  \pi^{*\,0}_{y,ij} \right)^3
    \nonumber  \\
      &\geq      \sum_{(i,j) \in \mD} s_{y,ij}   \left(  \pi^*_{y,ij} -  \pi^{*\,0}_{y,ij} \right)
           + \frac 1 2  
           \sum_{(i,j) \in \mD}   \left(  \pi^*_{y,ij} -  \pi^{*\,0}_{y,ij} \right)^2
           - \frac b 6  
           \sum_{(i,j) \in \mD}    \left|  \pi^*_{y,ij} -  \pi^{*\,0}_{y,ij} \right|^3 ,
     \label{LowerBoundTech}      
\end{align}
where $\underline \pi^*_{y,ij} $ is an intermediate values between $ \pi^*_{y,ij}$ and $\pi^{*\, 0}_{y,ij}$.
Analogously,
\begin{align}
     & \sum_{(i,j) \in \mD}    \ell^*_{y,ij}( \pi^*_{y,ij}) -  \sum_{(i,j) \in \mD}    \ell^*_{y,ij}( \pi^{*\, 0}_{y,ij})
    \nonumber   \\
      &\leq      \sum_{(i,j) \in \mD} s_{y,ij}   \left(  \pi^*_{y,ij} -  \pi^{*\,0}_{y,ij} \right)
           + \frac 1 2  
           \sum_{(i,j) \in \mD}   \left(  \pi^*_{y,ij} -  \pi^{*\,0}_{y,ij} \right)^2
           + \frac b 6  
           \sum_{(i,j) \in \mD}    \left|  \pi^*_{y,ij} -  \pi^{*\,0}_{y,ij} \right|^3 ,
          \label{UpperBoundTech}         
\end{align}
Evaluating \eqref{LowerBoundTech}
 at  $\pi^*_{y,ij} =   \pi^{*\, 0}_{y,ij} - \left( Q_y s_y  \right)_{ij} + \left( Q_y \zeta_y  \right)_{ij}$, and 
\eqref{UpperBoundTech} 
at $\pi^*_{y,ij}=  \pi^{*\, 0}_{y,ij} - \left( Q_y s_y  \right)_{ij}$ gives
\begin{align}
     & \sum_{(i,j) \in \mD}    \ell^*_{y,ij}\left[  \pi^{*\, 0}_{y,ij} - \left( Q_y s_y  \right)_{ij} + \left( Q_y \zeta_y  \right)_{ij} \right] -  \sum_{(i,j) \in \mD}    \ell^*_{y,ij}\left[  \pi^{*\, 0}_{y,ij} - \left( Q_y s_y  \right)_{ij} \right]
   \nonumber \\
      &\geq      \sum_{(i,j) \in \mD} s_{y,ij}    \left[   \left( Q_y \zeta_y  \right)_{ij} - \left( Q_y s_y  \right)_{ij}   \right]
           + \frac 1 2  
           \sum_{(i,j) \in \mD}   \left[  \left( Q_y \zeta_y  \right)_{ij} - \left( Q_y s_y  \right)_{ij}   \right]^2
           - \frac b 6  
           \sum_{(i,j) \in \mD}    \left| \left( Q_y \zeta_y  \right)_{ij} - \left( Q_y s_y  \right)_{ij}  \right|^3 
      \nonumber \\
     & \quad +        \sum_{(i,j) \in \mD} s_{y,ij}  \left( Q_y s_y  \right)_{ij} 
           - \frac 1 2  
           \sum_{(i,j) \in \mD}  \left( Q_y s_y  \right)_{ij}^2
           - \frac b 6  
           \sum_{(i,j) \in \mD}    \left|  \left( Q_y s_y  \right)_{ij}  \right|^3
  \nonumber \\
     &=   \frac 1 2      \sum_{(i,j) \in \mD}   
     \left[ \left( Q_y \zeta_y  \right)^2_{ij}   
       - \frac b 3     \left| \left( Q_y \zeta_y  \right)_{ij} - \left( Q_y s_y  \right)_{ij}  \right|^3 
             - \frac b 3     \left|  \left( Q_y s_y  \right)_{ij}  \right|^3
          \right]   
  \nonumber \\
    &\geq   \frac 1 2      \sum_{(i,j) \in \mD}   
     \left[ \left( Q_y \zeta_y  \right)^2_{ij}   
       - \frac {4 b} 3     \left| \left( Q_y \zeta_y  \right)_{ij}    \right|^3 
             - \frac {5 b} 3     \left|  \left( Q_y s_y  \right)_{ij}  \right|^3
          \right]   
  \nonumber \\
    &=   \frac 1 2      \sum_{(i,j) \in \mD}   
     \left\{ \left( Q_y \zeta_y  \right)^2_{ij}   
       \left[ 1  - \frac {4 b} 3     \left| \left( Q_y \zeta_y  \right)_{ij}    \right| \right] 
             - \frac {5 b} 3     \left|  \left( Q_y s_y  \right)_{ij}  \right|^3
          \right\}    ,
     \label{InequalityObjectiveStep1}     
\end{align}
where we also used that $Q_y Q_y = Q_y$
and $   \left| \left( Q_y \zeta_y  \right)_{ij} - \left( Q_y s_y  \right)_{ij}  \right|^3 
 \leq   4 \left| \left( Q_y \zeta_y  \right)_{ij} \right|^3 + 4 \left| \left( Q_y s_y  \right)_{ij}  \right|^3 
$.
By the result of Lemma~\ref{lemma:ScoreAve}$(i)$ we know that there exists a sequence $\kappa_n = o(1)$ such that wpa1
$$
   \sup_{y \in  \mY} \max_{(i,j) \in \mD}  \left| \left( Q_y s_y  \right)_{ij}  \right| \leq \kappa_n \; n^{-1/6},
$$
which implies that
\begin{align*}
        \sup_{y \in  \mY}   \sum_{(i,j) \in \mD}      \left|  \left( Q_y s_y  \right)_{ij}  \right|^3
        &\leq n^{1/2} \, \kappa_n^3 .
\end{align*}
Consider the sets
\begin{align*}
    \Pi^*_{y,n} &:= \left\{ \pi^*_y \in \mathbb{R}^n \, : \,  
      Q^\perp_y    \widehat \pi^*_{y} = 0
     \; \;
     \text{and}
     \; \;
      \sum_{(i,j) \in \mD}    \left( \pi^*_{y,ij} -   \pi^{*\, 0}_{y,ij} + \left( Q_y s_y  \right)_{ij} \right)^2
        \leq    n^{1/2} \, \kappa^2_n
             \right\} ,
  \\
      \overline \Pi^*_{y,n} &:= \left\{ \pi^*_y \in \mathbb{R}^n \, : \,  
      Q^\perp_y    \widehat \pi^*_{y} = 0
     \; \;
     \text{and}
     \; \;
      \sum_{(i,j) \in \mD}    \left( \pi^*_{y,ij} -   \pi^{*\, 0}_{y,ij} + \left( Q_y s_y  \right)_{ij} \right)^2
       =  n^{1/2} \, \kappa^2_n
             \right\} .
\end{align*}
Here, $  \overline \Pi^*_{y,n}$ is the boundary of $\Pi^*_{y,n}$ within the set of all $\pi^*_y$ that satisfy the constraint 
$ Q^\perp_y    \widehat \pi^*_{y} = 0$. 
For any $\zeta \in \mathbb{R}^n$ with  $Q^\perp_y   \zeta = 0$ we have
$Q_y  \zeta =\zeta$, and by applying 
Cauchy-Schwarz inequality we thus find
\begin{align}
     \|\zeta\|_\infty 
     &:=
     \max_{(i,j) \in \mD}  \left|  \zeta_{ij}    \right|
     = 
       \max_{(i,j) \in \mD}  \left|  \sum_{(i',j') \in \mD} Q_{y,ij,i'j'}   \zeta_{i'j'}   \right|
  \nonumber   \\
     &\leq    \max_{(i,j) \in \mD}  \left(  \sum_{(i',j') \in \mD} Q_{y,ij,i'j'}^2  \right)^{1/2}
           \left(  \sum_{(i',j') \in \mD} \zeta_{i'j'}^2  \right)^{1/2}
    \nonumber \\      
      &=     \max_{(i,j) \in \mD}  \left(   Q_{y,ij,ij}  \right)^{1/2}
           \| \zeta \|   
       = O_P(n^{-1/4})          \| \zeta \|  ,
    \label{HelperInequ}   
\end{align}
where we also used that $Q_yQ_y=Q_y$ and employed Lemma~\ref{lemma:PropQ}$(iii)$.
By applying \eqref{HelperInequ} to
$ \zeta_{ij} = \pi^*_{y,ij} -   \pi^{*\, 0}_{y,ij} + \left( Q_y s_y  \right)_{ij}$
we find that for $\pi^*_{y} \in   \Pi^*_{y,n}$ we have
\begin{align*}
     \sup_{y \in  \mY}
     \sup_{\pi^*_y \in    \Pi^*_{y,n}}
      \max_{(i,j) \in \mD}  \left|   \pi^*_{y,ij} -   \pi^{*\, 0}_{y,ij} + \left( Q_y s_y  \right)_{ij} \right|
     = O_P(\kappa_n)  ,
\end{align*}
and also using Lemma~\ref{lemma:ScoreAve}$(i)$ we thus have
\begin{align}
     \sup_{y \in  \mY}
     \sup_{\pi^*_y \in    \Pi^*_{y,n}}
      \max_{(i,j) \in \mD}  \left|   \pi^*_{y,ij} -   \pi^{*\, 0}_{y,ij}   \right|
     = O_P(\kappa_n) + o_P(n^{-1/6}) = o_P(1) .
    \label{ResMaxPiStar} 
\end{align}
Hence,  when applying \eqref{InequalityObjectiveStep1} to 
$\pi^*_y \in    \Pi^*_{y,n}$ with 
$\left( Q_y \zeta_y  \right)_{ij} =  \pi^*_{y,ij} -   \pi^{*\, 0}_{y,ij} + \left( Q_y s_y  \right)_{ij}$,
then the term $\frac {4 b} 3     \left| \left( Q_y \zeta_y  \right)_{ij}    \right|$ is of order $o_P(1)$ and the term $ \left|  \left( Q_y s_y  \right)_{ij}  \right|^3$ is of smaller order than $\left( Q_y \zeta_y  \right)^2_{ij}$.
In addition, note that
$
     \pi^{*\, 0}_{y,ij} - \left( Q_y s_y  \right)_{ij} \in   \Pi^*_{y,n} .
$
Thus, by applying \eqref{InequalityObjectiveStep1}
with $\left( Q_y \zeta_y  \right)_{ij} =  \pi^*_{y,ij} -   \pi^{*\, 0}_{y,ij} + \left( Q_y s_y  \right)_{ij}$,
and using \eqref{ResMaxPiStar} we find that with probability approaching one  we have
\begin{align}
       \sum_{(i,j) \in \mD}    \ell^*_{y,ij}\left( \pi^*_{y,ij} \right) -  \sum_{(i,j) \in \mD}    \ell^*_{y,ij}\left[  \pi^{*\, 0}_{y,ij} - \left( Q_y s_y  \right)_{ij} \right]
       > 0 ,
       \qquad
       \text{for all $\pi^*_y \in  \overline \Pi^*_{y,n}$.}
\end{align}
Thus, we have a convex set $ \Pi^*_{y,n}$ such that the convex    function 
$\pi^*_y \mapsto \sum_{(i,j) \in \mD}    \ell^*_{y,ij}\left(  \pi^*_{y,ij} \right) $
takes a smaller value inside the set $ \Pi^*_{y,n}$ than on any point of its boundary $  \overline \Pi^*_{y,n}$
(within the set of all $\pi^*_y$ that satisfy the constraint  $ Q^\perp_y    \widehat \pi^*_{y}=0$). This guarantees that the minimizer 
of the objective function needs to be inside the set  $ \Pi^*_{y,n}$, that is, we have
$\widehat \pi_y \in  \Pi^*_{y,n}$, which implies
\begin{align*}
      \sup_{y \in  \mY}   \sum_{(i,j) \in \mD}    \left( \widehat \pi^*_{y,ij} -   \pi^{*\, 0}_{y,ij} + \left( Q_y s_y  \right)_{ij} \right)^2
       &\leq \kappa_n^2 \, n^{1/2} = o_P(n^{1/2}) ,
\end{align*}
and by  the inequality \eqref{HelperInequ} 
with $ \zeta_{ij} = \widehat \pi^*_{y,ij} -   \pi^{*\, 0}_{y,ij} + \left( Q_y s_y  \right)_{ij}$,
and   Lemma~\ref{lemma:ScoreAve}$(i)$, we find
\begin{align}
     \sup_{y \in  \mY} \max_{(i,j) \in \mD}  \left| \widehat \pi^*_{y,ij} - \pi^{* \, 0}_{y,ij}  \right|
     =   O_P(\kappa_n) + o_P(n^{-1/6})  = o_P(1) .
    \label{ResultStep1}
\end{align}

\# Step 2: An expansion of \eqref{FOCpiStar} in $\widehat \pi^*_{y,i'j'}$ around $ \pi^{*\,0}_{y,i'j'}$ up to second order yields
\begin{align*}
     &  \sum_{(i',j') \in \mD}    Q_{y,ij,i'j'} 
       \Bigg[
       s_{y,i'j'}
       + \left( \widehat \pi^*_{y,i'j'} -  \pi^{*\,0}_{y,i'j'} \right)
       + \frac 1 2  \left(\partial_{\pi^{*3}}  \ell^*_{y,i'j'}(\widetilde\pi^*_{y,i'j'}  )  \right)
           \left( \widehat \pi^*_{y,i'j'} -  \pi^{*\,0}_{y,i'j'} \right)^2
       \Bigg]  = 0 ,
\end{align*}
where $\widetilde\pi^*_{y,i'j'} $ is a value between $ \pi^{*\,0}_{y,i'j'} $ and $\widehat \pi^*_{y,i'j'}$.
By combining this expansion with the constraint $Q^\perp_y    (\widehat \pi^*_{y}-\pi^{*\,0}_y) = 0$,
which implies that $\widehat \pi^*_{y}-\pi^{*\,0}_y = Q_y \xi_y$, for some $\xi_y$,
 we obtain
\begin{align*}
     \widehat \pi^*_{y,ij} - \pi^{* \, 0}_{y,ij} 
     &= - \left( Q_y s_y  \right)_{ij}
        +   r^{(1)}_{y,ij},
\end{align*}
where
\begin{align*}
      r^{(1)}_{y,ij} &=  -  \frac 1 2  
      \sum_{(i',j') \in \mD}    Q_{y,ij,i'j'} 
       \left(\partial_{\pi^{*3}}  \ell^*_{y,i'j'}(\widetilde\pi^*_{y,i'j'}  ) \right)
  \left( \widehat \pi^*_{y,i'j'} -  \pi^{*\,0}_{y,i'j'} \right)^2
\end{align*}
Using our initial convergence rate result \eqref{ResultStep1} in part 1 of this proof,
and  Lemma~\ref{lemma:PropQ}, and also uniform boundedness of all the derivatives of 
$ \ell^*_{y,i'j'}( \pi^* )$ within $\Pi_{\rm bnd}$, we   find
 $$
 \sup_{y \in  \mY} \max_{(i,j) \in \mD} \left| r^{(1)}_{y,ij} \right| = o_P(n^{-1/2+1/6}) + O_P(\kappa_n^2)  .
 $$ 
Hence, by Lemma~\ref{lemma:ScoreAve}$(i)$, 
\begin{align}
     \sup_{y \in  \mY} \max_{(i,j) \in \mD}  \left| \widehat \pi^*_{y,ij} - \pi^{* \, 0}_{y,ij}  \right|
     = o_P(n^{-1/6}) + O_P(\kappa_n^2)  .
   \label{IntermediateResultStep2}  
\end{align}
Using \eqref{IntermediateResultStep2} instead of \eqref{ResultStep1}, and  reapplying the same argument a second time we obtain
\begin{align*}
     \sup_{y \in  \mY} \max_{(i,j) \in \mD}  \left| \widehat \pi^*_{y,ij} - \pi^{* \, 0}_{y,ij}  \right|
     = o_P(n^{-1/6}) + O_P(\kappa_n^4)  .
\end{align*}
And by iterating this argument $q$-times we obtain 
\begin{align*}
     \sup_{y \in  \mY} \max_{(i,j) \in \mD}  \left| \widehat \pi^*_{y,ij} - \pi^{* \, 0}_{y,ij}  \right|
     = o_P(n^{-1/6}) + O_P(\kappa_n^{2q})  .
\end{align*}
for any positive integer $q$. Since $\kappa_n = o(1)$ we can choose $q$ large enough such that
\begin{align}
     \sup_{y \in  \mY} \max_{(i,j) \in \mD}  \left| \widehat \pi^*_{y,ij} - \pi^{* \, 0}_{y,ij}  \right|
     = o_P(n^{-1/6})   .
    \label{ResultStep2}
\end{align}
 
\# Step 3: An expansion of \eqref{FOCpiStar} in $\widehat \pi^*_{y,i'j'}$ around $ \pi^{*\,0}_{y,i'j'}$ up to third order yields 
\begin{align*}
     &  \sum_{(i',j') \in \mD}    Q_{y,ij,i'j'} 
       \Bigg[
       s_{y,i'j'}
       + \left( \widehat \pi^*_{y,i'j'} -  \pi^{*\,0}_{y,i'j'} \right)
       + \frac 1 2  \left(\partial_{\pi^{*3}}  \ell^*_{y,i'j'} \right)
           \left( \widehat \pi^*_{y,i'j'} -  \pi^{*\,0}_{y,i'j'} \right)^2
  \\& \qquad \qquad \qquad      \qquad  \qquad      \qquad  \qquad \quad
        +     \frac 1 6  \left(\partial_{\pi^{*4}}  \ell^*_{y,i'j'}(\overline \pi^*_{y,i'j'}  )  \right)
           \left( \widehat \pi^*_{y,i'j'} -  \pi^{*\,0}_{y,i'j'} \right)^3
       \Bigg]  = 0 ,
\end{align*}
where $\overline \pi^*_{y,i'j'} $ is  a value between $ \pi^{*\,0}_{y,i'j'} $ and $\widehat \pi^*_{y,i'j'}$.
By again using the constraint $\widehat \pi^*_{y}-\pi^{*\,0}_y = Q_y \xi_y$, for some $\xi_y$,
 we  obtain
\begin{align*}
     \widehat \pi^*_{y,ij} - \pi^{* \, 0}_{y,ij} 
     &= - \left( Q_y s_y  \right)_{ij}
        - \frac 1 2 
         \sum_{(i',j') \in \mD} 
        \underbrace{  \left(\partial_{\pi^{*3}}  \ell^*_{y,i'j'} \right)
        }_{
   =  \frac{ \partial_{\pi^3} \ell_{y,i'j'} } { \left(  \Lambda^{(1)}_{y,i'j'} \right)^{3/2} } }
       \,   Q_{y, ij, i'j'} \,
        \left[ \left( Q_y s_y  \right)_{i'j'} \right]^2
        + r_{y,ij},
\end{align*}
where
\begin{align*}
    r_{y,ij} &=  - \sum_{(i',j') \in \mD}    Q_{y,ij,i'j'} 
       \Bigg\{   
       \frac 1 2  \left(\partial_{\pi^{*3}}  \ell^*_{y,i'j'} \right)
       r^{(1)}_{y,i'j'}  \left[ 2 \left( Q_y s_y  \right)_{i'j'} + r^{(1)}_{y,i'j'} \right]
  \\& \qquad \qquad \qquad      \qquad  \qquad      \qquad  \qquad \quad
        +     \frac 1 6  \left(\partial_{\pi^{*4}}  \ell^*_{y,i'j'}(\overline \pi^*_{y,i'j'}  )  \right)
           \left( \widehat \pi^*_{y,i'j'} -  \pi^{*\,0}_{y,i'j'} \right)^3
       \Bigg\} ,
\end{align*}
and therefore
\begin{align*}
   \max_{(i,j) \in \mD}  \left| r_{y,ij} \right| &\leq 
   \left( \max_{(i,j) \in \mD}    \sum_{(i',j') \in \mD}    Q_{y,ij,i'j'}  \right)
     \max_{(i,j) \in \mD}   \Bigg|   
       \frac 1 2  \left(\partial_{\pi^{*3}}  \ell^*_{y,ij} \right)
       r^{(1)}_{y,ij}  \left[ 2 \left( Q_y s_y  \right)_{ij} + r^{(1)}_{y,ij} \right]
  \\& \qquad \qquad \qquad      \qquad  \qquad      \qquad  \qquad \qquad
        +     \frac 1 6  \left(\partial_{\pi^{*4}}  \ell^*_{y,ij}(\overline \pi^*_{y,ij}  )  \right)
           \left( \widehat \pi^*_{y,ij} -  \pi^{*\,0}_{y,ij} \right)^3
       \Bigg| .
\end{align*}
Thus,
using \eqref{ResultStep2},
Lemma~\ref{lemma:PropQ},
and  Lemma~\ref{lemma:ScoreAve}$(i)$, and also uniform boundedness of all the derivatives of 
$ \ell^*_{y,i'j'}( \pi^* )$ within  $\Pi_{\rm bnd}$, we thus find
 $\sup_{y \in  \mY} \max_{(i,j) \in \mD} \left| r_{y,ij} \right| = o_P(n^{-1/2})$. 
 This gives the result of the lemma, since
 $   \widehat \pi^*_{y,ij} - \pi^{* \, 0}_{y,ij}  = \left( \Lambda^{(1)}_{y,ij} \right)^{1/2} \left( \widehat \pi_{y,ij} - \pi^0_{y,ij} \right)$.
\end{proof}

\begin{proof}[\bf Proof of Lemma~\ref{lemma:InvW}]
    We prove this lemma by showing that the smallest eigenvalue of $W_y$ is bounded from below,
     uniformly over $y \in {\mathcal Y}$.
     We know that there exists $b_{\min}>0$ such that $ \Lambda^{(1)}_{y,ij} \geq b_{\min}$, uniformly over $y$, $i$, $j$.
     Then,
     \begin{align*}
          \lambda_{\min}( W_y )
          &= \min_{\| \delta \|=1}  \delta'   W_y  \delta
          \\
          &=  \min_{\| \delta \|=1}   \min_{(\alpha,\gamma) \in \mathbb{R}^{I+J}} 
       \left[ 
       \frac 1 n \sum_{(i,j) \in \mD}  \Lambda^{(1)}_{y,ij} \; \, (x'_{ij} \delta  -  \alpha_i - \gamma_j )^2 \right] 
        \\
          &\geq  \min_{\| \delta \|=1}   \min_{(\alpha,\gamma) \in \mathbb{R}^{I+J}} 
       \left[ 
       \frac 1 n \sum_{(i,j) \in \mD}  b_{\min} \;  (x'_{ij} \delta  -  \alpha_i - \gamma_j )^2 \right] 
        \\
          &=   b_{\min}  \; \min_{\| \delta \|=1}   \min_{(\alpha,\gamma) \in \mathbb{R}^{I+J}} 
       \left[ 
       \frac 1 n  \sum_{(i,j) \in \mD}     (x'_{ij} \delta  -  \alpha_i - \gamma_j )^2 \right] 
       \\
         &\geq    b_{\min}   \; c_3  >0 ,  
     \end{align*}
     where  existence of $c_3>0$ is guaranteed
     by
     Assumption~\ref{ass:baseline}$(vi)$.
\end{proof}

\begin{proof}[\bf Proof of Lemma~\ref{lemma:PropQ}]
    We showed that $ Q_y = Q^{(1)}_y + Q^{({\rm FE})}_y $ in equation \eqref{QisQ1QFE}.
    We now want to find the bound on $Q_y^{({\rm rem})} = Q^{({\rm FE})}_y -   Q_y^{(2)} -  Q_y^{(3)}$
    in part (i) of the lemma.
    Let
    \begin{align*}
         {\mathcal H}_y := \left[w^{(2)},w^{(3)} \right]' \Lambda^{(1)}_{y} \left[w^{(2)},w^{(3)} \right] .
    \end{align*}
    Then,
    \begin{align*}
          Q^{({\rm FE})}_y  := \left( \Lambda^{(1)}_{y} \right)^{1/2} 
    \left[w^{(2)},w^{(3)} \right]   {\mathcal H}_y^\dagger \left[w^{(2)},w^{(3)} \right]'  \left( \Lambda^{(1)}_{y} \right)^{1/2} ,
    \end{align*}
    where we  use the Moore-Penrose pseudo-inverse $\dagger$,  because ${\mathcal H}_y$ has one zero-eigenvalue
    with corresponding eigenvector $v=(1_I,-1_J)$, that is, $v$ is a column vector with $I$ ones follows by $J$ minus ones.\footnote{Note that the additively separable structure $\alpha_i(y) + \gamma_j(y)$ is invariant to adding a constant to all the $\alpha_i(y)$ and subtracting the same constant to all the $\gamma_j(y)$.}
    We can therefore write
    \begin{align*}
          {\mathcal H}_y^\dagger  &=
          \left[ {\mathcal H}_y + vv' / (I+J) \right]^{-1} - vv' / (I+J).
    \end{align*}
    The matrix ${\mathcal H}_y = \partial_{\phi_y \phi_y'}  \sum_{(i,j) \in \mD}   \ell_{y,ij}(\pi^0_{y,ij})$ is simply the $(I+J) \times (I+J)$ Hessian
    matrix    of minus the log-likelihood function with respect to all the fixed effects $\phi_y = (\alpha_y', \gamma_y')'$.
    We decompose ${\mathcal H}_y = {\mathcal D}_y + {\mathcal R}_y$, where
    \begin{align*}
         {\mathcal D}_y &:=  w^{(2) \prime}   \Lambda^{(1)}_{y}  w^{(2)} + w^{(3) \prime} \Lambda^{(1)}_{y}  w^{(3)} 
             = \left( \begin{array}{cc}
                       \diag\left(  \sum_{j \in \mD_i} \Lambda^{(1)}_{y,ij}   \right)_{i=1\ldots,I}
                       & 0_{I \times J}\\
                     0_{J \times I} &  
                        \diag\left(  \sum_{i \in \mD_j} \Lambda^{(1)}_{y,ij}  \right)_{j=1\ldots,J}
                 \end{array} \right)
         \\
         {\mathcal R}_y &:=  w^{(2) \prime}   \Lambda^{(1)}_{y}  w^{(3)} + w^{(3) \prime} \Lambda^{(1)}_{y}  w^{(2)} 
       =  \left( \begin{array}{cc}
                     0_{I \times I}
                     &
                    A_y                     
                     \\
                      A_y'  &   0_{J \times J}
                 \end{array} \right) ,
    \end{align*}
    where $A_y$ is the $I \times J$ matrix with entries $A_{y,ij}= \Lambda^{(1)}_{y,ij}$ if $(i,j) \in \mD$, and zero otherwise.
   Because $\infty > b_{\max} \geq  \Lambda^{(1)}_{y,ij} \geq b_{\min} >0$,
     Lemma D.1 in \cite{FernandezValWeidner2016} shows  that this incidental parameter Hessian satisfies
    \begin{align}
          \sup_{y \in  \mY}  \left\| {\mathcal H}_y^\dagger - {\mathcal D}_y^{-1} \right\|_{\max}
          = O_P( n^{-1} ),
          \label{HDbound}
    \end{align}
    where $\| A \|_{\max}$ refers to the maximum over the absolute values of all the elements of the matrix $A$.
    Note that  Lemma D.1 in \cite{FernandezValWeidner2016}  is for the ``expected Hessian'', but 
  for our logit model we have  ${\mathcal H}_y = \Ep {\mathcal H}_y$,
  conditional on regressors and fixed effects, so the distinction
  between Hessian and expected Hessian is irrelevant here.
 Also, in \cite{FernandezValWeidner2016} the Hessian is not indexed by $y$,
 but the derivation of the bound there is in terms of global constants $b_{\min}$, $b_{\max}$
 and thus holds uniformly over $y$. Finally, \cite{FernandezValWeidner2016} does not allow for missing observations, but since we only
 allow for a finite number of missing observations for every $i$ and $j$ that can only have a negligible effect on the Hessian matrix.
 
 We thus have
 \begin{align*}
     Q^{({\rm FE})}_y  &:= \left( \Lambda^{(1)}_{y} \right)^{1/2} 
    \left[w^{(2)},w^{(3)} \right]   {\mathcal D}_y^{-1} \left[w^{(2)},w^{(3)} \right]'  \left( \Lambda^{(1)}_{y} \right)^{1/2} 
   \\ & \quad 
    +
    \left( \Lambda^{(1)}_{y} \right)^{1/2} 
    \left[w^{(2)},w^{(3)} \right]  \left[  {\mathcal H}_y^\dagger - {\mathcal D}_y^{-1} \right] \left[w^{(2)},w^{(3)} \right]'  \left( \Lambda^{(1)}_{y} \right)^{1/2}
  \\
   &=    Q_y^{(2)} +  Q_y^{(3)} + Q_y^{({\rm rem})} ,
 \end{align*}
 where
 \begin{align*}
 Q_y^{({\rm rem})}
 &=
          \left( \Lambda^{(1)}_{y} \right)^{1/2} 
    \left[w^{(2)},w^{(3)} \right]  \left[  {\mathcal H}_y^\dagger - {\mathcal D}_y^{-1} \right] \left[w^{(2)},w^{(3)} \right]'  \left( \Lambda^{(1)}_{y} \right)^{1/2},
 \end{align*}
 and therefore
 \begin{align*}
       \sup_{y \in  \mY} \left\| Q_y^{({\rm rem})} \right\|_{\max}
       &\leq     \sup_{y \in  \mY} 
       \left\|   \left( \Lambda^{(1)}_{y} \right)^{1/2}  \right\|_{\max}
        \left\| {\mathcal H}_y^\dagger - {\mathcal D}_y^{-1} \right\|_{\max}
         \left\|   \left( \Lambda^{(1)}_{y} \right)^{1/2}  \right\|_{\max}
       \\
       &=   \sup_{y \in   \mY} 
       \left( \max_{(i,j)\in \mD}  \Lambda^{(1)}_{y,ij} \right) 
         \left\| {\mathcal H}_y^\dagger - {\mathcal D}_y^{-1} \right\|_{\max} = O_P(n^{-1}),
 \end{align*}
 which can equivalently
  be written as $   \sup_{y \in  \mY}  \max_{(i,j) \in \mD}  \max_{(i',j') \in \mD}   \left| Q_{y,ij,i'j'}^{({\rm rem})} \right| = O_P(n^{-1})$.
  
  Part $(ii)$ and $(iii)$ follow immediately from part $(i)$
  and the explicit formulas for the elements of $Q_y^{(1)}$, $Q_y^{(2)}$ and $Q_y^{(3)}$
  in \eqref{ElementsQ1} and \eqref{ElementsQ23} above.
\end{proof}

\section*{\bf Intermediate results for the proof of of Lemma~\ref{lemma:ScoreAve}}

The proof of Lemma~\ref{lemma:ScoreAve}   requires several intermediate results from the theory of stochastic processes,
which are presented in the following.
The notation $a_n \lesssim b_n$ means that $a_n \leq C \, b_n$
for some constant $C$ that is independent of the sample size $n$.
It is also convenient to
define ${\bf I} := \{ 1,2,\ldots, I\}$ and ${\bf J} := \{ 1,2,\ldots, J\}$.
In this section we assume that   $\mY$ is a bounded interval, and that $y_{ij}$ is continuously distributed
  with density bounded away from zero.
  The results for the case where $y_{ij}$ is discrete follow directly from \cite{FernandezValWeidner2016}.
  Results for a mixed distribution of $y_{ij}$ follow by combining the results for the continuous and discrete cases.

\subsection*{Bounds on sample averages over the score $\partial_{\pi} \ell_{y,ij}$}

For every $i \in {\bf I}$ we define the empirical process
\begin{align}
    \mathbb{G}_{J,i} f &:= \frac 1 {\sqrt{|\mD_i|}} \sum_{j\in \mD_i} \left[  f(y_{ij}) - \Ep f(y_{ij}) \right] ,
    \qquad
    f  \in {\mathcal F} := \left\{  \widetilde y \; \mapsto \; 1(\widetilde y \leq y) \; : \; y \in {\mathcal Y} \right\} .
    \label{DefGF}
\end{align} 
Here, for ease of notation, we use the subscript $J$ to denote the sample size, corresponding to the balanced panel case where $|\mD_i|=J$.
Following standard notation we write $\|  \mathbb{G}_{J,i} \|_{\mathcal F} :=  \sup_{f \in {\mathcal F} }   \left|  \mathbb{G}_{J,i} f  \right| $.
Every element of ${\mathcal F}$ correspond to exactly one $y \in {\mathcal Y} $, and
in the following we write $ f_y : \widetilde y \; \mapsto \; 1(\widetilde y \leq y)  $ for that element.
Since $\Ep \; 1\{y_{ij} \leq y\} = \Lambda_{y,ij}$,
\begin{align}
     \mathbb{G}_{J,i} f_y  &= - \frac 1 {\sqrt{|\mD_i|}} \sum_{j \in \mD_i} \partial_{\pi} \ell_{y,ij} ,
     \label{RelatingGscore}
\end{align}
where
\begin{align*}
     \partial_{\pi} \ell_{y,ij}
     := \partial_\pi  \ell_{y,ij}(\pi^0_{y,ij} )  =  \Lambda_{y,ij} - 1\{y_{ij} \leq y \} .
\end{align*}
Our goal  is to show that
$\max_{i \in {\bf I}} \|  \mathbb{G}_{J,i} \|_{\mathcal F} 
          = o_P\left(  I^{1/6}  \right)$
by using the following theorem.          
 The theorem uses standard notation 
$\mathbb{G}_n$ for the empirical process, 
denoting the sample size by $n$ (not $J$ or $|\mD_i|$), and without an extra index $i$. The definition
of $J(\delta, {\mathcal F})$ is given on p.239 of
van der Vaart and Wellner (1996).
All that matters to us is that $J(1, {\mathcal F})$ only depends on ${\mathcal F}$
(not on the probability measure or on the empirical process)
and that for the ${\mathcal F}$  defined in \eqref{DefGF} we have $J(1, {\mathcal F}) < \infty$, because ${\mathcal F}$ is
VC class. An obvious envelope function for that ${\mathcal F}$  is $F: \tilde y \mapsto 1$, which satisfies $\| F \|_n=1$.
The  minimal measurable majorant of $\{ \mathbb{G}_{n} f  :  f \in {\mathcal F} \} $
is denoted by $\| \mathbb{G}_{n} \|^*_{\mathcal F}$, and is identical to  $\| \mathbb{G}_{n} \|_{\mathcal F}$ for our purposes.


\begin{lemma}[\bf Restatement of Theorem 2.14.1 in van der Vaart and Wellner, 1996, for INID case]
    \label{lemma:Th2141}
   Let $\mathbb{G}_{n}  $ be the empirical process of an i.n.i.d. sample.\footnote{
    An example is \eqref{DefGF}. In that example the sample size is $n=|\mD_i|$. We require 
   results for non-identically distributed samples, because $y_{ij}$ conditional on regressors and fixed effects
   is independent across $j$ under our assumptions, but not identically distributed.
   }
    Let ${\mathcal F}$ be a $P$-measurable class of measurable functions with measurable envelope $F$.
    Then, for $p\geq 2$,
    \begin{align*}
         \Ep\left[  \left( \|  \mathbb{G}_{n} \|^*_{\mathcal F} \right)^p  \right]
         \lesssim   J(1, {\mathcal F})^p \; \Ep \| F \|^p_n ,
    \end{align*}
    where $ \| F \|_n$ is the $L_2( \mathbb{P}_n )$-seminorm
    and the inequality is valid up to a constant depending only on the $p$ involved in the statement.
\end{lemma}

\begin{proof}
    In van der Vaart and Wellner (1996) the theorem is stated for empirical processes from iid samples ,
    but their proof relies only on symmetrization arguments (their Lemma 2.3.1) and sub-Gaussianity of the symmetrized process,
    which continue to hold for INID samples that we consider here (our $y_{ij}$ are conditionally independent, but not identically
    distributed).
\end{proof}

\begin{corollary}
     \label{cor:ScoreFE}
    Under Assumption~\ref{ass:baseline} we have
    $$\sup_{y \in {\mathcal Y} } \max_{i \in {\bf I}}
      \left|  \frac 1 {\sqrt{|\mD_i|}}  \sum_{j \in \mD_i} \partial_{\pi} \ell_{y,ij} \right| 
          =o_P\left(  n^{1/12}  \right) ,$$
   and
        $$\sup_{y \in {\mathcal Y} } \max_{j \in {\bf J}}
      \left|  \frac 1 {\sqrt{|\mD_j|}}  \sum_{i \in \mD_j} \partial_{\pi} \ell_{y,ij} \right| 
          =o_P\left(  n^{1/12}  \right)   .$$  
\end{corollary}

\begin{proof} 
    The definition \eqref{DefGF} implies \eqref{RelatingGscore}, so we want to show 
    $\max_{i \in {\bf I}} \|  \mathbb{G}_{J,i} \|_{\mathcal F} 
          = o_P\left(  I^{1/6}  \right)$. 
     Applying Lemma~\ref{lemma:Th2141} for the function class $\mathcal{F}$with the envelope function $F: \tilde y \mapsto 1$ we find for $p \geq 1$,
     \begin{align*}
          \Ep \left( \max_{i \in {\bf I}} \|  \mathbb{G}_{J,i} \|_{\mathcal F}  \right)^p
           &=  \Ep  \max_{i \in {\bf I}} \left( \|  \mathbb{G}_{J,i} \|_{\mathcal F}  \right)^p
            \leq \Ep  \sum_{i \in {\bf I}}  \left( \|  \mathbb{G}_{J,i} \|_{\mathcal F}  \right)^p
           =  \sum_{i \in {\bf I}}   \Ep   \left( \|  \mathbb{G}_{J,i} \|_{\mathcal F}  \right)^p
        \\   
          &\leq I \,  J(1, {\mathcal F})^p =  O(I),
     \end{align*}
     where $J(1, {\mathcal F})$ is a finite constant, independent of $i$, as noted earlier above.
     By Markov's inequality we thus find
     $\max_{i \in {\bf I}} \|  \mathbb{G}_{J,i} \|_{\mathcal F}  =  O_P( I^{1/p} )$.
     Choosing $p>6$ gives the desired result.
          
     The second statement $\sup_{y \in {\mathcal Y} } \max_{j \in {\bf J}}
      \left|  \frac 1 {\sqrt{|\mD_j|}}  \sum_{i \in \mD_j}  \partial_{\pi} \ell_{y,ij} \right| 
          =o_P\left(  n^{1/12}  \right)$ can be shown analogously.
\end{proof}

\subsection*{Bounds on weighted sample averages over the score $\partial_{\pi} \ell_{y,ij}$}
We also need results on sample averages of the form e.g.
$ \frac 1 {\sqrt{n}}   \sum_{(i,j) \in \mD}  \,  b^{(n)}_{y,hij} \,   \partial_{\pi} \ell_{y,ij}$,
where $b^{(n)}_{y,hij}$ are weights that also depend on the index $y$. The following lemma is useful for that purpose.


\begin{lemma}
     \label{lemma:HelperEP}
     Suppose $Z_1(t)$, \ldots, $Z_n(t)$ are independent,  stochastic processes indexed by $t \in T$ 
     which are suitably measurable.     Let $B_i$ denote a measurable envelope of $\{ Z_i(t), t \in T \}$.
     such that $ \Ep B_i^p < \infty$ for $p \geq 1$. Let 
     $$
           X_n(t) = \frac 1 n \sum_{i=1}^n Z_i(t), \quad  \Ep X_n(t) = \frac 1 n \sum_{i=1}^n \Ep Z_i(t)
           \quad
           t \in T .
     $$
     Let $B_i$ denote a measurable envelope of $\{ Z_i(t), t \in T \}$.
     Let $T$ be equipped with the pseudo-metric 
     $$
          d_n(t,t') = \sqrt{ \frac 1 n \sum_{i=1}^n   \left( Z_i(t) - Z_i(t') \right)^2} .
     $$
     Let $N(\epsilon,T,d_n)$ denote the covering number of $T$ under $d_n$ balls of radius $\epsilon$.
     Let
     $$
          J_n(\delta,T) = \int_0^\delta
          \sqrt{1+\log N( \epsilon \|B\|_n, T,d_n)  } d \epsilon ,
     $$
     where $$\| B \|_n = \sqrt{\frac 1 n \sum_{i=1}^n |B_i|^2}.$$ Then
     \begin{align*}
           \left\| \left\| X_n - \Ep X_n \right\|_T^* \right\|_{P,p} 
            \lesssim  
            \left\| J_n(1,T) \, \|B\|_n \right\|_{P,p} .
      \end{align*}   
\end{lemma}

\begin{proof}
   The proof is analogous to the proof of Theorem 2.14.1 in
   \cite{vanderVaartandWellner1996}, p.239, with a few notational adjustments.
   
   Let
   $$
        X_n^o(t) = \frac 1 n \sum_{i=1}^n \varepsilon_i Z_i(t), 
        \quad
        t \in T ,
   $$
   denote the symmetrized version of $X_n$,
   where $\varepsilon = (\varepsilon_i)_{i=1}^n$ are independent Rademacher. 
   By Lemma 2.3.6 in  \cite{vanderVaartandWellner1996}  the $L^p(P)$
norm of   $\|X_n - \Ep X_n\|^*_T$ is bounded by  the $L^p(P)$ norm 
of $2 \|X_n^o\|^*_T$.   Let $P_\varepsilon$ denote the distribution of $\varepsilon$. 
   Then by the standard argument, conditional on $(Z_i)_{i=1}^n$,
   $X_n^o$ is sub-Gaussian with respect to $d_n$:
   $$
       \left\| X_n^o(t) - X_n^0(t') \right\|_{\Psi_2(P_\varepsilon)}
       \lesssim  
       d_n(t,t') .
   $$
   Hence by Corollary 2.2.5 in  \cite{vanderVaartandWellner1996}, we conclude
   $$
         \left\| \|X_n^o\|_T \right\|_{\Psi_2(P_\varepsilon)}
         \leq
         \int_0^{{\rm diam}(T,d_n)}
         \sqrt{1+ \log N(\epsilon,T,d_n)}
         d \epsilon .
   $$
   By a change of variables the right side is bounded by
   $$
       \|B\|_n
          \int_0^{{\rm diam}(T,d_n)/ \|B\|_n}
         \sqrt{1+ \log N(\epsilon  \|B\|_n,T,d_n)}
         d \epsilon ,
   $$
   which is further bounded by
   $$
      \|B\|_n \, J_n(1,T) .
   $$
   Every $L_p$-norm is bounded by a multiple of 
   the $\Psi_2$-Orliczs norm. Hence
   $$
        \Ep_\varepsilon
        \left\| X_n^o \right\|_T^p
           \lesssim  
        \left(  J_n(1,T) \,     \|B\|_n \right)^p,
   $$
   where $ \Ep_\varepsilon$ is the expectation conditional on $(Z_i)_{i=1}^n$.   
   Take expectations over $(Z_i)_{i=1}^n$ to obtain the lemma.
\end{proof}

Using Lemma~\ref{lemma:HelperEP} we obtain the following corollary.

\begin{corollary}
    \label{cor:Score2}
    Let Assumption~\ref{ass:baseline} hold. 
    For $y \in {\mathcal Y}$, $h \in \{1,\ldots,I+J\}$, $i \in {\bf I}$ and $j \in {\bf J}$,
    let  $b^{(n)}_{y,hij}, c^{(n)}_{y,i}$, $d^{(n)}_{y,j}$ be real numbers,
     which can depend on the sample size $n$, and on the regressors and fixed effects, but not on the outcome
        variable, and assume that 
        $ \sup_{y \in {\mathcal Y} }  \max_{h \in \{1,\ldots,I+J\}} \allowbreak \max_{i \in {\bf I}} 
          \max_{j \in {\bf I}} \max\left(  \left| b^{(n)}_{y,hij} \right| , \left| \frac{\partial b^{(n)}_{y,hij}} {\partial y} \right| \right) = O_P(1)$,
          and also that
        $ \sup_{y \in {\mathcal Y} }  \max_{i \in {\bf I}}  \max\left(  \left| c^{(n)}_{y,i} \right| , \left| \frac{\partial c^{(n)}_{y,i}} {\partial y} \right| \right) = O_P(1)$,
        and  $ \sup_{y \in {\mathcal Y} }  \max_{j \in {\bf I}}  \max\left(  \left| d^{(n)}_{y,j} \right| , \left| \frac{\partial d^{(n)}_{y,j}} {\partial y} \right| \right) = O_P(1)$.
    Then,
    \begin{align*}
         (i)&&&
           \sup_{y \in {\mathcal Y} }  \max_{h \in \{1,\ldots,I+J\}} 
        \left| \frac 1 {\sqrt{n}}   \sum_{(i,j) \in \mD}  \,  b^{(n)}_{y,hij} \,   \partial_{\pi} \ell_{y,ij} \right| = O_P \left(  n^{1/6}  \right),
        && \qquad \qquad\qquad\qquad \qquad
   \end{align*}
   and
   \begin{align*}
         (ii) &&&
         \sup_{y \in {\mathcal Y} }   \left| \frac 1 I  \sum_{i=1}^I  
                c^{(n)}_{y,i}  
              \left\{
                  \left( \frac 1 {\sqrt{|\mD_i|}} \sum_{j \in \mD_i}   \partial_{\pi} \ell_{y,ij} \right)^2
                  -  \Ep\left[ \left( \frac 1 {\sqrt{|\mD_i|}} \sum_{j \in \mD_i}   \partial_{\pi} \ell_{y,ij} \right)^2 \right]
              \right\}  \right|  = o_P(1),
            \\[10pt]
          &&& \sup_{y \in {\mathcal Y} }   \left| \frac 1 J  \sum_{j=1}^J  
                d^{(n)}_{y,j}  
              \left\{
                  \left( \frac 1 {\sqrt{|\mD_j|}} \sum_{i \in \mD_j}   \partial_{\pi} \ell_{y,ij} \right)^2
                  -  \Ep\left[ \left( \frac 1 {\sqrt{|\mD_j|}} \sum_{i \in \mD_j}   \partial_{\pi} \ell_{y,ij} \right)^2 \right]
              \right\}  \right| = o_P(1) .
    \end{align*}
\end{corollary}

%
\begin{proof}  For part (i) we apply Lemma~\ref{lemma:HelperEP}
    with $T= \mY$ and
    $$ Z_i(y)  =  b^{(n)}_{y,hij} \,   \partial_{\pi} \ell_{y,ij} ,
    $$
    for given $h \in \{1,\ldots,I+J\}$, we can use constant envelope $B_i =  B $
    and the bound $J_n(1,T) \leq C$, which can be established using standard arguments,   we find that
    $A_h =   \sup_{y \in {\mathcal Y} } 
        \left| \frac 1 {\sqrt{n}}   \sum_{(i,j) \in \mD}  \,  b^{(n)}_{y,hij} \,   \partial_{\pi} \ell_{y,ij} \right| $
        satisfies
        $ \max_{h \in \{1,\ldots,I+J\}}  \Ep A_h^4 = O_P(1)$.
    We therefore have
    \begin{align*}
      \Ep     \left(  \max_{h \in \{1,\ldots,I+J\}}  A_h \right)^4  
      \leq   \Ep    \left(  \sum_{h=1}^{I+J}  A_h^4 \right)  
      \leq  (I+J) \max_{h}  \Ep    \left(    A_h^4 \right)
      =O_P( n^{1/2}) ,
    \end{align*}    
    and therefore $  \max_{h \in \{1,\ldots,I+J\}}  A_h = O_P(n^{1/6})$ as desired.
    
    For the first result in part (ii), we apply Lemma~\ref{lemma:HelperEP} with $T= \mY$
    and
    $$
    Z_i(y) :=  c^{(n)}_{y,i}  \left[ \left( \frac 1 {\sqrt{|\mD_i|}} \sum_{j \in \mD_i}   \partial_{\pi} \ell_{y,ij} \right)^2-   \Ep \left( \frac 1 {\sqrt{|\mD_j|}} \sum_{i \in \mD_j}   \partial_{\pi} \ell_{y,ij} \right)^2 \right].
    $$
    Verification of the conditions of the lemma  gives the desired result. 
    The  second result in part (ii) follows analogously.
\end{proof}

\subsection*{FCLT for weighted sample averages over the score $\partial_{\pi} \ell_{y,ij}$}
The following theorem will be used in the proof of part $(ii)$ and $(iii)$ of Lemma~\ref{lemma:ScoreAve}.

\begin{lemma}[\bf Theorem 2.11.11 in van der Vaart and Wellner, 1996]
    \label{lemma:Th21111}
    For each $n$, let $Z_{n1}, \ldots, Z_{n,m_n}$ be independent stochastic processes indexed by an arbitrary 
    index set ${\mathcal F}$. Suppose that there exists a Gaussian-dominated semimetric $\rho$ on ${\mathcal F}$ 
    such that
    \begin{align*}
        \text{(i)} \qquad \quad &\sum_{\ell=1}^{m_n} 
        \Ep^* \left[ \| Z_{n\ell} \|_{\mathcal F}  \; \; 1 \left\{ \| Z_{n\ell} \|_{\mathcal F} > \eta \right\} \right]
        \rightarrow 0 ,
        \qquad
        \text{for every $\eta>0$,}
        \\
         \text{(ii)} \qquad \quad  & \sum_{\ell=1}^{m_n} \Ep\left( Z_{n\ell}(f) - Z_{n\ell}(g) \right)^2 \leq \rho^2(f,g),
        \qquad
        \text{for every $f,g \in {\mathcal F}$,}
        \\
         \text{(iii)} \qquad \quad  &\sup_{t>0} \, \sum_{\ell=1}^{m_n} t^2 \, \mathbb{P}^*\left( \sup_{f,g \in {\mathcal B}(\varepsilon)} 
             \left| Z_{n\ell}(f) - Z_{n\ell}(g) \right| > t \right) \leq \varepsilon^2 ,
    \end{align*}
    for every $\rho$-ball ${\mathcal B}(\varepsilon) \subset {\mathcal F}$ of radius less than $\varepsilon$ and for every $n$.
    Then the sequence $ \sum_{\ell=1}^{m_n}  \left( Z_{\ell,n} - \Ep\, Z_{\ell,n} \right)$ is asymptotically tight
    in $\ell^\infty({\mathcal F})$. It converges in distribution provided it converges marginally.
\end{lemma}

A semi-metric $\rho$ is Gaussian-dominated if it is bounded above by a Gaussian semi-metric. Any semi-metric such that  $\int_0^\infty \sqrt{\log N(\epsilon, \mathcal{F}, \rho)} d \epsilon < \infty$ is Gaussian dominated.  

\section*{\bf Proof of Lemma~\ref{lemma:ScoreAve}}

\begin{proof}[\bf Proof of Lemma~\ref{lemma:ScoreAve}, Part $(i)$]
     We have
     \begin{align}
            \left( Q^{(2)}_y s_y  \right)_{ij} 
            &=  |\mD_i|^{-1/2}  \frac{ \left( \Lambda^{(1)}_{y,ij} \right)^{1/2}} {|\mD_i|^{-1} \sum_{j' \in \mD_i} \Lambda^{(1)}_{y,ij'}} 
           \left[ |\mD_i|^{-1/2}  \sum_{j' \in \mD_i} \partial_{\pi} \ell_{y,ij'} \right]
            ,
     \nonumber       \\
             \left( Q^{(3)}_y s_y  \right)_{ij} 
            &=   |\mD_j|^{-1/2}  \frac{ \left( \Lambda^{(1)}_{y,ij} \right)^{1/2}} {  |\mD_j|^{-1} \sum_{i' \in \mD_j} \Lambda^{(1)}_{y,i'j}} 
           \left[   |\mD_j|^{-1/2}  \sum_{i' \in \mD_j} \partial_{\pi} \ell_{y,i'j} \right] .
           \label{ExpressQ23s}
     \end{align}
    With $\max_i |\mD_i|^{-1/2} = O_P(n^{-1/4})$,  $\max_j |\mD_j|^{-1/2} = O_P(n^{-1/4})$, 
     \begin{align*}
           \sup_{y \in  \mY} \max_{(i,j) \in \mD}  \left|  \frac{ \left( \Lambda^{(1)}_{y,ij} \right)^{1/2}} {|\mD_i|^{-1} \sum_{j' \in \mD_i} \Lambda^{(1)}_{y,ij'}}  \right| &= O_P(1) ,
           &
            \sup_{y \in  \mY} \max_{(i,j) \in \mD}  \left|  \frac{ \left( \Lambda^{(1)}_{y,ij} \right)^{1/2}} {  |\mD_j|^{-1} \sum_{i' \in \mD_j} \Lambda^{(1)}_{y,i'j}}   \right| &= O_P(1) ,
     \end{align*}
     we obtain by Corollary~\ref{cor:ScoreFE} that
     \begin{align}
           \sup_{y \in  \mY} \max_{(i,j) \in \mD}  \left| \left( Q^{(2)}_y s_y    \right)_{ij}  \right| &= o_P(n^{-1/6})
           &
             \sup_{y \in  \mY} \max_{(i,j) \in \mD}  \left| \left( Q^{(3)}_y s_y    \right)_{ij}  \right| &= o_P(n^{-1/6}) .
          \label{QsOrder1}   
     \end{align}
     Next, 
     \begin{align*}
            \left( Q^{(1)}_y s_y  \right)_{ij} 
            &= \left( \Lambda^{(1)}_{y,ij}   \right)^{1/2}
               \widetilde x'_{y,ij}  \, W_y^{-1}   
              \left[ \frac 1 n  \sum_{(i',j')\in \mD} 
                    \widetilde x_{y,i'j'} \, (\partial_{\pi} \ell_{y,i'j'}) \right] ,
          \\
             \left( Q^{({\text{rem}})}_y s_y  \right)_{ij}      
             &=   \left\{   \left( \Lambda^{(1)}_{y} \right)^{1/2} 
    \left[w^{(2)},w^{(3)} \right]   \left(  {\mathcal H}_y^\dagger - {\mathcal D}_y^{-1} \right)  \left[w^{(2)},w^{(3)} \right]'  \left( \Lambda^{(1)}_{y} \right)^{1/2} 
      s_y \right\}_{ij}
      \\
        &= \left( \Lambda^{(1)}_{y,ij} \right)^{1/2} 
            \sum_{(i',j') \in \mD}  
            \left[ {\mathcal G}^{(I \times I)}_{y,ii'}
               +  {\mathcal G}^{(I \times J)}_{y,ij'}
               +  {\mathcal G}^{(J \times I)}_{y,ji'}
               +  {\mathcal G}^{(J \times J)}_{y,jj'}
            \right]
            \partial_{\pi} \ell_{y,i'j'} 
     \end{align*}
     where ${\mathcal H}_y$ and ${\mathcal D}_y$ are $(I+J) \times (I+J)$ matrices introduced in the proof of Lemma~\ref{lemma:PropQ},
     and ${\mathcal G}_y :=   {\mathcal H}_y^\dagger - {\mathcal D}_y^{-1}  $,
     and ${\mathcal G}_y^{(I \times I)}$, $ {\mathcal G}_y^{(I \times J)}$, $ {\mathcal G}_y^{(J \times I)}$, $ {\mathcal G}_y^{(J \times J)}$
     denotes the various blocks of this $(I+J) \times (I+J)$ matrix.
     Remember that according to
     \eqref{HDbound} 
     all the elements of     ${\mathcal G}_y$ are uniformly bounded of order $n^{-1}$.
     Thus, by applying Corollary~\ref{cor:Score2}$(i)$ with 
     $b^{(n)}_{y,hij}$ equal to  $\widetilde x^h_{y,ij}$, for $h=1,\ldots,d_x$,
     and also with $b^{(n)}_{y,hij}$ equal to $n \left( {\mathcal G}^{(I \times I)}_{y,hi}  +  {\mathcal G}^{(I \times J)}_{y,hj} \right)$,
     for $h=1,\ldots,I$,
     and equal to $n \left( {\mathcal G}^{(J \times I)}_{y,hi}  +  {\mathcal G}^{(J \times J)}_{y,hj} \right)$,
     for $h=1,\ldots,J$, we find that
     \begin{align}
           \sup_{y \in  \mY} \max_{(i,j) \in \mD}  \left| \left( Q^{(1)}_y s_y    \right)_{ij}  \right| &= O_P(n^{-1/2+1/6})
           &
             \sup_{y \in  \mY} \max_{(i,j) \in \mD}  \left| \left( Q^{({\text{rem}})}_y s_y    \right)_{ij}  \right| &= O_P(n^{-1/2+1/6}) .
          \label{QsOrder2}   
     \end{align}
     Combining the above we find that
     $Q_y s_y = Q^{(1)}_y s_y + Q^{(2)}_y s_y + Q^{(3)}_y s_y + Q^{({\text{rem}})}_y s_y$
    indeed satisfies
    $\sup_{y \in  \mY} \max_{(i,j) \in \mD}  \left| \left( Q_y s_y  \right)_{ij}  \right| = o_P(n^{-1/6})$.
\end{proof}

\begin{proof}[\bf Proof of Lemma~\ref{lemma:ScoreAve}, Part $(ii)$ and $(iii)$]
Here, we use Theorem 2.11.11 in \cite{vanderVaartandWellner1996}, which is restated above
as Lemma~\ref{lemma:Th21111}.
To relate this to our model we define
\begin{align}
     Z_{\ell,n}(y,k) = \frac {\left[ W_y^{-1}  \;  \widetilde x_{y,i_\ell j_\ell} \right]_k} {\sqrt{n}}  \; 1\{y_{i_\ell j_\ell} \leq y \}  ,
     \label{ProcessEx1}
\end{align}
where $\ell \in \{1,\ldots,n\}$, and $i_\ell \in {\bf I}$, $j_\ell \in {\bf J}$ are chosen such that $B = \{ (i_\ell, j_\ell) \; : \; \ell  = 1,\ldots,n\}$.
$Z_{\ell,n}$ defines a stochastic process with index set ${\mathcal F} = {\mathcal Y} \times \{1,\ldots,\dim \beta\}$.
For $f=(y,k) \in {\mathcal F}$ we write $Z_{\ell,n}(f)$.
Part $(ii)$ of Lemma~\ref{lemma:ScoreAve} can then be written as
\begin{align*}
    \sum_{\ell=1}^n  \left( Z_{\ell,n} - \Ep\, Z_{\ell,n} \right)   \rightsquigarrow  {\mathcal Z}^{(\beta)} ,
\end{align*}
where the limiting process ${\mathcal Z}^{(\beta)}$ is also indexed by $f \in {\mathcal F}$. We also define
the following metric on ${\mathcal F}$,
\begin{align}
    \label{DefMetric}
    \rho(f_1,f_2) := C \left[  \left| y_1 - y_2 \right|^{1/2} +  1( k_1 \neq k_2 ) \right] ,
\end{align}
for some sufficiently large constant $C>0$.
For a general index set ${\mathcal F}$,
a sufficient condition for a metric $\rho$ on ${\mathcal F}$ to be ``Gaussian dominated'' 
is given by (see \citealt{vanderVaartandWellner1996}, p.212)
\begin{align}
    \int_0^\infty \sqrt{ \log N(\varepsilon, {\mathcal F}, \rho)  } \, d \varepsilon \; < \; \infty ,
    \label{ConditionGD}
\end{align}
where $N(\varepsilon, {\mathcal F}, \rho)$ denotes the covering number. 

$Z_{\ell,n}$
is a triangular array, because $W_y$ and $ \widetilde x_{y,ij}$
both implicitly depend on $n$, implying that $Z_{\ell,n_1} \neq Z_{\ell,n_2}$ for $n_1 \neq n_2$. 
Remember that the probability measure we use throughout is conditional on $x$, $\alpha^0$, $\gamma^0$,
implying that the $Z_{\ell,n}$ are independent (but not identically distributed) across $\ell$, according to our assumptions.

    Using the model and the definition \eqref{ProcessEx1} we have
    \begin{align*}
    W_y^{-1}
       \left[ - \frac 1 {\sqrt{n}}  \sum_{(i,j) \in \mD}   \partial_{\pi} \ell_{y,ij} \;  \widetilde x_{y,ij} \right]
       &=  \sum_{\ell=1}^n  \left( Z_{\ell,n} - \Ep\, Z_{\ell,n} \right) =: Z_n .
    \end{align*}
    Using the Lyapunov CLT it is easy to verify that all the finite dimensional
     marginals $(Z_n(f_1), Z_n(f_2),  \allowbreak \ldots, Z_n(f_p))$
    of the stochastic process $Z_n$ converge weakly to a zero mean Gaussian limit process 
    $({\mathcal Z}^{(\beta)}(f_1), {\mathcal Z}^{(\beta)}(f_2),  \allowbreak \ldots, {\mathcal Z}^{(\beta)}(f_p))$. It is also easy to show that the second moments
    of the limit process are given by $\Ep {\mathcal Z}^{(\beta)}(f_1) {\mathcal Z}^{(\beta)}(f_2) = 
    \left[ \overline W_{y_1}^{-1} \;  \overline V_{y_1,y_2} \;  \overline W_{y_2}^{-1} \right]_{k_1 k_2}$.
    
    In order to conclude that the process $Z_n$ is weakly convergent we also need to show that $Z_n$ is tight.
    For this we employ  Lemma~\ref{lemma:Th21111} above with  $m_n = n$
    and metric $\rho$ given in \eqref{DefMetric}.
    This $\rho$ is Gaussian dominated on ${\mathcal F}$, because  we have
    \begin{align*}
        \log N(\varepsilon, {\mathcal F}, \rho) \lesssim
        \left\{
        \begin{array}{ll}
              \log( K/ \varepsilon ),
              &
              \text{for} \; 0<\varepsilon<K ,
              \\
              0,
              &
              \text{for} \;  \varepsilon \geq K ,
        \end{array}
        \right.
    \end{align*}
    for some constant $K>0$, implying that \eqref{ConditionGD} is satisfied.
    
    To verify condition (i) of Lemma~\ref{lemma:Th21111}, we calculate 
    \begin{align*}
        \sum_{\ell=1}^{n}    \Ep^* \left[ \| Z_{n\ell} \|_{\mathcal F} \; \; 1 \left\{ \| Z_{n\ell} \|_{\mathcal F} > \eta \right\} \right]
        & \leq   n \, \max_{\ell}   \Ep^* \left[ \| Z_{n\ell} \|_{\mathcal F}  \; \; 1 \left\{ \| Z_{n\ell} \|_{\mathcal F} > \eta \right\} \right]
        \\
        & \leq   n \, \max_{\ell}   \Ep^* \left[ \frac{  \left\|  Z_{n\ell}   \right\|^2_{\mathcal F} } {\eta}    \; \; 1  \left\{ \| Z_{n\ell} \|_{\mathcal F} > \eta \right\} \right]
        \\
        &\leq    \max_{i,j}   \Ep \left[ \frac{ \sup_y \left\|  W_y^{-1}  \;  \widetilde x_{y,ij}   \right\|_\infty^2  } {\eta}   
         \; \; 1  \left\{ \frac{ \sup_y \| W_y^{-1}  \;  \widetilde x_{y,ij} \|_\infty} {\sqrt{n}}  > \eta  \right\} \right] ,
     \\
       & \rightarrow 0  .
    \end{align*}
    where for the second inequality we multiplied with $ \| Z_{n\ell} \|_{\mathcal F} / \eta$ inside the expectation, which is
    larger than one for $ \| Z_{n\ell} \|_{\mathcal F} > \eta$; for the third inequality we used that     
    $ \| Z_{n\ell} \|_{\mathcal F} \leq \sup_y \| W_y^{-1}  \;  \widetilde x_{y,i_\ell j_\ell} \|_\infty / \sqrt{n}$;
    and for the final conclusion we used that $ \sup_{i,j} E \sup_{y} \left\|  W_y^{-1}  \;  \widetilde x_{y,ij}   \right\|^{2+\delta}$ is uniformly bounded.
    
    Next, for $y_1 \leq y_2$ we have
    \begin{align}
          \sqrt{n} \left| Z_{n\ell}(f_1) - Z_{n\ell}(f_2) \right| 
           &= \left|      [W_{y_1}^{-1}  \;  \widetilde x_{y_1,i(\ell)j(\ell)}]_{k_1}   \; 1\{y_{i(\ell)j(\ell)} \leq y_1 \} 
                                                             -    [W_{y_2}^{-1}  \;  \widetilde x_{y_2,i(\ell)j(\ell)}]_{k_2}   \; 1\{y_{i(\ell)j(\ell)} \leq y_2 \} \right|
          \nonumber    \\
           &\leq \left|   [W_{y_2}^{-1}  \;  \widetilde x_{y_2,i(\ell)j(\ell)}]_{k_2} \right|   1\{y_1 < y_{i(\ell)j(\ell)} \leq y_2 \}    
              \nonumber \\ & \qquad \qquad \qquad \qquad          
                             +  \left|       [W_{y_1}^{-1}  \;  \widetilde x_{y_1,i(\ell)j(\ell)}  ]_{k_1}
                                                             -    [W_{y_2}^{-1}  \;  \widetilde x_{y_2,i(\ell)j(\ell)} ]_{k_2} \right|   
           \nonumber    \\
           & \lesssim    1\{y_1 < y_{i(\ell)j(\ell)} \leq y_2 \}        
              \nonumber \\ & \qquad \qquad        
                     +   \left\|  W_{y_1}^{-1}  \;  \widetilde x_{y_1,i(\ell)j(\ell)} -    W_{y_2}^{-1}  \;  \widetilde x_{y_2,i(\ell)j(\ell)}         \right\|_{\infty}     
                     +      1( k_1 \neq k_2 ) 
         \nonumber   \\                   
           & \lesssim   \left| 1\{  y_{i(\ell)j(\ell)} \leq y_1 \}  -  1\{  y_{i(\ell)j(\ell)} \leq y_2 \}      \right|  
                     +   \left|  y_1 - y_2    \right| 
                     +      1( k_1 \neq k_2 ) .
            \label{BndZdiff}         
    \end{align}
    where we used 
    uniform boundedness of $W_y^{-1}  \;  \widetilde x_{y,ij} $
    and of its derivative wrt $y$.
    The final result in \eqref{BndZdiff} is written such that the bound is also applicable for $y_1>y_2$.
    
    Using the bound \eqref{BndZdiff} we now verify
  condition (ii) of Lemma~\ref{lemma:Th21111},
  \begin{align*}
        \sum_{\ell=1}^{n} \Ep\left( Z_{n\ell}(f_1) - Z_{n\ell}(f_2) \right)^2 
        &\leq     n \, \max_{\ell}  \Ep\left( Z_{n\ell}(f_1) - Z_{n\ell}(f_2) \right)^2 
     \\
        &=    \max_{\ell}  \Ep\left[ \sqrt{n} \left( Z_{n\ell}(f_1) - Z_{n\ell}(f_2) \right) \right]^2 
     \\
            & \lesssim \max_{i,j}  \Ep\left\{    \left| 1\{  y_{ij} \leq y_1 \}  -  1\{  y_{ij} \leq y_2 \}      \right|      
                     +   \left|  y_1 - y_2    \right| 
                     +      1( k_1 \neq k_2 ) \right\}^2
  \\
           &\lesssim       \max_{i,j}  \Ep   \left| 1\{  y_{ij} \leq y_1 \}  -  1\{  y_{ij} \leq y_2 \}      \right|^2  
               +    \left|  y_1 - y_2    \right|^2 +      \left[ 1( k_1 \neq k_2 )   \right]^2        
     \\
          & \lesssim    \max_{i,j}  
                  \left|  \Lambda(\pi^0_{y_2,ij} ) - \Lambda(\pi^0_{y_1,ij} )  \right|       
                +   \left|  y_1 - y_2    \right|^2 +   1( k_1 \neq k_2 )                                                                                                                                
     \\   
        &  \lesssim      \left|  y_1 - y_2    \right| +   \left|  y_1 - y_2    \right|^2 +   1( k_1 \neq k_2 )         
     \\
         & \lesssim   \left[  \left| y_1 - y_2 \right|^{1/2} +  1( k_1 \neq k_2 ) \right]^2.
   \end{align*}
   where, we used that $ \Ep\left|  1\{  y_{ij} \leq y_1 \}  -  1\{  y_{ij} \leq y_2 \}  \right| = 
   \left| \Lambda(\pi^0_{y_2,ij} ) - \Lambda(\pi^0_{y_1,ij} ) \right| \lesssim  \left|  y_1 - y_2    \right|$;
   and we also used that ${\mathcal Y}$ is bounded, implying that $ \left|  y_1 - y_2    \right|^2 \lesssim    \left|  y_1 - y_2    \right|$.
    Thus,  condition (ii) of Lemma~\ref{lemma:Th21111} holds  for sufficiently large $C$ in the definition of $\rho$ in \eqref{DefMetric}.

   To verify condition (iii) of Lemma~\ref{lemma:Th21111}, let $C_1>0$ be the omitted constant that makes
   the result in \eqref{BndZdiff} a regular inequality. We then have
   \begin{align*}
      &  \mathbb{P}^*\left( \sup_{f_1,f_2 \in {\mathcal B}(\varepsilon)} 
             \left| Z_{n\ell}(f_1) - Z_{n\ell}(f_2) \right| > \frac{t} {\sqrt{n}} \right) 
      \\       
        &\leq      
            \mathbb{P}^*\left( \sup_{f_1,f_2 \in {\mathcal B}(\varepsilon)}  
             C_1 \left[   \left| 1\{  y_{i(\ell)j(\ell)} \leq y_1 \}  -  1\{  y_{i(\ell)j(\ell)} \leq y_2 \}      \right|     
                     +   \left|  y_1 - y_2    \right| 
                     +      1( k_1 \neq k_2 )  \right]  > t \right) 
   \\
        &\leq     \mathbb{P}^*\Bigg(
     \sup_{f_1,f_2 \in {\mathcal B}(\varepsilon)} 
           \left| 1\{  y_{i(\ell)j(\ell)} \leq y_1 \}  -  1\{  y_{i(\ell)j(\ell)} \leq y_2 \}      \right|     
   \\ & \qquad  \qquad  \qquad                       
                     +   \sup_{f_1,f_2 \in {\mathcal B}(\varepsilon)}   \left|  y_1 - y_2    \right| 
                     +   \sup_{f_1,f_2 \in {\mathcal B}(\varepsilon)}   1( k_1 \neq k_2 )    > \frac{t} {C_1} \Bigg)   
    \\
         &\leq    \mathbb{P}^*\left(
     \sup_{f_1,f_2 \in {\mathcal B}(\varepsilon)} 
          \left| 1\{  y_{i(\ell)j(\ell)} \leq y_1 \}  -  1\{  y_{i(\ell)j(\ell)} \leq y_2 \}      \right|          
                      > \frac{t} {3 \, C_1} \right)   
   \\ & \qquad                      
            +      
                  \mathbb{P}^*\left(
     \sup_{f_1,f_2 \in {\mathcal B}(\varepsilon)}  \left|  y_1 - y_2    \right| 
                    > \frac{t} {3 \, C_1} \right)           
            + 
                 \mathbb{P}^*\left(
       \sup_{f_1,f_2 \in {\mathcal B}(\varepsilon)}   1( k_1 \neq k_2 )    > \frac{t} {3 \, C_1} \right)      .                               
   \end{align*}
   Any given $\rho$-ball  ${\mathcal B}(\varepsilon)$ of radius less then $\varepsilon$
   also corresponds to a given ball in ${\mathcal Y}$ of radius less than $(\varepsilon/C)^2$.
   The event $\left[ \sup_{f_1,f_2 \in {\mathcal B}(\varepsilon)} 
          \left| 1\{  y_{i(\ell)j(\ell)} \leq y_1 \}  -  1\{  y_{i(\ell)j(\ell)} \leq y_2 \}      \right|          
                      > \frac{t} {3 \, C_1} \right]$
     can only occurs if $ \frac{t} {3 \, C_1}\leq 1$
     and if  $ y_{i(\ell)j(\ell)}$ is realized in that particular ball in ${\mathcal Y}$ of radius less than $(\varepsilon/C)^2$.
     Since our assumptions guarantee that the pdf of $ y_{i(\ell)j(\ell)}$ is uniformly bounded from below
     by a constant $C_2>0$ we thus find that
     \begin{align*}
            \mathbb{P}^*\left(
     \sup_{f_1,f_2 \in {\mathcal B}(\varepsilon)} 
          \left| 1\{  y_{i(\ell)j(\ell)} \leq y_1 \}  -  1\{  y_{i(\ell)j(\ell)} \leq y_2 \}      \right|          
                      > \frac{t} {3 \, C_1} \right)
              &\leq    2 \, C_2     \left( \frac{ \varepsilon} C \right)^2   
                     \;  1\left(   \frac{t} {3 \, C_1} \leq 1 \right) .
     \end{align*}
     Similarly we find
     \begin{align*}
            \mathbb{P}^*\left(
     \sup_{f_1,f_2 \in {\mathcal B}(\varepsilon)}  \left|  y_1 - y_2    \right| 
                    > \frac{t} {3 \, C_1} \right)  
                    &\leq  1\left(   \frac{t} {3 \, C_1} \leq 2   \left( \frac{ \varepsilon} C \right)^2 \right) ,
           \\      
           \mathbb{P}^*\left(
       \sup_{f_1,f_2 \in {\mathcal B}(\varepsilon)}   1( k_1 \neq k_2 )    > \frac{t} {3 \, C_1} \right)     
       &\leq   1\left(   \frac{t} {3 \, C_1} \leq 1 \; \;  \& \; \;   C \leq \varepsilon \right).
     \end{align*}
    We thus calculate   
   \begin{align*}
      &  \sup_{t>0} \, \sum_{\ell=1}^{n} t^2 \, \mathbb{P}^*\left( \sup_{f_1,f_2 \in {\mathcal B}(\varepsilon)} 
             \left| Z_{n\ell}(f_1) - Z_{n\ell}(f_2) \right| > t \right) 
     \\  
        &\leq    \sup_{t>0} \, \max_{\ell} \, n \, t^2 \, \mathbb{P}^*\left( \sup_{f_1,f_2 \in {\mathcal B}(\varepsilon)} 
             \left| Z_{n\ell}(f_1) - Z_{n\ell}(f_2) \right| > t \right) 
     \\
        &=    \sup_{t>0} \, \max_{\ell} \, t^2 \,   \mathbb{P}^*\left( \sup_{f_1,f_2 \in {\mathcal B}(\varepsilon)} 
             \left| Z_{n\ell}(f_1) - Z_{n\ell}(f_2) \right| > \frac{t} {\sqrt{n}} \right) 
    \\
        &\leq   \sup_{t>0} \, t^2 \,
         \left\{      2 \, C_2     \left( \frac{ \varepsilon} C \right)^2  
                     \;  1\left(   \frac{t} {3 \, C_1} \leq 1 \right)      
                  +     1\left(   \frac{t} {3 \, C_1} \leq 2   \left( \frac{ \varepsilon} C \right)^2 \right) 
                  + 1\left(   \frac{t} {3 \, C_1} \leq 1 \; \;  \& \; \;   C \leq \varepsilon \right) 
               \right\}
       \\
          &\leq \varepsilon^2           ,
   \end{align*}
   for sufficiently large choice of $C$. In the last step we also use that ${\mathcal Y}$ is bounded, which together with 
   ${\mathcal B}(\varepsilon) \subset {\mathcal F}$ implies that the possible values of $\varepsilon$ are bounded,
   so that we can always choose $C$ sufficiently large to guarantee that 
   $ 1\left(   \frac{t} {3 \, C_1} \leq 1 \; \;  \& \; \;   C \leq \varepsilon \right)  = 0$.
   
   Thus, we can apply Lemma~\ref{lemma:Th21111} to find that   
  $ \sum_{\ell=1}^n  \left( Z_{\ell,n} - \Ep\, Z_{\ell,n} \right)   \rightsquigarrow  {\mathcal Z}^{(\beta)}$,
  where ${\mathcal Z}^{(\beta)}$ is a tight zero mean Gaussian process with second moments given above.
  
  The proof 
  of part $(iii)$ of Lemma~\ref{lemma:ScoreAve}
         is analogous.         
\end{proof}

\begin{proof}[\bf Proof of Lemma~\ref{lemma:ScoreAve}, Part $(iv)$ and $(v)$]
     Decomposing  $Q_y s_y = Q^{(1)}_y s_y + Q^{(2)}_y s_y + Q^{(3)}_y s_y + Q^{({\text{rem}})}_y s_y$
     and using \eqref{QsOrder1}
     and \eqref{QsOrder2} we find that
     \begin{align*}
            \sup_{y \in \mY} \max_{(i,j) \in \mD} 
            \left\{  \left[ \left( Q_y s_y  \right)_{ij} \right]^2
           - \left[  \left( Q^{(2)}_y s_y  \right)_{ij}  + \left( Q^{(3)}_y s_y  \right)_{ij}  \right]^2
           \right\} = o_P(n^{-1}) ,
     \end{align*}
     and therefore 
     \begin{align*}
          - \frac 1 2 
        W_y^{-1} \frac 1 {\sqrt{n}}      \sum_{(i,j) \in \mD}   \; 
           \widetilde x_{y,ij}          
           \left( \Lambda^{(1)}_{y,ij} \right)^{-1}
      \Lambda^{(2)}_{y,ij} 
        \left[ \left( Q_y s_y  \right)_{ij} \right]^2
        &=  \frac{I} {\sqrt{n}}    C^{(1,\beta)}_y
           + \frac{J} {\sqrt{n}}    C^{(2,\beta)}_y 
           +  C^{(3,\beta)}_y + o_P(1) ,
     \end{align*}
     where
     \begin{align*}
            C^{(1,\beta)}_y
            &:=  - \frac 1 2 
        W_y^{-1}  \;  \frac 1 I       \sum_{(i,j) \in \mD}   \; 
           \widetilde x_{y,ij}          
           \left( \Lambda^{(1)}_{y,ij} \right)^{-1}
      \Lambda^{(2)}_{y,ij} 
        \left[ \left( Q^{(2)}_y s_y  \right)_{ij} \right]^2 ,
        \\
            C^{(2,\beta)}_y
            &:=  - \frac 1 2 
        W_y^{-1}  \;  \frac 1 I       \sum_{(i,j) \in \mD}   \; 
           \widetilde x_{y,ij}          
           \left( \Lambda^{(1)}_{y,ij} \right)^{-1}
      \Lambda^{(2)}_{y,ij} 
        \left[ \left( Q^{(3)}_y s_y  \right)_{ij} \right]^2 ,
      \\  
       C^{(3,\beta)}_y
            &:=  -  
        W_y^{-1}  \frac 1 {\sqrt{n}}      \sum_{(i,j) \in \mD}   \; 
           \widetilde x_{y,ij}          
           \left( \Lambda^{(1)}_{y,ij} \right)^{-1}
      \Lambda^{(2)}_{y,ij} 
         \left( Q^{(2)}_y s_y  \right)_{ij}    \left( Q^{(3)}_y s_y  \right)_{ij}   .
     \end{align*}
     Using that $\Ep    \left[ \left( Q^{(2)}_y s_y  \right)_{ij} \right]^2  = Q^{(2)}_{y,ij,ij}=
      \Lambda^{(1)}_{y,ij}   \left( \sum_{j' \in \mD_i} \Lambda^{(1)}_{y,ij'}\right)^{-1} $
      and 
      $\Ep    \left[ \left( Q^{(3)}_y s_y  \right)_{ij} \right]^2  = Q^{(3)}_{y,ij,ij}=
      \Lambda^{(1)}_{y,ij}   \left( \sum_{i' \in \mD_j} \Lambda^{(1)}_{y,i'j}\right)^{-1} $
      we find that 
      \begin{align*}
           \Ep   C^{(1,\beta)}_y &=  B^{(\beta)}_y,
           &
            \Ep   C^{(2,\beta)}_y &=  D^{(\beta)}_y.
      \end{align*}
      Furthermore, using 
      the expressions for $\left( Q^{(2)}_y s_y  \right)_{ij}$ and $\left( Q^{(3)}_y s_y  \right)_{ij}$
      in  \eqref{ExpressQ23s} above we can write,  for given $\ell \in \{1,\ldots, d_x\}$,
      \begin{align*}
             \left[ W_y\,  C^{(1,\beta)}_y \right]_\ell
             &= - \frac 1 2 \, \frac 1 I  \sum_{i=1}^I  
                c^{(n)}_{y,i}  
                  \left( \frac 1 {\sqrt{|\mD_i|}} \sum_{j \in \mD_i}   \partial_{\pi} \ell_{y,ij} \right)^2 ,
            \\      
            \left[ W_y\,  C^{(2,\beta)}_y \right]_\ell  
            &= - \frac 1 2 \,  \frac 1 J  \sum_{j=1}^J  
                d^{(n)}_{y,j}  
                  \left( \frac 1 {\sqrt{|\mD_j|}} \sum_{i \in \mD_j}   \partial_{\pi} \ell_{y,ij} \right)^2 ,
      \end{align*}
      where
      \begin{align*}
           c^{(n)}_{y,i}  &=  \frac{ |\mD_i|^{-1} \sum_{j \in \mD_i}  \widetilde x_{y,ij}          \Lambda^{(2)}_{y,ij}  } 
           {\left( |\mD_i|^{-1} \sum_{j \in \mD_i} \Lambda^{(1)}_{y,ij} \right)^2}  ,
           &
           d^{(n)}_{y,j}  &=\frac{ |\mD_j|^{-1} \sum_{i \in \mD_j}  \widetilde x_{y,ij}          \Lambda^{(2)}_{y,ij}  } 
           {\left( |\mD_j|^{-1} \sum_{i \in \mD_j} \Lambda^{(1)}_{y,ij} \right)^2} ,
      \end{align*}
      which are of order $O_P(1)$, uniformly over $y$ and $i$ and $j$.      
      By employing part $(ii)$ of Corollary~\ref{cor:Score2} we thus find that
      \begin{align*}
            C^{(1,\beta)}_y -  \Ep   C^{(1,\beta)}_y &= o_P(1) ,
            &
            C^{(2,\beta)}_y -  \Ep   C^{(2,\beta)}_y &= o_P(1).
      \end{align*}
      Finally, again using \eqref{ExpressQ23s} we can write
      \begin{align*}
            C^{(3,\beta)}_y
            &=  -  
        W_y^{-1}  \frac 1 {\sqrt{n}}      \sum_{(i,j) \in \mD}   
          \frac{ \widetilde x_{y,ij}                   
      \Lambda^{(2)}_{y,ij} }
      {   |\mD_i|^{1/2}  |\mD_j|^{1/2}  }
         \frac{  \left( |\mD_i|^{-1/2}  \sum_{j' \in \mD_i} \partial_{\pi} \ell_{y,ij'} \right)  \left(   |\mD_j|^{-1/2}  \sum_{i' \in \mD_j} \partial_{\pi} \ell_{y,i'j} \right) } {\left( |\mD_i|^{-1} \sum_{j' \in \mD_i} \Lambda^{(1)}_{y,ij'} \right)
            \left(   |\mD_j|^{-1} \sum_{i' \in \mD_j} \Lambda^{(1)}_{y,i'j} \right)}  ,
      \end{align*}
      and therefore      
      \begin{align*}
            \left\| C^{(3,\beta)}_y \right\|
            &\leq   
            \max_{j \in {\bf J}}  \left|  |\mD_j|^{-1/2}  \sum_{i \in \mD_j} \partial_{\pi} \ell_{y,ij}   \right|
            \\ & \quad
        \times \left\|    
        W_y^{-1}  \frac 1 {\sqrt{n}}      \sum_{(i,j) \in \mD}   
          \frac{ \widetilde x_{y,ij}                   
      \Lambda^{(2)}_{y,ij} }
      {   |\mD_i|^{1/2}  |\mD_j|^{1/2}  }
         \frac{  \left( |\mD_i|^{-1/2}  \sum_{j' \in \mD_i} \partial_{\pi} \ell_{y,ij'} \right)   } {\left( |\mD_i|^{-1} \sum_{j' \in \mD_i} \Lambda^{(1)}_{y,ij'} \right)
        \left(  |\mD_j|^{-1} \sum_{i' \in \mD_j} \Lambda^{(1)}_{y,i'j} \right)  }  \right\|
       \\
          &=  \left(   \max_{i \in I}  |\mD_i|^{-1/2} \right)
           \left( \max_{j \in {\bf J}}  \left|  |\mD_j|^{-1/2}  \sum_{i \in \mD_j} \partial_{\pi} \ell_{y,ij}   \right| \right)
             \left\|   \frac 1 {\sqrt{n}} \sum_{(i,j) \in \mD}   
                   e^{(n)}_{y,i} \; \partial_{\pi} \ell_{y,ij}
             \right\|
        ,
      \end{align*} 
      where
      \begin{align*}
             e^{(n)}_{y,i}
             &=  W_y^{-1} 
              |\mD_i|^{-1/2}
           \sum_{j \in \mD_i}     |\mD_j|^{-1/2} 
         \frac{ \widetilde x_{y,ij}                   
      \Lambda^{(2)}_{y,ij}     } {\left( |\mD_i|^{-1} \sum_{j' \in \mD_i} \Lambda^{(1)}_{y,ij'} \right)
        \left(  |\mD_j|^{-1} \sum_{i' \in \mD_j} \Lambda^{(1)}_{y,i'j} \right)  }  ,
      \end{align*}
      which is of order one, uniformly over $y$ and $i$.
      Thus, by applying Corollary~\ref{cor:ScoreFE},
      and Corollary~\ref{cor:Score2}
      with $b^{(n)}_{y,hij}$ equal to the elements of the $d_x$-vector $e^{(n)}_{y,i}$ (i.e. no $j$-dependence), 
      we find
      that $  \left\| C^{(3,\beta)}_y \right\| = O_P(n^{-1/4}) o_P\left(  n^{1/12}  \right) O_P \left(  n^{1/6}  \right) = o_P(1)$.
 Combining the above we conclude
 $$ - \frac 1 2 
        W_y^{-1}  \;   n^{-1/2}       \sum_{(i,j) \in \mD}   \; 
           \widetilde x_{y,ij}          
           \left( \Lambda^{(1)}_{y,ij} \right)^{-1}
      \Lambda^{(2)}_{y,ij} 
        \left[ \left( Q_y s_y  \right)_{ij} \right]^2
      -\left(  n^{-1/2}  I    B^{(\beta)}_y
                   +  n^{-1/2}  J      D^{(\beta)}_y \right)   \rightarrow_P  0 .
        $$
      The proof for  
        $$  \frac  1 {2 \sqrt{n}} \,  \sum_{(i,j) \in \mD}  
         \left( \Lambda^{(1)}_{y,ij} \right)^{-1}
         \left( \Lambda^{(2)}_{y,ij,k} - \Lambda^{(2)}_{y,ij} \Psi_{y,ij,k} \right) \left[ \left( Q_y s_y  \right)_{ij}  \right]^2   
       - \left( n^{-1/2}  I    B^{(\Lambda)}_{y,k}
                   +  n^{-1/2}  J      D^{(\Lambda)}_{y,k} \right)  \rightarrow_P   0
        $$ 
     is analogous.
\end{proof}

\section*{\bf Proof of Lemma~\ref{lemma:bias_estimators}}

\begin{proof}[\bf Proof of Lemma~\ref{lemma:bias_estimators}] Let
   \begin{align*}
           \widetilde x_{y}(\pi_y)
         &=  x_{y}  -   
    \left[w^{(2)},w^{(3)} \right]  \left( \left[w^{(2)},w^{(3)} \right]' \Lambda^{(1)}(\pi_y) \left[w^{(2)},w^{(3)} \right] \right)^\dagger \left[w^{(2)},w^{(3)} \right]'  \left( \Lambda^{(1)}(\pi_y) \right)    x_{y}   ,
   \end{align*}
   and
   \begin{align*}
       W_y(\pi_y) = \frac 1 {n} \sum_{(i,j) \in \mD} \Lambda^{(1)}(\pi_{y,ij}) \, \widetilde x_{y,ij}(\pi_y) \, \widetilde x_{y,ij}'(\pi_y) ,
   \end{align*}
   and\footnote{
   Note that instead of $B_y^{(\beta)}(\pi_y) $
   we could simply write $B^{(\beta)}(\pi_y) $ here, because all the dependence on $y$
   is through the parameter $\pi_y$.
   The only reason to write $B_y^{(\beta)}(\pi_y) $ is to avoid confusion
   with the notation $B^{(\beta)}(y)$ in the main text.
   }
    \begin{align*}
         B_y^{(\beta)}(\pi_y)  =  -  \frac 1 2    W_y(\pi_y)^{-1}  \left[ \frac 1 {I}  \sum_{i=1}^I  
               \frac{J^{-1} \sum_{j \in \mD_i}    \Lambda^{(2)}(\pi_{y,ij}) \,  \widetilde x_{y,ij}(\pi_y)  }
                      {J^{-1} \sum_{j \in \mD_i}    \Lambda^{(1)}(\pi_{y,ij}) } \right] .
    \end{align*}
    Then we can write $ B^{(\beta)}_y =   B_y^{(\beta)}(\pi^0_y)  $
    and   $ \widehat B^{(\beta)}_y  = B_y^{(\beta)}(\widehat \pi_y)  $.
    The consistency result for $\widehat B^{(\beta)}_y= \widehat B^{(\beta)}(y)$  follows from an expansion of $= B_y^{(\beta)}(\widehat \pi_y) $ 
    in $\widehat \pi_y$ around $\pi^0_y$.
    $\Lambda^{(1)}(\pi_{y,ij}) $
     and $\Lambda^{(2)}(\pi_{y,ij}) $
     and $\left( \Lambda^{(1)}(\pi_{y,ij}) \right)^{-1}$ are all uniformly bounded
     over $\pi_{y,ij} \in [\pi_{\min},\pi_{\max}]$,
     for any bounded interval $[\pi_{\min},\pi_{\max}]$.
    Using this one obtains
     \begin{align*}
          b_n
          := \sup_{y \in  \mY} \max_{(i,j) \in \mD} 
          \sup_{\pi_y \in [\pi_{\min},\pi_{\max}]^n}
          \left\|
           \frac{ \partial B_y^{(\beta)}(\pi_y)   }
          {\partial \pi_{y,ij}} \right\|     
          = O_P(n^{-1}) ,
     \end{align*}
     because any individual $\pi_{y,ij}$ only enters via an appropriately normalized sample average into $B_y^{(\beta)}(\pi_y) $.  Indeed for any function of the form
     $$
      B(\pi_y)  =  \frac 1 {I}  \sum_{i=1}^I  
               \frac{J^{-1} \sum_{j \in \mD_i}    f_1(\pi_{y,ij})}
                      {J^{-1} \sum_{j \in \mD_i}    f_2(\pi_{y,ij}) }  ,
     $$
     where $f_1$ and $f_2$ are differentiable with bounded derivatives $f_1'$ and $f_2'$,
     $$
     \frac{\partial B(\pi_y)}{\partial \pi_{y,ij}} =  \frac 1 {IJ} \frac{ f_1'(\pi_{y,ij})J^{-1} \sum_{j' \in \mD_i}    f_2(\pi_{y,ij'}) - J^{-1} \sum_{j' \in \mD_i}    f_1(\pi_{y,ij'})f_2'(\pi_{y,ij})}
                      {\left[J^{-1} \sum_{j' \in \mD_i}    f_2(\pi_{y,ij'})\right]^2 } = O_P(n^{-1}).
     $$
         
      Lemma~\ref{lemma:ScoreExpansionPi} together with 
    Lemma~\ref{lemma:PropQ}$(ii)$
     and Lemma~\ref{lemma:ScoreAve}$(i)$ guarantee that 
     \begin{align*}
          \sup_{y \in  \mY} \max_{(i,j) \in \mD}  \left|  \widehat \pi_{y,ij} - \pi^0_{y,ij} \right| = o_P(n^{-1/6}) = o_P(1).
     \end{align*}
     By a mean value of expansion in $\widehat \pi_y$ around $\pi^0_y$ we thus obtain
     \begin{align*}
           \sup_{y \in  \mY} 
           \left\| B_y^{(\beta)}(\widehat \pi_y) - B_y^{(\beta)}(\pi^0_y) \right\|
          \leq   b_n  \sum_{(i,j) \in \mD}   \left|  \widehat \pi_{y,ij} - \pi^0_{y,ij} \right| = o_P(1).
     \end{align*}
     The proof of consistency for the other  estimators in Lemma~\ref{lemma:bias_estimators} is analogous.
     \end{proof}

\end{document}